\documentclass[conference]{IEEEtran}
\IEEEoverridecommandlockouts
\def\BibTeX{{\rm B\kern-.05em{\sc i\kern-.025em b}\kern-.08em
    T\kern-.1667em\lower.7ex\hbox{E}\kern-.125emX}}

\usepackage[full]{optional}
\opt{narrowmargin}{
\usepackage{geometry} 
\newgeometry{top=0.75in, bottom=1in, left=0.625in, right=0.625in, twocolumn}
}
\usepackage{fancyhdr}
\opt{short}{\newcommand{\revise}[1]{\textcolor{blue}{#1}}} 
\opt{full}{\newcommand{\revise}[1]{{#1}}} 

\usepackage[utf8]{inputenc}
\usepackage{graphics}
\usepackage{blindtext, graphicx}
\usepackage{xcolor}
\usepackage{amsfonts,amssymb}
\usepackage{amsthm,amsmath}
\usepackage{xfrac}
\usepackage[spaces,hyphens]{xurl}
\usepackage[hidelinks]{hyperref}
\usepackage{makecell}
\usepackage{subfigure}
\usepackage{scalefnt}
\newtheorem{lemma}{Lemma}
\newtheorem{definition}{Definition}

\newtheorem{theorem}{Theorem}
\theoremstyle{definition}
\newtheorem{example}{Example}

\newenvironment{proofsk}{\proof[Proof Sketch]}{\endproof}

\newcommand{\argmax}{\operatornamewithlimits{argmax}}
\usepackage[ruled,vlined,linesnumbered,procnumbered]{algorithm2e}

\newcommand\xqed[1]{%
  \leavevmode\unskip\penalty9999 \hbox{}\nobreak\hfill
  \quad\hbox{#1}}
\newcommand\exend{\xqed{$\blacksquare$}}

\SetCommentSty{mycommfont}

\usepackage{tikz}
\usetikzlibrary{arrows}
\usetikzlibrary{positioning}
\usetikzlibrary{shapes.geometric}
\tikzset{
    >=stealth',
    shorten >=1pt,
    auto,
    thick,
    item node/.style={rectangle,draw,align=center,font=\ttfamily\small},   
    company node/.style={ellipse,draw,align=center,font=\ttfamily\small},
    event node/.style={trapezium,trapezium left angle=60,trapezium right angle=-60,draw, align=center,font=\ttfamily\small},
    node distance=1.4cm,
    edge/.style={->,font=\scshape\small},
}

\newcommand{\arwen}[1] 
{
{\textcolor{orange}{#1}}
}

\newcommand{\q}[1] 
{
{\textcolor{red}{#1}}
}

\newcommand{\hana}[1] 
{
{\textcolor{purple}{#1}}
}

\newcommand{\fullproblem}{Influence Maximization based on Dynamic Personal Perception}
\newcommand{\problem}{IMDPP}
\newcommand{\fullsproblem}{Simple IMDPP}
\newcommand{\sproblem}{SIMDPP}
\newcommand{\fullalgo}{\underline{Dy}namic perception for \underline{s}eeding \underline{i}n target \underline{m}arkets}
\newcommand{\algo}{Dysim}
\newcommand{\fullsalgo}{Simple Dysim}
\newcommand{\salgo}{SDysim}

\begin{document}
\setlength{\textfloatsep}{5pt}
\setlength{\floatsep}{0pt}

\title{Influence Maximization Based on Dynamic Personal Perception in Knowledge Graph}

\makeatletter
\newcommand{\linebreakand}{%
  \end{@IEEEauthorhalign}
  \hfill\mbox{}\par
  \mbox{}\hfill\begin{@IEEEauthorhalign}
}
\makeatother

\author{
\IEEEauthorblockN{Ya-Wen Teng}
\IEEEauthorblockA{\textit{Academia Sinica, Taiwan} \\
ywteng@citi.sinica.edu.tw}
\and
\IEEEauthorblockN{Yishuo Shi}
\IEEEauthorblockA{\textit{Wenzhou University, China} \\
yishuo@wzu.edu.cn}
\and
\IEEEauthorblockN{Chih-Hua Tai}
\IEEEauthorblockA{\textit{National Taipei University, Taiwan} \\
hanatai@mail.ntpu.edu.tw} 
\linebreakand
\IEEEauthorblockN{De-Nian Yang}
\IEEEauthorblockA{\textit{Academia Sinica, Taiwan} \\
dnyang@iis.sinica.edu.tw}
\and
\IEEEauthorblockN{Wang-Chien Lee}
\IEEEauthorblockA{\textit{Penn State University, U.S.A.} \\
wlee@cse.psu.edu}
\and
\IEEEauthorblockN{Ming-Syan Chen}
\IEEEauthorblockA{\textit{National Taiwan University, Taiwan} \\
mschen@ntu.edu.tw}
}
\maketitle
\begin{abstract}
Viral marketing on social networks, also known as \textit{Influence Maximization (IM)}, aims to select $k$ users for the promotion of a target item by maximizing the total spread of their influence. However, most previous works on IM do not explore the dynamic user perception of promoted items in the process. In this paper, by exploiting the knowledge graph (KG) to capture dynamic user perception, we formulate the problem of \textit{\fullproblem\ (\problem)} that considers user preferences and social influence reflecting the impact of relevant item adoptions. We prove the hardness of \problem\ and design an approximation algorithm, named \textit{\fullalgo\ (\algo)}, by exploring the concepts of dynamic reachability, target markets, and substantial influence to select and promote a sequence of relevant items. We evaluate the performance of \algo\ in comparison with the state-of-the-art approaches using real social networks with real KGs. The experimental results show that \algo\ effectively achieves up to 6.7 times of influence spread in large datasets over the state-of-the-art approaches.
\end{abstract}



\section{Introduction}
\label{sec:intro}

Social influence \cite{Kempe2003,han2018discount,Nguyen2016Cost} refers to the impact of a social environment on people's behavior. By exploiting the social influence of users, a wide spectrum of applications (e.g., item promotion and viral marketing) have been formulated as various research problems, such as \textit{influence maximization} (IM) \cite{Kempe2003}, revenue maximization (RM) \cite{han2018discount}, and profit maximization (PM) \cite{Nguyen2016Cost}. Among them, IM selects $k$ users as the seeds to promote \textit{one} target item to maximize the number of influenced users. Nevertheless, in real life, companies often promote relevant items in \textit{multiple} events, e.g., Apple Inc.\ usually promotes iPhones, AirPods, and iPads in September, followed by a series of subsequent promotions.\footnote{\url{https://www.apple.com/apple-events/}.} In this work, we address a new IM problem formulated for a sequence of promotions on relevant items.\footnote{After the influence propagation of the seed group for the first promotion finishes, the second follows, and so on.} 

\revise{\label{para:support}For multiple promotions, exploring the dynamic changes in personal perceptions on promoted items is important, since users' perceptions of item relationships may vary according to the changes in users' demand indicated by research in the marketing field \cite{shocker2004product}.} First, the \textit{complementary} and \textit{substitutable} relationships between items affect users' preferences on items \revise{\cite{shocker2004product,venkatesh2003optimal,gerlach2014never}}. In economics, \textit{cross elasticity of demand} \cite{frank1991microeconomics} indicates that adopting complementary items of an item increases the preference for it, while adopting its substitutable items has the opposite effect. For example, users who own iPhones with no headphone jack may be interested in AirPods (due to its complementary relationship with iPhones),\footnote{A real example is in \url{https://amzn.to/2Zl23oT}.} while users who have iPhones may have less interest in iPads (due to their substitutable relationship). Second, the association between items may trigger extra adoptions without promotions \revise{\cite{zhang2020complements,koukova2012multiformat}}. For example, AirPods may be directly adopted together with iPhones due to their complementary relationship. 

Third, the perceptions of these relationships between items are usually personal and dynamic \revise{\cite{shocker2004product,shi2019semrec,gu2019relevance}}, as the items got newly adopted usually bring fresh experiences to users. For example, users who care more about large screens than mobility may treat iPhones as substitutable items of iPads; when these iPad users start to care about the mobility, they may tend to regard iPhones as complementary items of iPads. In turn, the changes in personal perceptions of item relationships lead to changes in users' preferences. Fourth, the dynamic personal perceptions of item relationships also affect users' social influence strength over friends, since friends adopting similar items and sharing similar perceptions tend to become closer \revise{\opt{short}{\cite{ma2015latent,yavacs2014impact}}\opt{full}{\cite{ma2015latent,yavacs2014impact,gu2017co,farajtabar2015coevolve}}}. \revise{To address IM in a sequence of promotions on relevant items, it is essential to carefully examine dynamic personal perceptions of item relationships, together with their ripple effect on personal preferences for items, social influence strength, and item associations.}

Knowledge graph (KG) (along with weighted meta-graphs) to capture the relationships (e.g., the complementary and substitutable relationships) has been well-explored in recommendation systems \opt{short}{\cite{shi2019semrec}}\opt{full}{\cite{zhao2017improving,shi2019semrec,gu2019relevance,Huang2016Meta,zhang2016collaborative}}. As illustrated in Fig.~\ref{Fig:first_examp}, KG represents facts (e.g., \textsc{item} \textsf{iPhone} and \textsc{item} \textsf{AirPods} \textsc{support} the \textsc{feature} \textsf{Bluetooth} in Fig.~\ref{Fig:new_ex_kg}), while meta-graphs capture relationships in the KG (e.g., $m_1$ in Fig.~\ref{Fig:new_ex_meta} describes two \textsc{item}s \textsc{support}ing the same \textsc{feature} are complementary). Note that these meta-graphs can be used to reflect the perception of item relationships, in forms of \textit{personal item network}, for each individual. The personal weighting on each meta-graph describes the significance of this meta-graph to an individual (e.g., the values next to $m_1,\ldots,m_3$ in Fig.~\ref{Fig:new_ex_network1}), while the relevance scores between items describe the strength of their relationships in the mind of this individual \cite{shi2019semrec,gu2019relevance} (e.g., the values on dotted edges in Fig.~\ref{Fig:new_ex_network1}). By adjusting the weightings on meta-graphs according to previous adoptions \cite{shi2019semrec,gu2019relevance}, dynamic personal perceptions of item relationships in individual users can be updated (in Fig.~\ref{Fig:new_ex_network2}). In this paper, we aim to leverage dynamic personal item networks for a sequence of IM promotions.

\begin{figure}
    \centering
    \subfigure[]{
        \centering
        \includegraphics[width=0.235\textwidth]{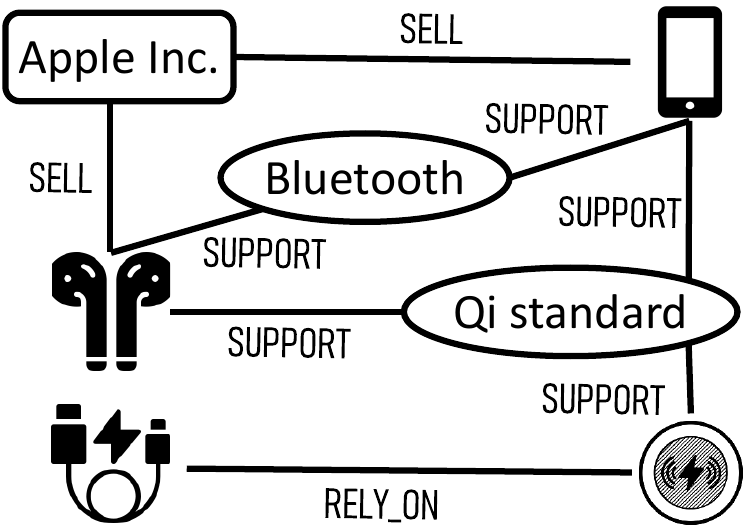}
        \label{Fig:new_ex_kg}
    }%
    \subfigure[]{
        \centering
        \includegraphics[width=0.235\textwidth]{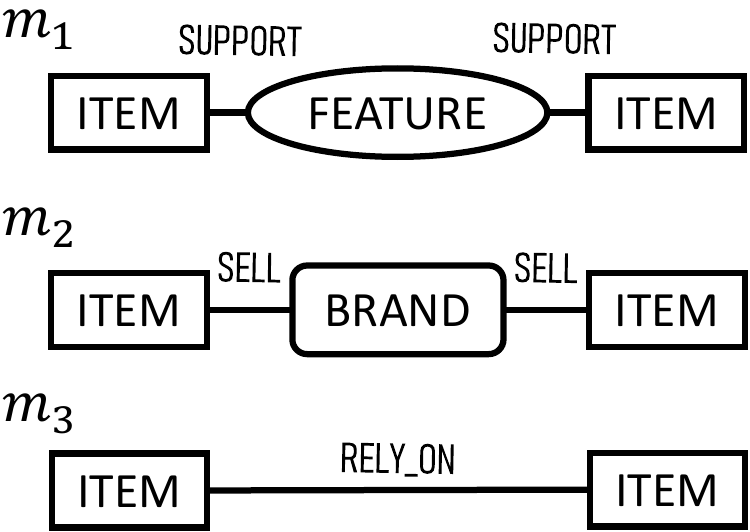}
        \label{Fig:new_ex_meta}
    }
    \subfigure[]{
        \centering
        \includegraphics[width=0.235\textwidth]{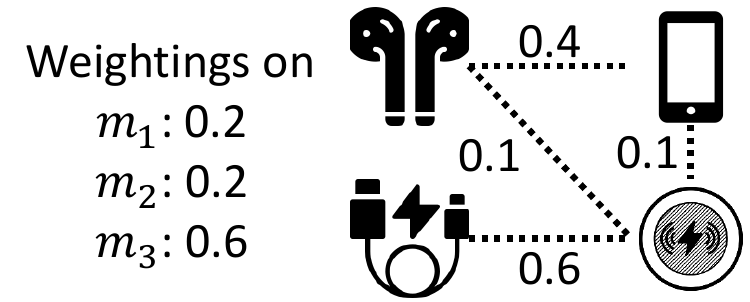}
        \label{Fig:new_ex_network1}
    }%
    \subfigure[]{
        \centering
        \includegraphics[width=0.235\textwidth]{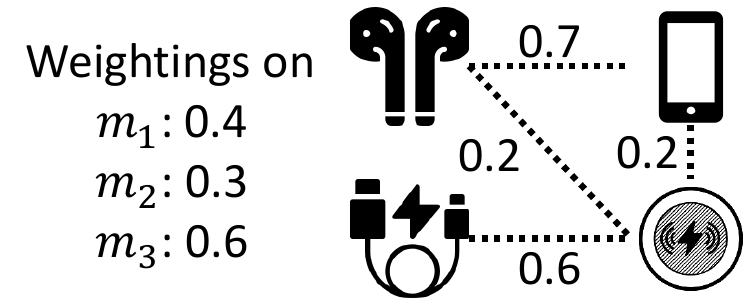}
        \label{Fig:new_ex_network2}
    }
    \caption{(a)~A tiny KG describing facts about the iPhone, AirPods, wireless charger, and charging cable. 
    (b)~Three meta-graphs specifying the complementary relationship. (c)~Bob's initial personal item network, where a dotted edge denotes a complementary relationship. 
    (d)~Update of Bob's personal item network: after adopting iPhone and AirPods, Bob's weightings on $m_1$ and $m_2$ grow, which increases the relevance scores between iPhone, AirPods, and the wireless charger.}
    \label{Fig:first_examp}
\end{figure}

Following up the example in Fig.~\ref{Fig:first_examp}, Fig.~\ref{Fig:ex} illustrates the IM process considering dynamic personal perceptions of item relationships, personal preferences for items, social influence strength, and item associations. As shown, the number of hearts indicates Bob's preference for a not-yet-adopted item, and a solid arrow represents the social influence between users (thickness implies strength). After Bob is promoted iPhone by Alice, Bob's purchase decision depends not only on the influence strength from Alice but also on his own preference for iPhone (in Fig.~\ref{Fig:ex_a}). Meanwhile, item associations usually trigger extra adoptions of relevant items, such as AirPods, according to Bob's item network (in Fig.~\ref{Fig:new_ex_network1}). After Bob purchases iPhone and AirPods, his perception of the complementary relationship changes (i.e., he becomes to regard items supporting common features or belonging to the same brand as complementary), which increases the relevance between iPhone, AirPods, and the wireless charger (in Fig.~\ref{Fig:new_ex_network2}). After that, as Bob has adopted iPhone and AirPods, and their relevance to the wireless charger increases, Bob's preference for the wireless charger grows accordingly.\footnote{A real example is in \url{https://amzn.to/3fW7JLC}.} Moreover, if Cindy acts as a seed to promote the wireless charger, as Bob and Cindy have similar adopted items (in Fig.~\ref{Fig:ex_b}) (indicating Bob shares a similar perception of item relationships with Cindy and tends to behave similarly with Cindy), the influence strength from Cindy to Bob thus becomes stronger. It is easier for Cindy to promote the wireless charger to Bob now, since both Cindy's influence strength to Bob and Bob's preference for the wireless charger increase.

To incorporate factors depicted in the example above, several new challenges arise.
(i)~\textit{Propagation of item impact (i.e., impact due to item adoption)}: Item adoptions change users' personal perceptions of item relationships, their preferences for other items, their strength of social influence among friends, and the item associations. In other words, the promotion of an item may affect the adoptions of subsequent items and thereby the planning for the next promotions. The order of items being promoted matters. 
(ii)~\textit{Antagonism of the substitutable relationship}: Promoting an item after adopting a substitutable item is not beneficial when the first item has met the users' needs. It is thus vital to avoid promoting substitutable items to the same users in consecutive promotions.
(iii)~\textit{Determination of promotional timing}: As the promotions are dependent on previous ones, a seed in early promotions should facilitate subsequent promotions, while a seed in later promotions should focus on potential adoptions benefited from previous promotions. Therefore, determining the proper promotional timing for each seed is essential.

\begin{figure}
    \centering
    \subfigure[]{
        \centering
        \includegraphics[width=0.235\textwidth]{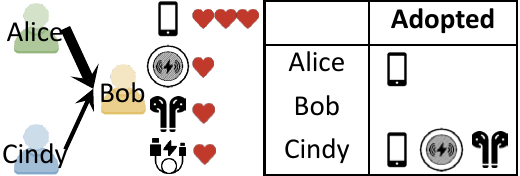}
        \label{Fig:ex_a}
    }%
    \centering
    \subfigure[]{
        \centering
        \includegraphics[width=0.235\textwidth]{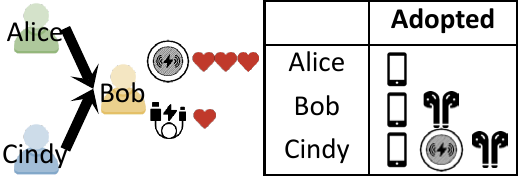}
        \label{Fig:ex_b}
    }
    \caption{Illustration of the \problem\ problem. (a) The states before Bob adopts iPhone and AirPods. (b) The states after Bob adopts iPhone and AirPods.}
    \label{Fig:ex}
\end{figure}

In this work, we formulate a new problem, named \textit{\fullproblem} (referred to as \problem). In contrast to most previous works \opt{short}{\cite{huang2020competitive}}\opt{full}{\cite{huang2020competitive,lu2015competition,lu2013bang}} focusing on one item, given the social network, KG, and meta-graphs for different item relationships, \problem\ targets on multiple promotions to maximize the overall spread of influence by choosing items and selecting seed users for promotion at proper timings under a total budget, where users have different costs as seeds \cite{Nguyen2016Cost}, and each promotion allows multiple items to be promoted. We exploit personal item networks to capture dynamic personal perceptions of complementary and substitutable relationships between items. The adoptions of items dynamically adjust users' weightings on meta-graphs, reflecting dynamic personal perceptions and updating personal item networks. Also, users' preferences for other items, their social influence strength over friends, and the item associations change accordingly, in turn affecting other users' adoptions, their personal weightings on meta-graphs, their preferences for other items, their social influence strength, and the item associations in a ripple effect.
\opt{short}{\revise{For the ease of understanding, we first present the fundamental problem of \problem, referred to as \textit{\fullsproblem\ (\sproblem)}, by focusing on the important factors of dynamic personal perceptions of item relationships and their fundamental ripple effect on dynamic preferences for items, i.e., neglecting dynamic social influence strength and item associations.}}

We prove that \revise{\opt{short}{\sproblem\ and \problem\ are}\opt{full}{\problem\ is}} NP-hard and inapproximable within $O(\frac{1}{\left|V\right|^{1-\epsilon}})$, where $\left|V\right|$ is the number of users and $\epsilon$ is an arbitrarily small constant. We design an approximation algorithm, named \textit{\fullalgo\ (\algo)}, to tackle the above challenges of \problem. For the first challenge in the propagation of item impacts, \algo\ introduces \textit{dynamic reachability} to evaluate the impacts from previously promoted items on the currently chosen item, as well as the potential impact from the current item on any candidate item in subsequent promotions. For the second challenge in the antagonism of the substitutable relationship, \algo\ identifies \textit{target markets} to promote complementary items to socially close users in consecutive promotions. For the third challenge in determining the promotional timing, \algo\ introduces \textit{substantial influence} to evaluate both immediate and subsequent adoptions under the impact of a candidate seed (assigned at some promotional timing). 
\revise{\opt{short}{For \sproblem, where dynamic social influence strength and item associations are not considered, we develop an approximation algorithm, namely \textit{\fullsalgo\ (\salgo)}, in this paper as well.}
We evaluate the performance of \opt{short}{\salgo\ and }\algo\ on real social networks with KGs, i.e., Amazon, Yelp, Douban, and Gowalla. 
\opt{short}{Due to the space constraint, the details of \problem\ and \algo\ are presented in the full-length version \cite{online} of this paper.}} 
The contributions of this work include:

\begin{itemize}
\item \revise{To the best of our knowledge, \problem\ is the first attempt to study the IM problem under a sequence of promotions on relevant items, where the personal perceptions of item relationships are dynamically captured from users' previously adopted items by KG and meta-graphs, and the changes in preferences for items, social influence strength, and item associations are considered as a ripple effect in the diffusion process.}

\item We prove that \revise{\opt{short}{\sproblem\ and \problem\ are}\opt{full}{\problem\ is}} inapproximable within $O(\frac{1}{\left|V\right|^{1-\epsilon}})$ even for a simple case with only the complementary relationship and only one promotion. 

\item \revise{We design an approximation algorithm \algo, which plans a distinct effective promotional strategy for each target market to avoid antagonism between substitutable items. \algo\ carefully examines the dynamic reachability of items to prioritize the promotion of relevant items, and evaluates the substantial influence of candidate seeds to properly determine the promotional timing.}

\item Via real social networks and real KGs, experimental results demonstrate that \revise{\opt{short}{\salgo\ and }\algo\ effectively achieve up to \opt{short}{10.95 times and }6.7 times\opt{short}{, respectively,} of the influence spread over the state-of-the-arts.}
\end{itemize}

\section{Related work}
\label{sec:related_work}

Influence maximization (IM) aims at maximizing the number of influenced users by selecting seed users. It was first formulated as a discrete optimization problem and proved as NP-hard by Kempe et al.\ \cite{Kempe2003}.
\opt{full}{In the work \cite{Kempe2003}, they prove the influence maximization is NP-hard under the triggering models, i.e., the Linear Threshold (LT) and Independent Cascade (IC) diffusion models. The LT model considers a user's threshold and the weighted fraction of her all influenced friends as the criterion to be influenced while the IC model treats every influence independently with a probability. An influenced user can then propagate influence to her friends, and the influence thus diffuses until no user is newly influenced. }Since then, various issues in IM have been actively studied. To address the inefficiency in computing influence spread, some exploit the submodular property and certain heuristics \opt{short}{\cite{chen2010scalable}}\opt{full}{\cite{goyal2011celf++,chen2010scalable}}. Recent works further introduce the reverse influence sampling to approximate the influence with guarantees \opt{short}{\cite{guo2020influence}}\opt{full}{\cite{guo2020influence,Borgs2014Maximizing}}.\opt{short}{\footnote{More introduction on diffusion models and other IM problems is presented in \cite{online}.}} \opt{full}{In addition, extensive research efforts study the maximization of influence on target users \cite{Nguyen2016Cost,guo2013personalized}, at specified locations \cite{song2016targeted,wang2016efficient}, for specific contents \cite{chen2015online,li2015real}, or at specific time \cite{zhuang2013influence,wang2017real}. IM for target users \cite{Nguyen2016Cost,guo2013personalized} focuses on a subset of users and even allows users to contribute different benefits to the company. The location-based IM \cite{song2016targeted,wang2016efficient} takes locations of users and the promoted event into account. As the geographical distance is an important constraint for users, the probability to influence a user is affected by her distance to the promoted event's location.
Different from \cite{wang2016efficient} focusing on the IC model, a variant of the IC model is employed in \cite{song2016targeted}, which considers the login event so that the influence can be propagated to a friend only after she is online with a login probability. In \cite{wang2014novel,wang2015modeling}, the reachability-based diffusion process is regarded as generating a branching tree from the seed to be the egocentric influence rings of the root node (i.e., the seed). The influence is assumed to decay along the path, and it is less likely to reach the nodes that are farther away from the root.}
\revise{Recently, Huang et al.\ \cite{huang2020competitive} point out that users' adopting probabilities of the promoted item should depend on users' previously adopted complementary and substitutable items (which is modeled as dynamic preferences for items in \problem). However, \cite{huang2020competitive} targets only on a specified item in only one promotion with fixed item relationships, whereas \problem\ explores multiple promotions on relevant items and carefully examines the dynamic user perceptions of item relationships. Although various issues, e.g., target audience, scalability, and complementary/substitutable items, are studied, previous works \cite{Kempe2003,Nguyen2016Cost,huang2020competitive,chen2010scalable,guo2020influence} promote only one target item in only a single promotion, instead of multiple target items in multiple promotions, and the phenomenon of dynamic personal perceptions of item relationships together with its ripple effect are not considered.
By contrast, \problem\ aims at a sequence of promotions on \textit{relevant items} modeled by Knowledge Graph, where the item relationships and promotional timings are able to alter users adoption decisions.}

\revise{Some studies investigate IM on promoting multiple target items}, e.g., making exclusive adoption among items \cite{teng2018revenue}, avoiding spamming seeds by overwhelming promotions \cite{datta2010viral}, learning diffusion probabilities of different items \cite{hung2016social}, and maximizing utility-based adoption among desired items \cite{banerjee2019maximizing}. However, they focus on a single promotion and do not consider multiple promotions to promote a sequence of relevant items modeled by KG and meta-graphs. Moreover, they study the problems under simpler scenarios without capturing the dynamic changes in personal perceptions of item relationships, personal preferences for items, social influence strength, and item associations. \revise{By contrast, in \problem, the adoption of items dynamically changes the personal perceptions of complementary and substitutable relationships between items. The changed perceptions of item relationships affect users' preferences for items and users' social influence strength over time, in turn affecting other users' adoptions, their preferences for other items, and their social influence strength as a ripple effect. Therefore, the above works have limitations to \problem, since the promotional timing is critical as users' perceptions of item relationships, preferences for items, and influence strength on friends are dynamic.}

\revise{Research on adaptive IM \cite{sun2018multi,han2018efficient,peng2019adaptive} aims to select the seeds adaptively based on the adoptions in the previous influence diffusion. However, although multiple promotional timings are considered, they consider only one item in the IM problem and ignore multiple target items, item relationships, and dynamic preference for items. Moreover, adaptive IM requires a predefined budget allocation to different promotions, and it does not have the adaptive monotonicity and the adaptive submodularity (or even the adaptive bounded weak-submodularity).\footnote{\revise{The adaptive monotonicity is a property that conditional expected marginal adoptions of any item are non-negative. The adaptive submodularity (or the adaptive bounded weak-submodularity) is the property that conditional expected marginal adoptions of any fixed item do not increase (or boundedly increase) as more items are selected and their states are observed.}}
By contrast, \problem\ does not require a predefined budget allocation to promotions and can be solved by \algo\ with an approximation guarantee (detailed in \opt{short}{\cite{online}}\opt{full}{Sec.~\ref{sec:theoretical}}). }

Knowledge graph (KG) is employed to describe facts in a wide spectrum of applications, e.g., relevance measures and search \cite{gu2019relevance,Huang2016Meta}, and recommendation \opt{short}{\cite{shi2019semrec}}\opt{full}{\cite{shi2019semrec,zhang2016collaborative}}. Shi et al.\ \cite{shi2019semrec} present a new similarity measure through personal weighted meta-paths to include different semantics of similarity. Users' own preferences can thus be derived from these meta-paths. Gu et al.\ \cite{gu2019relevance} point out that a user may have different perceptions of similarity due to the change in her interests. They propose to automatically pick up meta-paths to best characterize the similarity by user-provided examples. Huang et al.\ \cite{Huang2016Meta} further extend meta-paths to meta-graphs to measure similarity with more complex connections. Note that these works focus on predicting the ratings of unknown items for users, which is essentially different from the IM problem. Inspired by the above research, we first attempt to incorporate the above relevance measurements and adopt the above meta-graphs with dynamic personal weightings in influence diffusion of multiple relevant items.

\section{\opt{short}{Fundamental }Problem Formulation}
\label{sec:problem}

\label{para:diffusion}
\revise{To study various issues in multiple promotions of relevant items, we first introduce \opt{short}{two}\opt{full}{four} important factors, which can be easily incorporated into existing diffusion models\opt{short}{, e.g., triggering models \cite{Kempe2003},}\opt{full}{\footnote{For example, triggering models \cite{Kempe2003,sun2018multi} and reachability-based influence-diffusion models \cite{wang2015modeling}.}} by extending the diffusion process, to consider dynamic changes in personal perceptions of item relationships \cite{shocker2004product,shi2019semrec,gu2019relevance} and \opt{short}{their fundamental ripple effect on personal preferences for items \cite{shocker2004product,venkatesh2003optimal,gerlach2014never,frank1991microeconomics}.\footnote{\revise{The complete problem that considers the comprehensive ripple effect on social influence strength \cite{ma2015latent,yavacs2014impact,timmor2008being} and item associations \cite{zhang2020complements,koukova2012multiformat} is presented in the full version \cite{online}.}}}\opt{full}{the ripple effect on personal preference for items \cite{shocker2004product,venkatesh2003optimal,gerlach2014never,frank1991microeconomics}, social influence strength \cite{ma2015latent,yavacs2014impact,timmor2008being}, and item associations \cite{zhang2020complements,koukova2012multiformat}.}}
(1)~\opt{short}{\textit{Relevance measurement}:}\opt{full}{\textit{Relevance measurement} \cite{shi2019semrec,gu2019relevance,Huang2016Meta,zhang2016collaborative}:} KG is leveraged to measure the relevance between two items and find personal item networks, by learning the personal weightings on meta-graphs from users' previously adopted items.\revise{\footnote{\label{para:alternative}\revise{Instead of KG, some lightweight alternatives, such as Tagging algorithm~\cite{zheng2009substitutes}, Sceptre~\cite{mcauley2015inferring}, PMSC~\cite{wang2018path}, and DecGCN~\cite{liu2020decoupled} can also be adopted to learn item relationships. Since the above works derive the item relationships according to all users' adoption history, the item network is no longer personalized. When they are adopted in our problem, all personal item networks will be identical.}}}
(2)~\opt{short}{\textit{Preference estimation}:}\opt{full}{\textit{Preference estimation} \cite{zhao2017improving,banerjee2019maximizing,xin2019relational}:} Users' preferences for not-yet-adopted items are derived and updated based on previously adopted items and personal item networks.
\opt{full}{(3)~\textit{Influence learning} \cite{zhang2019learning,qiu2018deepinf,wu2019dual}: The strength of social influence between two users is inferred and updated according to their similarity in adopted items and personal item networks, since friends with similar backgrounds and intentions usually become closer and more easily influence each other \cite{gu2017co,farajtabar2015coevolve}.
(4)~\textit{Item associations} \cite{zhao2017improving,hung2016social,zhang2016collaborative,xin2019relational}: When users are promoted and preferring the promoted item, extra adoptions of relevant items are usually triggered due to item associations based on users' personal item networks and the probability of being promoted and preferring the promoted item.}
The dependency of these factors is illustrated in \revise{Fig.~\ref{Fig:simple_factors}}, while the discussion of deriving and updating them is detailed in \opt{short}{\cite{online}}\opt{full}{Sec.~\ref{sec:extend_problem}}. For ease of discussion, we summarize the notations in \revise{Table~\ref{T:simple_notation}}.

\opt{short}{\begin{figure}
    \centering
    \includegraphics[width=0.48\textwidth]{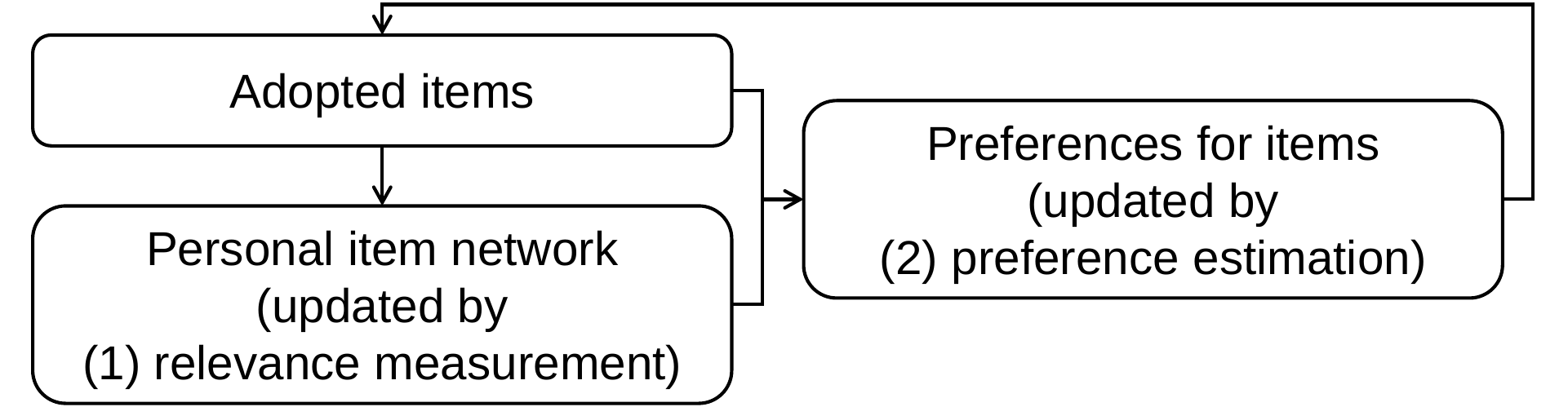}
    \caption{\revise{Illustration of the two dependent factors.}}
    \label{Fig:simple_factors}
\end{figure}
}
\opt{full}{
\begin{figure}
    \centering
    \includegraphics[width=0.48\textwidth]{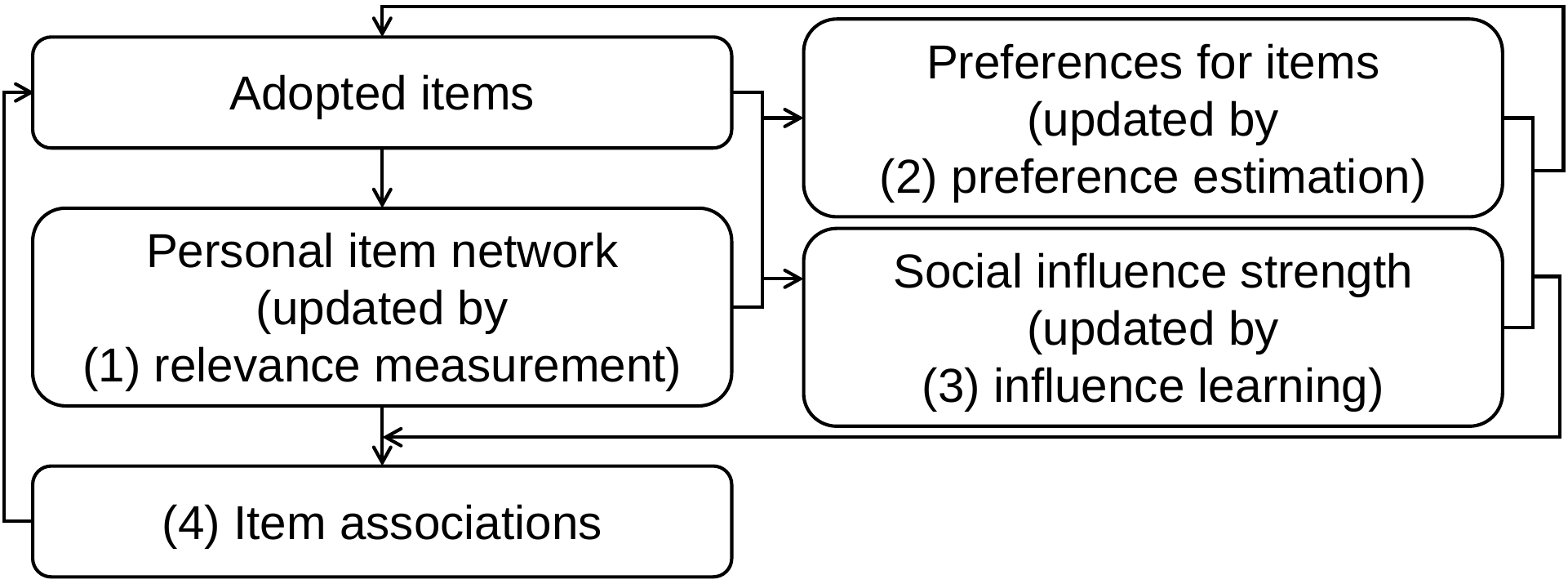}
    \caption{\revise{Illustration of the two dependent factors.}}
    \label{Fig:simple_factors}
\end{figure}
}

Accordingly, we elaborate on the \textit{diffusion process} as follows. A campaign includes $T$ promotions. The $t$-th promotion contains multiple steps $\zeta_{t}=0,1,\ldots$, where each step represents an influence propagation from users adopting items to their friends that haven't adopted those items yet.\footnote{Note that $\zeta_t -1 \equiv \zeta_{t-1}^\text{last}$ if $\zeta_t=0$, where $\zeta_{t-1}^\text{last}$ is the last step of the $(t-1)$-th promotion.} As a promotion depends on previous promotions, the initial state of a user in the $t$-th promotion \revise{(i.e., adopted item, perceptions, \opt{short}{and preferences}\opt{full}{preference, and influence strength} at $\zeta_{t}=0$)} is the same as the state at the end of the $(t-1)$-th promotion, while the seeded users in the $t$-th promotion newly adopt the promoted items at $\zeta_{t}=0$. When the diffusion starts at step $\zeta_{t} \geq 1$, a user $u$ may be promoted $x$ by any friend $u'$ who newly adopted $x$ at $\zeta_t -1$ only if $u$ has not adopted $x$ yet. \revise{The probability that $u$ will adopt $x$ is derived according to the social influence strength from $u'$ (denoted as $p_{u',u}$) and $u$'s preference for $x$ (denoted as $P_\text{pref}(u,x,\zeta_t-1)$), i.e., $p_{u',u} \times P_\text{pref}(u,x,\zeta_t-1)$\opt{short}{.}\opt{full}{ \cite{zhou2014preference}. In addition, being promoted $x$, $u$ may further adopt $x$'s relevant item $y$ (not previously adopted by $u$) due to \textit{(4) item associations}, according to the probability of $u$ being promoted and preferring $x$ and the relationships between $x$ and $y$.\footnote{Note that the extra adoptions of $x$'s relevant items can be independent of the adoption of $x$ since item associations are inspired by the promotion of $x$ rather than the purchase decision of $x$.}}} 
Then, at the end of this step, $u$'s personal perceptions of item relationships (i.e., personal item network) are updated by \textit{(1) relevance measurement} (detailed in \opt{short}{\cite{online}}\opt{full}{Sec.~\ref{sec:extend_problem}}) if $u$ newly adopts any item, while her preferences for not-yet-adopted items \opt{full}{and the influence strength }also change accordingly by \textit{(2) preference estimation} \opt{full}{and \textit{(3) influence learning} }(detailed in \opt{short}{\cite{online}).}\opt{full}{Sec.~\ref{sec:extend_problem}), respectively.} If there is any new adoption at $\zeta_t$, the next step $\zeta_t +1$ starts with users having new adoptions at $\zeta_t$ to promote their newly adopted items to their friends (who have not yet adopted those items). In other words, the diffusion of $t$-th promotion stops when no new adoptions happen since users cannot be promoted the adopted items again. Thus, the diffusion of the $(t+1)$-th promotion follows.

Based on the above diffusion process for relevant items in multiple promotions, we aim to choose a number of items, seed suitable users, and decide the proper timing, such that the influence spread (defined below) is maximized. Formally, $S=\{(u,x,t)\}$ is a seed group, where a seed $(u,x,t)$ indicates that an item $x$ is chosen for promotion starting at a seeded user $u$ in the decided $t$-th promotion.\footnote{An item $x$ can be assigned to multiple seeded users at multiple promotions; each promotion can promote multiple chosen items by multiple seeded users.} Let $S_t \subseteq S$ denote a subgroup of seeds chosen for the $t$-th promotion. We first define the influence spread and then formulate the problem as follows.

\opt{short}{
\renewcommand\theadalign{tl}
\begin{table}[t]
\caption{\revise{Summary of notations in Sections~\ref{sec:problem}-\ref{sec:algo}.}}
\label{T:simple_notation}
\input{simple_notation.tex}
\end{table}
}

\renewcommand\theadalign{tl}
\begin{table}[t]
\caption{Summary of notations.}
\vspace{-3mm}
\centering
\small
\begin{tabular}{|p{2.4cm}p{5.6cm}|}
\hline
\textbf{Notation} & \textbf{Description} \\ \hline
\hline\multicolumn{2}{|l|}{Sec.~\ref{sec:problem}} \\ \hline
$\zeta_t$ & Step $\zeta$ of the $t$-th promotion \\
$P_\text{pref}(u,y,\zeta_t)$ & $u$'s preference for $y$ updated at $\zeta_t$ \\
$P_\text{act}(u,v,\zeta_t)$ & $u$'s influence strength on $v$ updated at $\zeta_t$ \\
$P_\text{ext}(u,u',x,y,\zeta_t)$ & $u$'s extra adoption probability of $y$ (when being promoted $x$ by $u'$) at $\zeta_t$ \\
$(u,x,t)$; $(u,x)$ & Seed; nominee \\
$S=\{(u,x,t)\}$; $S_t$ & Seed group; subgroup of seeds in the $t$-th promotion \\
$w_x$; $\mathcal{W}$ & Importance of $x$; set of item importance \\
$\sigma^{G_\text{SN}}(S)$; $T$ & Importance-aware influence; number of promotions \\
$\{m^\textsf{C}\}$ / $\{m^\textsf{S}\}$ & Sets of meta-graphs for describing complementary/substitutable relationships \\
\hline\multicolumn{2}{|l|}{Sec.~\ref{sec:algo}} \\ \hline
$\tau$; $\mathcal{G}$ & Target market; set of target markets with common users\\
$\bar{r}^{\mathtt{C}}_{x,y}$ / $\bar{r}^{\mathtt{S}}_{x,y}$ & 
Average complementary/substitutable relevance between $x$ and $y$ per user (the timing is specified from context) \\
$DR^{\mathcal{W},\tau}(S^\mathcal{G},x)$ & $x$'s dynamic reachability of $\tau$'s users given $S^\mathcal{G}$ and $\mathcal{W}$\\
$SI^{\tau}(S^\mathcal{G},(u,x,t),T)$ & Substantial influence of $(u,x,t)$ in $\tau$ given $S^\mathcal{G}$ and $T$\\
$\hat{t}$ & Latest promotional timing in $S^\mathcal{G}$ \\
\hline\multicolumn{2}{|l|}{Sec.~\ref{sec:extend}} \\ \hline
$W_\text{meta}(u,m,\zeta_t)$ & $u$'s weighting of meta-graph $m$ at $\zeta_t$ \\
$s(x,y \mid m)$ & Relevance between $x$ and $y$ defined by meta-graph $m$ \\ 
$r^{\texttt{C}}(u,x,y,\zeta_t)$ / $r^{\texttt{S}}(u,x,y,\zeta_t)$ & The complementary/substitutable relevance between $x$ and $y$ in $u$'s perception updated at $\zeta_t$ \\
$G_\text{PIN}(u,\zeta_t)$ & $u$'s personal item network updated at $\zeta_t$ \\
$A(u,\zeta_t)$ & $u$'s adoption set of items updated at $\zeta_t$ \\
$PI^{\mathcal{W},\tau}(S^\mathcal{G},x,d)$; $RI^{w_x,\tau}(S^\mathcal{G},x,d)$; $d^\tau$ & $x$'s proactive impact of $\tau$'s users propagated within $d$ hops given $S^\mathcal{G}$ and $\mathcal{W}$; $x$'s reactive impact of $\tau$'s users propagated within $d$ hops given $S^\mathcal{G}$; diameter of $\tau$\\
$\mathcal{L}^{\mathtt{C},\tau}(x,y,S^\mathcal{G})$ / $\mathcal{L}^{\mathtt{S},\tau}(x,y,S^\mathcal{G})$ & Likelihood of regarding $x$ and $y$ as complementary/substitutable for each user in $\tau$ given $S^{\mathcal{G}}$ \\
$MA^{\tau}(S^\mathcal{G}, (u,x,t))$; $ML^{\tau}(S^\mathcal{G}, (u,x,t))$ & Marginal adoption of $(u,x,t)$ in $\tau$ given $S^\mathcal{G}$; marginal likelihood of $(u,x,t)$ in $\tau$ given $S^\mathcal{G}$ \\
$\pi^{\tau}(S^\mathcal{G})$ & Likelihood of not-yet-adopted items being adopted in $\tau$ in future promotions given $S^\mathcal{G}$ \\
\hline
\end{tabular}
\label{T:simple_notation}
\end{table}

\revise{
\begin{definition}[\opt{short}{Influence}\opt{full}{Importance-aware influence} function]
\label{def:inf}
\opt{short}{Let $T$ denote the number of promotions.}\opt{full}{Let $w_x$ be the importance of an item $x$ and $T$ denote the number of promotions.\footnote{Importance reflects how much items are valued in this campaign, e.g., higher-quality items in product lines usually have a greater importance \cite{moorthy1984market}.}}
For a seed group $S$, the influence spread in the social network $G_\text{SN}=(V,E)$, denoted as $\sigma^{G_\text{SN}}(S)$, is the expected adoptions \opt{full}{(weighted by item importance) }in all $T$ promotions, i.e., \opt{short}{$\sigma^{G_\text{SN}}(S) = \sum\limits_{t=1}^T \sigma_t^{G_\text{SN}}(S_t \mid S_1, S_2, \ldots, S_{t-1}) = \sum\limits_{t=1}^T \sum\limits_{x \in I} n_x(S_t \mid S_1, S_2, \ldots, S_{t-1})$}\opt{full}{$\sigma^{G_\text{SN}}(S) = \sum\limits_{t=1}^T \sigma_t^{G_\text{SN}}(S_t \mid S_1, S_2, \ldots, S_{t-1}) = \sum\limits_{t=1}^T \sum\limits_{x \in I} w_x n_x(S_t \mid S_1, S_2, \ldots, S_{t-1})$}, where $n_x(S_t \mid S_1, S_2, \ldots, S_{t-1})$ is the expected new adoptions of $x$ for $S_t$ in the $t$-th promotion conditioned on $S_1, \ldots, S_{t-1}$ in previous promotions.\footnote{\revise{\label{para:monte_carlo}Following \cite{Kempe2003,chen2010scalable}, $\sigma$ is estimated by the Monte Carlo method, which simulates the influence diffusion of seeds according to the probabilities.}} (When $G_\text{SN}$ is clear from context, we write $\sigma(S)$ for short.)
\end{definition}
}

\begin{definition}[\revise{\opt{short}{\fullsproblem\ (\sproblem)}\opt{full}{\fullproblem\ (\problem)}}]
Let $m^\textsf{C}$ and $m^\textsf{S}$ denote the meta-graphs for describing the complementary and substitutable relationships between items, respectively. Based on the diffusion process described earlier, given \revise{a social network $G_\text{SN}=(V,E)$\opt{short}{ with the influence strength $p_{u,v}$ for all $u,v \in V$}}, a KG $G_\text{KG}=(\mathcal{V},\mathcal{E},\Phi,\Psi)$, two sets of meta-graphs $\{m^\textsf{C}\}$ and $\{m^\textsf{S}\}$, a target item set $I=\{x\}$\opt{full}{ together with its importance set $\mathcal{W}=\{w_x\}$}, the cost $c_{u,x}$ of hiring a user $u \in V$ to promote an item $x \in I$, the budget $b$, and the total number of promotions $T$, the \opt{short}{\sproblem}\opt{full}{\problem} problem is to find the seed group $S = \bigcup_{t=1}^T S_t$ such that the influence spread $\sigma(S)$ is maximized within the budget constraint $b$, i.e., $\sum_{t=1}^{T}\sum\limits_{(u,x,t) \in S_t}{c_{u,x}} \leq b$.
\end{definition}

\label{para:problem}

\begin{theorem} \label{prob-hard}
\opt{short}{\sproblem}\opt{full}{\problem} cannot be approximated within $O(\frac{1}{|V|^{1-\epsilon}})$ in polynomial time unless $P=NP$, even with only the complementary relationship and $P_\text{ext} \equiv 0$ in only one promotion.
\end{theorem}
\opt{short}{\revise{\begin{proofsk}\label{proofsk1}\input{3_revision_problem_proofsketch}\end{proofsk}}}
\opt{full}{
\begin{proof}
We prove the theorem with the gap-introducing reduction from the decision problem of Set Cover. Given a set cover instance, by constructing a corresponding special case of \problem, we can prove that if there is a set cover solution with at most $k$ sets, there is a feasible solution of \problem\ with the total influence at least $|U|^{c}+2|U|+2$, where $|U|$ is the size of the ground set of set cover instance, and $c$ is a large constant. Otherwise, if there does not exist a set cover solution with at most $k$ sets, the optimal value of \problem\ is at most $2|U|+k+2$. 
Then, we assign $c$ suitably to satisfy $|V|^{1-\epsilon}\leq \frac{|U|^{c}+2|U|+2}{2|U|+k+2}$, where $|V|$ and $\epsilon$ are related to $|U|$ and $c$. Consequently, if there is a $|V|^{1-\epsilon}$ approximation algorithm of \problem, we can solve the decision problem of set cover in polynomial time, which implies $P=NP$, which is a contradiction.

Specifically, given a ground set $U$ consisting of elements and a sub-collection $\mathcal S$ on $2^U$ consisting of sets, an element is covered by a set if the set contains the element. The decision problem of Set Cover asks whether there are at most $k$ sets to cover all the elements.
Given a Set Cover instance $(U, \mathcal S)$, we construct an instance of \problem\ as follows.
First, we construct the topology structure of the special case.
Let each element and each set correspond to a user node in the social network. We number each element node by $1,2,\ldots,\left| U \right|$. Denote the sets of element nodes and set nodes by $V_e=\{ v_{e_1},v_{e_2},\ldots,v_{e_{\left|U\right|}}\}$ and $V_S=\{ v_{S_1},v_{S_2},\ldots,v_{S_{\left|\mathcal S\right|}}\}$, respectively.
For each pair of set node $v_{S_j}$ and element node $v_{e_i}$, we construct a directed edge from $v_{S_j}$ to $v_{e_i}$ if set $S_j$ covers element $e_i$. We add a new node $v_b$ to connect each element node $v_{e_i} \in V_e$ by $\left|U\right|$ new disjoint directed paths with length $2$ from $v_b$ to each $v_{e_i}$ ($1\leq i\leq \left|U\right|$). We also add a new directed path $P_{C}$ with length $\left|U\right|-1$ and number each node of the path from one end node to another by $1,2,\ldots,\left|U\right|$. Denote the node set of path $P_{C}$ by $V_{P_C}=\{ v_{p_1},v_{p_2},\ldots,v_{p_{\left|U\right|}}\}$. Then, we copy $P_C$ by $\left|U\right|^c$ times, where $c$ is a large constant. Moreover, we sort the $\left|U\right|^c$ copied paths by an arbitrary ordering and label them by $1,2,\ldots,\left|U\right|^c$. We add a new node $v_a$ to connect the same end node $v_{p_1}$ of all the $\left|U\right|^c$ paths. Afterwards, we connect each element node $v_{e_i}$ ($1\leq i\leq \left|U\right|-1$) to the corresponding node $v_{p_i}$ of all the $\left|U\right|^c$ paths, by $\left|U\right|^c$ new disjoint directed paths from $v_{e_i}$ to $v_{p_{i}}$ with corresponding length $2i-1$. In particular, when $i=\left|U\right|$, we connect the element node $v_{e_{\left|U\right|}}$ to the corresponding end node $v_{p_{\left|U\right|}}$ of the first path by a directed path from $v_{e_{\left|U\right|}}$ to $v_{p_{\left|U\right|}}$ with length $2\left|U\right|-1$. Last, we connect each end node $v_{p_{\left|U\right|}}$ of the $\left|U\right|^c$ paths by exact one directed path $P_{T}$ from the end node $v_{p_{\left|U\right|}}$ of the first path to the corresponding end node of the last path. Note that $P_{T}$ has $\left|U\right|^c$ nodes since it consists of $\left|U\right|^c$ end nodes $v_{p_{\left|U\right|}}$ of $\left|U\right|^c$ copies of path $P_C$. The illustration of the constructed special case is shown in Figure~\ref{hardness}.

Second, we set the parameter of each edge and node in both the social network and KG.
Assume the influence strength of each edge in social network is equal to $1$, i.e., $P_\text{act}\equiv 1$, and the budget $b=k$.
We construct $\left|U\right|+1$ items numbered by $1,2,\ldots,\left|U\right|+1$. Denote the item set by $I=\{x_1,x_2,\ldots,x_{\left|U\right|+1}\}$.
For each user, except users corresponding to nodes $v_a$, $v_b$, and set nodes in $V_{S}$, we set their costs for promoting any item as a large constant $M_0>k$. Moreover, we set both $v_a$'s and $v_b$'s cost for promoting all the items as $0$ and each set node $v_{S_j}$'s ($v_{S_j}\in V_S$) cost for promoting all the items as $1$. In other words, nodes $v_a$ and $v_b$ can be chosen as seeds with all the items for free. We set $w_{x_i}=0$ for $i=1,\ldots,\left|U\right|$ and $w_{x_{\left|U\right|+1}}=1$. That is, only promoting item $x_{\left|U\right|+1}$ directly or promoting users to adopt item $x_{\left|U\right|+1}$ can bring influence.
Assume that there is only the complementary relationship between items.
For each element node $v_{e_i}$, we set $P_\text{pref}(v_{e_i},x_1)=1$;\footnote{For simplicity, in this paper, when $\zeta$ and $t$ are clear or have little impact on the proof, we ignore them in the presentation of $P_\text{ext}$, $P_\text{pref}$ and $P_\text{act}$.} otherwise, $P_\text{pref}(v_{e_i},x_j)=0$ where $j\neq 1$. That is, each element node $v_{e_i}$ can only initially adopt item $x_1$. Assume after $v_{e_i}$ adopts item $x_1$, its preference for item $x_{i+1}$ increases from $0$ to $1$. For each mid-node $v_i$ of corresponding directed path $P_{v_bv_iv_{e_i}}$ from $v_b$ to element node $v_{e_i}$ ($i=1,2,\ldots,\left|U\right|$), we set $P_\text{pref}(v_i, x_{i+1})=1$ and the preferences for other items as $0$. For each mid-node $v_j$ of the corresponding directed path with length $2i-1$ from element node $v_{e_i}$ ($i=1,2,\ldots,\left|U\right|$) to the corresponding node $v_{p_i}$ of $P_{C}$, we set $P_\text{pref}(v_j, x_{i+1})=1$ and the preferences for other items as $0$. For each node $v_{p_i}$ ($i=1,2,\ldots,\left|U\right|$) of $P_{C}$, we set $P_\text{pref}(v_{p_i}, x_{i})=1$, and $P_\text{pref}(v_{p_i}, x_{j})=0$ for $j \neq i$; after $v_{p_i}$ adopts item $x_i$, $P_\text{pref}(v_{p_i}, x_{i+1})=1$. 

Next, we show the relation between any feasible solution of Set Cover and that of the constructed special case of \problem.
For any feasible solution of Set Cover, we construct a feasible solution of \problem\ by 1) letting the chosen set nodes promote items arbitrarily and 2) seeding users corresponding to nodes $v_a$ and $v_b$ to promote all the items 
without violating the budget constraint $k$.
On the other hand, for any feasible solution of \problem, we construct a corresponding feasible solution of Set Cover by 1) discarding those seeding nodes that are not set nodes in the solution of \problem\ and 2) regarding the seeding set nodes in the solution of \problem\ as the corresponding chosen sets. Note that the constructed solution is still a feasible solution of Set Cover, i.e., the number of chosen set is not larger than $k$.

Finally, we derive the gap.
If there is a Set Cover solution with at most $k$ sets, in order to maximize the influence of adopting item $x_{\left|U\right|+1}$, we let each set node in $V_S$ choose item $x_{1}$ as seeds, and let nodes $v_a$, $v_b$ choose all the items as seeds satisfying the budget constraint $k$. As a consequence, all the nodes in path $P_{T}$ adopt item $x_{\left|U\right|+1}$, implying that at least $\left|U\right|^{c}+2\left|U\right|+2$ users adopt item $x_{\left|U\right|+1}$ with total influence at least $\left|U\right|^{c}+2\left|U\right|+2$. 
To see this, assume that the initial seeding step is $0$ and each step is one-hop influence propagation.
After initial seeding, at step $1$, each set node $v_{S_i}$ ($i=1,2,\ldots,\left|U\right|$) adopts item $x_{1}$, and their preferences for item $x_{i+1}$ are increased from $0$ to $1$. At step $2$, $v_b$ propagates influence to make $v_{S_i}$ adopt item $x_{i+1}$ at step $2$. On the other hand, $v_{p_1}$ first adopts item $x_1$ and its preference for item $x_2$ is increased from $0$ to $1$ at step $1$. Afterwards, $v_{S_1}$ continues to propagate influence to make $v_{p_1}$ adopt item $x_{2}$ at step $3$ because $P_\text{pref}(v_{p_1}, x_{2})=1$ at step $2$. Similarly, each $v_{S_i}$ ($i=2,3,\ldots,\left|U\right|$) continues to propagate influence to make $v_{p_i}$ of directed path $P_C$ adopt item $x_{i+1}$ at step $2i+1$ after each $v_{p_i}$ adopts item $x_{i}$ at step $2i$. Moreover, all the end nodes $v_{p_{\left|U\right|}}$ of $P_C$ (i.e. all the nodes of $P_T$) will adopt item $x_{\left|U\right|}$ at step $2\left|U\right|$. Notice that, at step $2\left|U\right|+1$, the preferences for item $x_{\left|U\right|+1}$ of all the end nodes $v_{p_{\left|U\right|}}$ in $P_C$ (i.e. all the nodes of $P_T$) have been increased from $0$ to $1$.    
Consequently, the first end node of $P_T$ adopts item $x_{\left|U\right|+1}$ at step $2\left|U\right|+1$ after its in-neighbor adopts item $x_{\left|U\right|+1}$ at step $2\left|U\right|$. Then, all the left nodes of $P_T$ adopt item $x_{\left|U\right|+1}$ at step $2\left|U\right|+2, 2\left|U\right|+3,\ldots,2\left|U\right|+\left|U\right|^c$.

Otherwise, if there does not exist a Set Cover solution with at most $k$ sets to maximize the influence of adopting item $x_{\left|U\right|+1}$, then at most $2\left|U\right|+k+2$ users (consisting of $v_a$, two paths $P{v_bv_{\left|U\right|}v_{e_{\left|U\right|}}}$ and $P{v_{e_{\left|U\right|}},\ldots,v_{p_{\left|U\right|}}}$, and any $k$ set nodes in $v_S$) adopt item $x_{\left|U\right|+1}$ with total influence at most $2\left|U\right|+k+2$. By setting $c$ sufficiently large, we construct the special case of \problem\ with user node size $\left|V\right|$ such that $\left|U\right|^{c}+2\left|U\right|+2\geq \left|V\right|^{1-\frac{\epsilon}{2}}$ and $2\left|U\right|+k+2\leq \left|V\right|^{\frac{\epsilon}{2}}$, respectively, where $\epsilon$ is an arbitrarily small constant. 
Notice that the total nodes of the newly constructed graph can be partitioned into two parts: one is the nodes in the $\left|U\right|^c$ paths; the other is the left nodes. It can be seen that the number of the left nodes is polynomial on $\left|U\right|$. Thus, assume that the total left nodes is $\left|U\right|^{c_1}$, where $c_1$ is a constant. Since each path has $\left|U\right|$ nodes, the total number of nodes of $\left|U\right|^c$ paths is $\left|U\right|^{c+1}$. Then, $\left|V\right|=\left|U\right|^{1+c}+\left|U\right|^{c_1}=\left|U\right|\cdot \left|U\right|^{c}+\left|U\right|^{c_1}$. Thus, given an arbitrarily small constant $\epsilon$, there exists a large enough constant $c$ such that $\left|U\right|^{c}+2\left|U\right|+2 \geq \left|V\right|^{1-\frac{\epsilon}{2}}$. Note that $2\left|U\right|+k+2=O(\left|V\right|^{\frac{1}{1+c}})\leq \left|V\right|^{\frac{\epsilon}{2}}$ when $c$ is sufficiently large. Therefore, if there is a $\left|V\right|^{1-\epsilon}=\frac{\left|V\right|^{1-\frac{\epsilon}{2}}}{\left|V\right|^{\frac{\epsilon}{2}}} \leq \frac{\left|U\right|^{c}+2\left|U\right|+2}{2\left|U\right|+k+2}$ approximation algorithm, then we can solve the decision problem of Set Cover in polynomial time, which implies $P=NP$, leading to a contradiction.
\end{proof}

\begin{figure}[t]
            \centering
            \includegraphics[height=2in,width=1.05\linewidth]{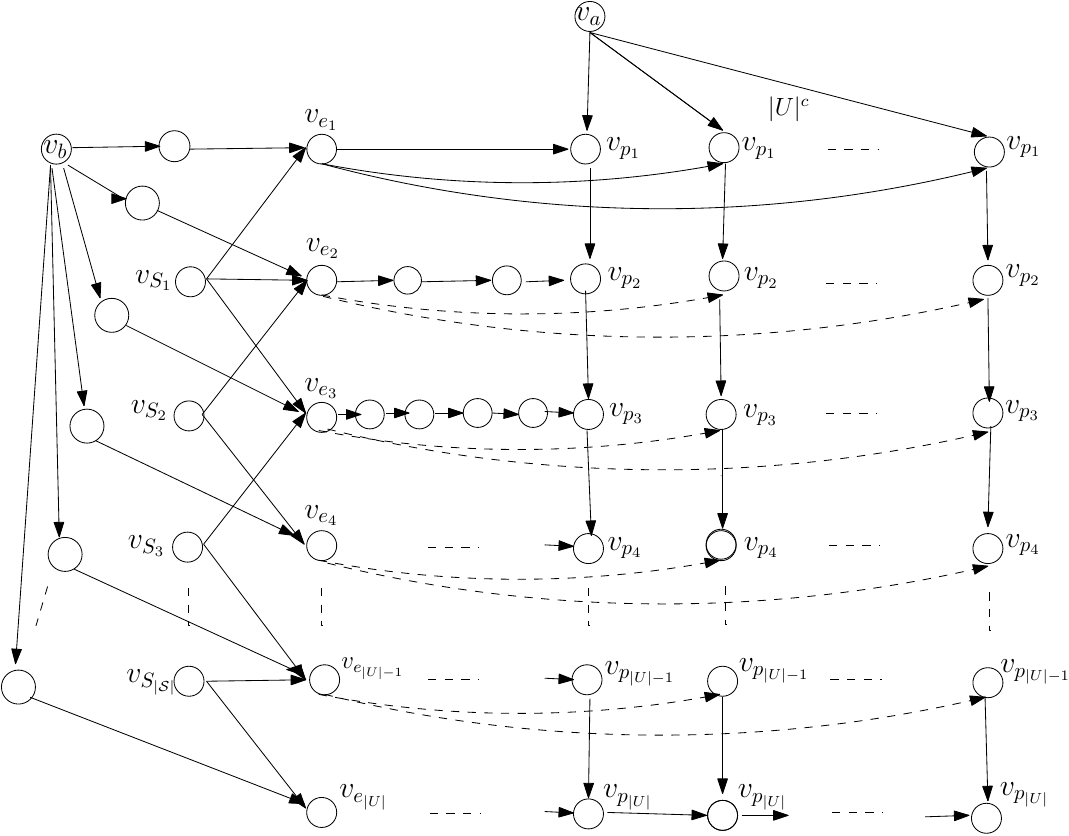}
    \caption{\small{Illustrative example of hardness proof.}}
    \label{hardness}
\end{figure}

In addition, when $P_\text{pref}\equiv 1$ and $P_\text{ext}\equiv 0$, assume there is only one item, which implies a single user can be regarded as a seed directly. A user's adoption decision of an item is thus equivalent to checking whether the user is influenced, which simply depends on the influence strength $P_\text{act}$.
Therefore, \problem\ is reduced to the conventional influence maximization problem that is to find $k$ users as seeds to maximize the influence spread. Since the influence maximization problem cannot be approximated within $1-\frac{1}{e}+\epsilon$, the special case of {\problem} cannot, either.
}

\section{Approximation Algorithm}
\label{sec:algo}
\subsection{Algorithm Overview}
\label{sec:algo_overview}

To efficiently solve \opt{short}{\sproblem}\opt{full}{\problem}, we design an approximation algorithm, namely \opt{short}{\fullsalgo\ (\salgo)}\opt{full}{\fullalgo\ (\algo)}, which embodies a number of ideas. (i) To tackle the challenge in the propagation of item impacts, \opt{short}{\salgo}\opt{full}{\algo} introduces \textit{Dynamic Reachability (DR)} to measure the impact made by an item promotion and the impacts resulted from the promotions of other items based on users' dynamic perceptions of item relationships (detailed later in Eq.~\eqref{eq:dr}). Specifically, DR evaluates both \textit{proactive} and \textit{reactive impacts} for each item. For an item, the proactive impact is the probability for this item to result in an increase of users' preferences \textit{on other items}. The reactive impact is the probability to increase users' preferences \textit{on this item} resulted from other items promoted previously. The item with the highest DR is prioritized for promotion. Previous works \cite{teng2018revenue,hung2016social,banerjee2019maximizing} select users only and do not consider the items in IM.

(ii) To avoid antagonism between substitutable items, \opt{short}{\salgo}\opt{full}{\algo} identifies \textit{target markets}, each of which consists of socially close users to promote complementary items in consecutive promotions. Specifically, it identifies some \textit{nominees} (where a nominee is a user-item pair $(u,x)$) as candidate seeds, denoted by $(u,x,t)$, for an incoming promotion at time $t$ (decided later). Note that a target market targets on a cluster of nominees in order to promote complementary items to socially close users. Since different target markets may share some common users, it is important to avoid promoting substitutable items to them. Accordingly, \opt{short}{\salgo}\opt{full}{\algo} prioritizes the target market promoting items with the least substitutable relevance to items in the overlapping target markets (i.e., the target markets sharing many common users). By contrast, prior works \cite{teng2018revenue,hung2016social,banerjee2019maximizing} consider only one relationship and thereby may promote substitutable items to the same users.

\opt{full}{(iii) To find the promotional timing, \algo\ introduces \textit{Substantial Influence (SI)} for each candidate seed $(u,x,t)$ to evaluate immediate and subsequent adoptions if nominee $(u,x)$ is assigned as the seed in the $t$-th promotion. Specifically, SI derives the \textit{marginal adoption} and \textit{marginal likelihood} of $(u,x,t)$, where the former is the difference of the total adoptions with and without $(u,x,t)$, and the latter is the difference of the likelihood with and without $(u,x,t)$ for not-yet-adopted items to be adopted after the $t$-th promotion. $(u,x,t)$ with the highest SI is selected as a new seed at each iteration. In contrast, previous works \cite{teng2018revenue,hung2016social,banerjee2019maximizing} are designed for only one promotion and do not leverage item impacts at different timings. }

\opt{short}{
\setlength{\algomargin}{2.5ex}
\SetInd{3.5pt}{3.5pt}
\begin{algorithm}[t]
\caption{\salgo}
\label{algo:simple_main}
\input{simple_algo.tex}
\end{algorithm}
}
%
%

\setlength{\algomargin}{2.5ex}
\SetInd{3.5pt}{3.5pt}
\begin{algorithm}[t]
\caption{DPSP}
\scalefont{0.90}
\DontPrintSemicolon
\SetNlSty{texttt}{}{}
\SetAlgoNlRelativeSize{-1}
\SetNlSkip{0.25em}
\SetVlineSkip{1pt}
\label{algo:main}
\SetKwFunction{selectNominees}{selectNominees}
\SetKwFunction{clusterNominees}{clusterNominees}
\SetKwFunction{prioritizeTargetMarket}{prioritizeTargetMarket}
\KwIn{Social network $G_\text{SN}=(V,E)$; knowledge graph $G_\text{KG}=(\mathcal{V},\mathcal{E},\Phi,\Psi)$; item set $I$; item importance set $\mathcal{W}$; total budget $b$; total number of promotions $T$}
\KwOut{Seed group}
\tcc{TMI phase}
$U \gets \{(u,x) \mid u \in V, x \in I\}$\;
$\{N^\tau\} \gets $ \clusterNominees{$N$}\;
\For{each $N^\tau$}{
    Identify the target market $\tau$ by $N^\tau$\;
}
$CG \gets$ \prioritizeTargetMarket{$\{\tau\}$}\;
\For{each $\mathcal{G}$ in $CG$}{
    $S^\mathcal{G} \gets \emptyset$\;
	\For{each $\tau_k \in \mathcal{G}$, where $k=1,2,\ldots$}{
	    \tcc{DRE phase}
	    $N^{\tau_k} \gets $ nominees in $\tau_k$\;
	    $T^{\tau_k} \gets \lfloor\frac{\left| N^{\tau_k}\right| \cdot T}{\sum_{\tau_i \in \mathcal{G}}{\left| N^{\tau_i}\right|}}\rfloor$\;
		$I^{\tau_k} \gets \{x \mid (u,x)\in N^{\tau_k}\}$\;
        \While{$I^{\tau_k} \not= \emptyset$}{
	        $x_\text{p} \gets \argmax_{x \in I^{\tau_k}}{DR^{\mathcal{W},\tau_k}(S^{\mathcal{G}},x)}$\;
	        $I^{\tau_k} \gets I^{\tau_k} \setminus \{x_\text{p}\}$\;
	        $N_\text{p} \gets \{(u,x_\text{p}) \mid (u,x_\text{p}) \in N^{\tau_k}\}$\; 
	        \tcc{TDSI phase}
            \While{$N_\text{p} \not= \emptyset$}{
                \If{$S^\mathcal{G} \not= \emptyset$}{
		            $\hat{t} \gets \max \{t \mid (u,x,t) \in S^\mathcal{G}$\}\;
		        }
		        \Else{
		            $\hat{t} \gets 1$\;
		        }
		        $C \gets \emptyset$\;
		        \For{$t \gets \hat{t}$ \KwTo $\hat{t}+1$}{
		            \If{$t \leq \sum_{i\leq k}T^{\tau_i}$}{
		                \For{$(u,x_\text{p}) \in N_\text{p}$}{
		                    $C \gets C \cup \{(u,x_\text{p},t)\}$\;
		                }
		            }
		        }
		        $(u_\text{s}, x_\text{p}, t_\text{s}) \gets \argmax\limits_{(u,x_\text{p},t) \in C}SI^{\tau_k}(S^\mathcal{G}, (u,x_\text{p}, t), T)$\;
		        $N_\text{p} \gets N_\text{p} \setminus \{(u_\text{s}, x_\text{p})\}$\;
		        $S^{\mathcal{G}} \gets S^{\mathcal{G}} \cup \{(u_\text{s}, x_\text{p}, t_\text{s})\}$\; 
	        }
        }
	}
}
\KwRet{$\bigcup_\mathcal{G}S^\mathcal{G}$}\;
\end{algorithm}

\label{para:algo_summary1}
\revise{Equipped with the above strategies, \opt{short}{\salgo}\opt{full}{\algo} includes \opt{short}{two}\opt{full}{three} phases: Target Market Identification (TMI) and Dynamic Reachability Evaluation (DRE).}\footnote{\revise{\opt{short}{The complete algorithm \algo\ to tackle all challenges is presented in \cite{online}. Besides, our}\opt{full}{Our} proposed algorithm can deal with adaptive IM (even without a predefined budget allocation to different promotions), detailed in \opt{short}{\cite{online}}\opt{full}{Sec.~\ref{sec:adaptive}}.}} Since users in social networks usually have different needs and diverse purchase intentions, a promotional strategy is planned more sophisticatedly if the target users are identified first. Intuitively, intensively promoting a few items within a short period can better draw users' attention. Hence, \opt{short}{\salgo}\opt{full}{\algo} first exploits TMI to identify target markets and then leverages DRE \opt{full}{and TDSI }to plan the distinct effective promotional strategy for each target market. Specifically, TMI selects and clusters nominees to promote complementary items to each target market and prioritizes target markets with fewer substitutable items to the nominees in the overlapping target markets. \revise{For each target market, DRE finds the item with the highest DR to exploit item impacts and \opt{short}{decides the promotional timings for the corresponding nominees to be seeds. Algorithm~\ref{algo:simple_main}}\opt{full}{suggests the corresponding nominees as candidate seeds. TDSI decides the promotion timing $t$ for the candidate seed $(u,x,t)$ with the highest SI and adds it to the seed group. Algorithm~\ref{algo:main}}} presents the pseudo-code of \opt{short}{\salgo}\opt{full}{\algo}.\opt{short}{\footnote{Please refer to \cite{online} for more pseudo-codes.}}
\label{para:algo_summary2}

\subsection{Algorithm Description}
\label{sec:algo_description}

\textit{1. Target Market Identification (TMI):}
TMI selects the nominees that exert large influence spread, clusters select nominees to identify each target market for promoting complementary items to socially close users, and prioritizes the target market with fewer substitutable items to the nominees in the overlapping target markets.

For nominee selection, TMI carefully examines the marginal gain of influence for each nominee. It's crucial to select a cost-effective nominee due to different costs of nominees and a limited budget. Therefore, we propose \textit{marginal cost-performance ratio (MCP)} to jointly consider the above factors and ensure the approximation ratio of \opt{short}{\salgo}\opt{full}{\algo} in Theorem~\ref{the-sp3}.   
Specifically, given a set $N$ of selected nominees, MCP of a nominee $(u,x)$ is $\frac{f(N \cup \{(u,x)\})-f(N)}{c_{u,x}}$, where \revise{$f$ is the \opt{short}{influence spread $\sigma$ with the nominees placed in the first promotion as the seeds and $P_\text{pref}$}\opt{full}{importance-aware influence spread $\sigma$ with the nominees placed in the first promotion as the seeds and $P_\text{pref}$, $P_\text{act}$, and $P_\text{ext}$} assigned at the beginning of this promotion.} For the nominees with the costs satisfying $ c_{u,x} < b - \sum_{(u',x') \in N}{c_{u',x'}}$, TMI iteratively extracts the one with the highest MCP into $N$.

Afterward, TMI identifies the target markets by clustering the nominees. To promote complementary items to the users in a target market, TMI first clusters the nominees in $N$ (e.g., by clustering methods POT \cite{chen2017people} and FGCC \cite{wang2015concurrent}\label{cluster})\revise{\footnote{\revise{POT discovers social communities according to users' opinions. TMI employs POT by treating the opinions of a community as items of a target market. Accordingly, TMI regards the detected item-based user communities as the clusters of nominees, where the users in the same detected community are socially close to each other, and their promoting items have larger complementary relevance than substitutable relevance. On the other hand, FGCC identifies the clusters for two types of related objects simultaneously. TMI employs FGCC by constructing the incidence matrix of users and items to identify their clusters. TMI clusters the nominees together if their users and items belong to the same user and item clusters, respectively.}}} according to the social distances between the nominees and the relevance between their promoting items, i.e., $\bar{r}^{\mathtt{C}}_{x,y} - \bar{r}^{\mathtt{S}}_{x,y}$, where $\bar{r}^{\mathtt{C}}_{x,y}$ and $\bar{r}^{\mathtt{S}}_{x,y}$ are the average complementary and substitutable relevance between $x$ and $y$ over all users, respectively.\footnote{The derivation of relevance is described in \opt{short}{\cite{online}.}\opt{full}{Sec.~\ref{sec:extend_problem}.}} Larger complementary and smaller substitutable relevance are encouraged. \revise{For each cluster, a target market $\tau$ is identified by exploring the influenced users from the nominees $N^{\tau}$ (e.g., by MIOA \cite{chen2010scalable}\label{influenced_user}).\footnote{\revise{MIOA identifies an influence region from a source node, and other nodes in the region can be reached from the source with sufficient influence probabilities on the maximum influence paths. TMI employs MIOA to identify all users that can be effectively influenced by the users of the nominees as the users in the target market.}}}
\revise{\label{para:budget}Note that with TMI, the budget allocation of \opt{short}{\salgo}\opt{full}{\algo} is realistic since a larger target market is inclined to have a larger budget to promote items. In TMI, the target markets are identified by the influence of nominees, where more nominees and influential nominees lead to a larger target market. As more nominees and influential nominees usually incur a larger cost \cite{Nguyen2016Cost}, \opt{short}{\salgo}\opt{full}{\algo} allocates larger budgets to those target markets accordingly.}

Afterward, TMI prioritizes the target market with fewer substitutable items to the nominees in the overlapping target markets. Let $\mathcal{G}$ denote a set of target markets with common users. A target market $\tau_i$ is in $\mathcal{G}$ if there is another target market $\tau_j \in \mathcal{G}$ with the common user number above a threshold $\theta$.\footnote{The sensitivity of \opt{short}{\salgo}\opt{full}{\algo} to $\theta$ is evaluated in \opt{short}{\cite{online}.}\opt{full}{Sec.~\ref{sec:sensitivity}.}} TMI arranges the promoting order for the target markets in each $\mathcal{G}$ by deriving \textit{Antagonistic Extent (AE)} of each target market $\tau_i$ according to the substitutable relationship between every promoting item $x$ and the items of other target market $\tau_j$, i.e., $AE(\tau_i) = \sum_{x\in \tau_i,y\in \tau_j}{\bar{r}^{\mathtt{S}}_{x,y}}$, where $\tau_i, \tau_j \in \mathcal{G}, i \neq j$.
The target market (and the items in the corresponding nominees) with a smaller AE is promoted earlier in $\mathcal{G}$.\revise{\footnote{\revise{Alternatively, according to research in the marketing field, the profitability~\cite{simkin1998prioritising} of a target market is also a good metric to prioritize target markets. The comparison of different marketing orders is presented in Sec.~\ref{sec:market_order}.}}}

\begin{figure}
    \centering
    \opt{short}{
    \hspace*{\fill}%
    }
    \subfigure[]{
        \centering
        \includegraphics[width=0.13\textwidth]{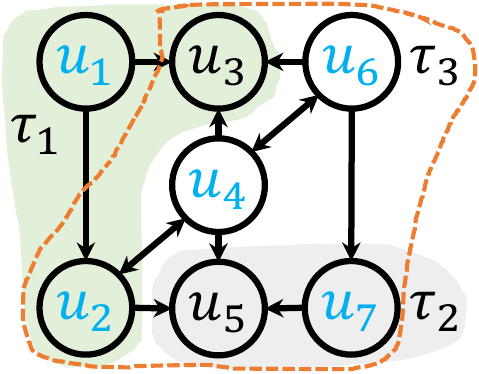}
        \label{Fig:fig_algoex_1_a}
    }\opt{short}{\hfill}%
    \subfigure[]{
        \centering
        \includegraphics[width=0.16\textwidth]{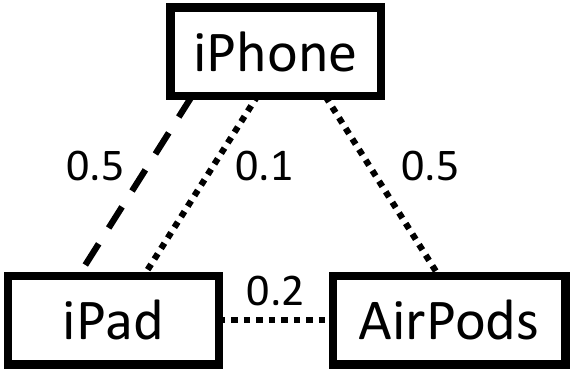}
        \label{Fig:fig_algoex_1_b}
    }%
    \opt{short}{
    \hspace*{\fill}%
    }
    \opt{full}{\subfigure[]{
        \centering
        \includegraphics[width=0.16\textwidth]{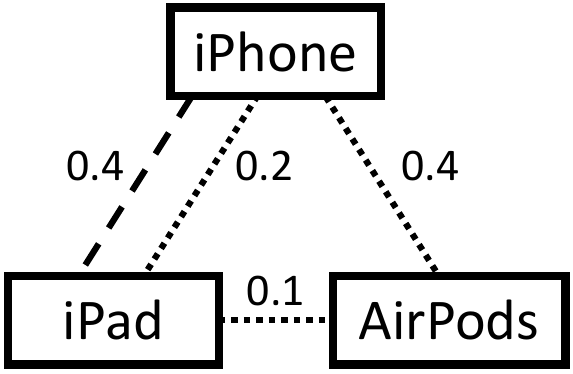}
        \label{Fig:fig_algoex_2}
    }\hfill}
    \caption{\revise{An example of \opt{short}{TMI}\opt{full}{\algo}. (a) A social network. (b) Average relevance over all users in the whole social network. \opt{full}{(c) Average relevance over all users in $\tau_3$.}}}
    \label{Fig:fig_algoex_1}
\end{figure}

\setlength{\algomargin}{2.5ex}
\SetInd{3.5pt}{3.5pt}
\begin{procedure}[t]
\caption{selectNominees()}
\DontPrintSemicolon
\SetNlSty{texttt}{}{}
\SetAlgoNlRelativeSize{-1}
\SetNlSkip{0.25em}
\SetVlineSkip{1pt}
\SetKwProg{Procedure}{Procedure}{}{}
\Procedure{\selectNominees{$U,b$}}{
	$N \gets \emptyset$; $c \gets 0$\;
    $U \gets U \setminus \{(u,x)\mid c_{u,x} > b - c\}$\;
	\While{$U \not= \emptyset$}{
		$s \gets \argmax_{(u,x) \in U} \frac{f(N \cup \{(u,x)\}) - f(N)}{c_{u,x}}$\;
		$N \gets N \cup \{s\}$; $c \gets c + c_s$\;
		$U \gets U \setminus (\{s\} \cup \{(u,x)\mid c_{u,x} > b - c\})$
	}
	\textbf{return} $N$\;
}
\end{procedure}

\setlength{\algomargin}{2.5ex}
\SetInd{3.5pt}{3.5pt}
\begin{procedure}[t]
\caption{clusterNominees()}
\DontPrintSemicolon
\SetNlSty{texttt}{}{}
\SetAlgoNlRelativeSize{-1}
\SetNlSkip{0.25em}
\SetVlineSkip{1pt}
\SetKwProg{Procedure}{Procedure}{}{}
\Procedure{\clusterNominees{$N$}}{
	$\{\tau\} \gets $ clusters of $N$ according to social distances and relevance between promoting items\;
	\textbf{return} $\{\tau\}$\;
}
\end{procedure}

\setlength{\algomargin}{2.5ex}
\SetInd{3.5pt}{3.5pt}
\begin{procedure}[t]
\caption{prioritizeTargetMarket()}
\DontPrintSemicolon
\SetNlSty{texttt}{}{}
\SetAlgoNlRelativeSize{-1}
\SetNlSkip{0.25em}
\SetVlineSkip{1pt}
\SetKwProg{Procedure}{Procedure}{}{}
\Procedure{\prioritizeTargetMarket{$\{\tau\}$}}{
	$CG \gets \emptyset$\;
	\For{each pair $(\tau_i, \tau_j)$}{
	    $V^{\tau_i} \gets $ users in $\tau_i$\;
	    $V^{\tau_j} \gets $ users in $\tau_j$\;
		\If{$\left| V^{\tau_i} \cap V^{\tau_j} \right| > \theta$}{
			Put $\tau_i$ and $\tau_j$ in the same $\mathcal{G} \in CG$ \;
		}
	}
	\For{each $\mathcal{G} \in CG$}{
		\For{each $\tau_i \in \mathcal{G}$}{
		    $AE(\tau_i) = \sum_{x\in \tau_i,y\in \tau_j}{\bar{r}^{\mathtt{S}}_{x,y}}$, where $\tau_i, \tau_j \in \mathcal{G}, i \neq j$\;
        }
		Arrange the promoting order of $\tau \in \mathcal{G}$ according to $AE(\tau)$ ascendingly\;		
	}
	\textbf{return} $CG$\;
}
\end{procedure}

\begin{example}
\label{Ex:TMI}
Figs.~\ref{Fig:fig_algoex_1_a} and~\ref{Fig:fig_algoex_1_b} present an example of TMI with a social network and the average relevance over all users in the whole social network, where the dotted and dashed edges are the complementary and substitutable relationships, respectively. The number beside each edge is the relevance. Assume $N = \{ (u_1, $ iPad$),(u_2, $ AirPods$),$ $(u_4, $ iPhone$), (u_6, $ AirPods$),  (u_7, $ iPad$)\}$ by TMI according to MCP. Then, TMI finds three clusters \revise{$N^{\tau_1}=\{(u_1, $ iPad$)\}$, $N^{\tau_2}=\{(u_7, $ iPad$)\}$, and $N^{\tau_3}=\{ (u_2, $ AirPods$),$ $(u_4, $ iPhone$), (u_6, $ AirPods$)\}$} from Figs.~\ref{Fig:fig_algoex_1_a} and~\ref{Fig:fig_algoex_1_b}, and identifies $\tau_1$, $\tau_2$, and $\tau_3$ accordingly, as shown in Fig.~\ref{Fig:fig_algoex_1_a}. Assume $\theta=1$. Then, $\tau_1$, $\tau_2$, and $\tau_3$ belong to the same $\mathcal{G}$ since $\tau_1$ and $\tau_3$ have two common users, and $\tau_2$ and $\tau_3$ have two common users. After that, according to the substitutable relevance in Fig.~\ref{Fig:fig_algoex_1_b}, $AE(\tau_1)=0.5$ since iPad promoted in $\tau_1$ is substitutable to iPhone promoted in $\tau_3$. Similarly, $AE(\tau_2)=0.5$ and $AE(\tau_3)=0.5+0.5=1$. TMI thereby promotes $\tau_1$, $\tau_2$, and $\tau_3$ sequentially.
\exend
\end{example}

\textit{2. Dynamic Reachability Evaluation (DRE):}
For each target market $\tau_k \in \mathcal{G}$ selected by TMI, DRE evaluates \textit{Dynamic Reachability (DR)} of \revise{each item in $\tau_k$, and the nominees (in $N^{\tau_k}$)} promoting the item with the highest DR serve as the candidate seeds. \revise{\label{para:candidate_seeds}In other words, after TMI has identified target markets with socially close users to promote complementary items and has prioritized the target markets using AE, DRE allows each target market to prioritize its promoting items differently and lets the nominees promoting items with higher DR be the seeds earlier.\footnote{\revise{That the nominees are not in $N^{\tau_k}$ implies that they do not promote complementary items to the users in $\tau_k$ (based on the identification strategy in TMI). Moreover, if the noominees not in $N^{\tau_k}$ have not served as seeds yet when DRE finds the candidate seeds for $\tau_k$, they must belong to some target markets posterior to $\tau_k$. Accordingly, considering such nominees as candidate seeds for $\tau_k$ may imply to neglect the design of TMI and cause the antagonism of the substitutable relationship, which is expected to be avoided by TMI.}}}
Specifically, let $d^{\tau_k}$ denote the diameter of the target market $\tau_k$, and $S^{\mathcal{G}}$ is the seed group determined so far for all the target markets in $\mathcal{G}$. Let $I^{\tau_k}$ denote the items that have not yet been promoted in $\tau_k$. DR of an item $x \in I^{\tau_k}$ is  
\revise{
\opt{short}{\begin{align}
\label{eq:dr}
DR^{\tau_k}(S^{\mathcal{G}},x) = PI^{\tau_k}(S^{\mathcal{G}},x,d^{\tau_k}) + RI^{\tau_k}(S^{\mathcal{G}},x,d^{\tau_k}).
\end{align}}
\opt{full}{\begin{align}
\label{eq:dr}
DR^{\mathcal{W},\tau_k}(S^{\mathcal{G}},x) = PI^{\mathcal{W},\tau_k}(S^{\mathcal{G}},x,d^{\tau_k}) + RI^{w_x,\tau_k}(S^{\mathcal{G}},x,d^{\tau_k}).
\end{align}}}
The proactive impact \revise{\opt{short}{$PI^{\tau_k}(S^{\mathcal{G}},x,d^{\tau_k})$}\opt{full}{$PI^{\mathcal{W},\tau_k}(S^{\mathcal{G}},x,d^{\tau_k})$}} is the probability of $x$ to increase the preferences of users in $\tau_k$ \textit{for other items}. The reactive impact \revise{\opt{short}{$RI^{\tau_k}(S^{\mathcal{G}},x,d^{\tau_k})$}\opt{full}{$RI^{w_x,\tau_k}(S^{\mathcal{G}},x,d^{\tau_k})$}} is the probability to increase the preferences of users in $\tau_k$ \textit{for $x$} under the impact from other items in $S^{\mathcal{G}}$.\footnote{$d^{\tau_k}$ appears in PI and RI to restrict the item impact propagation to the users at most $d^{\tau_k}$ away in $\tau_k$.} The derivations of the proactive impact and the reactive impact are detailed in \opt{short}{\cite{online}}\opt{full}{Sec.~\ref{sec:extend_dr}} to capture the probability of increasing users' preferences for items based on their dynamic perceptions of item relationships.

Consequently, for each $\tau_k$ selected by TMI, DRE extracts the nominees \revise{\opt{short}{$\{(u,x_\text{p}) \mid x_\text{p} = \argmax\limits_{x \in I^{\tau_k}} {DR^{\tau_k}(S^{\mathcal{G}},x)},$ $(u,x_\text{p}) \in N^{\tau_k}\}$}\opt{full}{$\{(u,x_\text{p}) \mid x_\text{p} = \argmax\limits_{x \in I^{\tau_k}} {DR^{\mathcal{W},\tau_k}(S^{\mathcal{G}},x)},$ $(u,x_\text{p}) \in N^{\tau_k}\}$}} with the highest DR as the candidate seeds iteratively, and this property (i.e., the highest DR) is important to approximate the optimal solution in Theorem~\ref{the-sp3}.

\begin{figure}
\centering
    \subfigure[$t=0$]{
        \centering
        \includegraphics[width=0.15\textwidth]{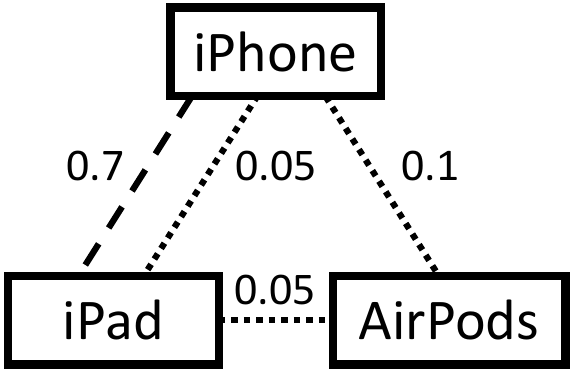}
        \label{Fig:fig_algoex_u5}
    }%
    \subfigure[$t=1$]{
        \centering
        \includegraphics[width=0.15\textwidth]{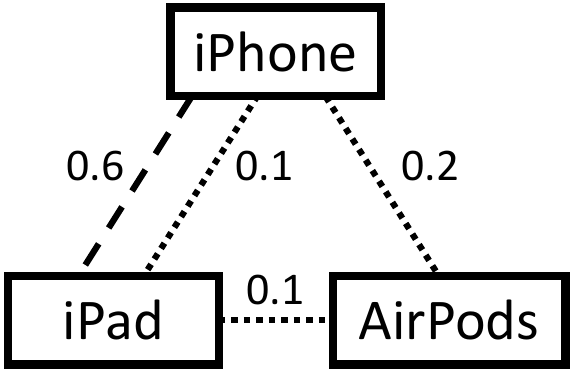}
        \label{Fig:fig_algoex_new_u5}
    }%
    \subfigure[$t=1$]{
        \centering
        \includegraphics[width=0.15\textwidth]{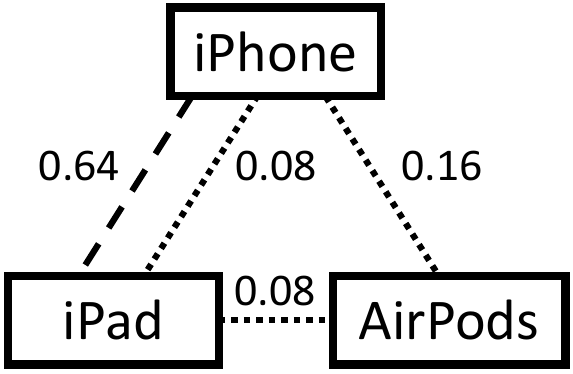}
        \label{Fig:fig_algoex_expected_u5}
    }\hfill
    \caption{\revise{An example of DRE: an illustration of the dynamics in $u_5$'s personal item network. (a)~Initially. (b)~After $u_5$ adopts iPad. (c) Expectation.}}
    \label{Fig:fig_algoex_dynamic_u5}
\end{figure}


\begin{example}
\label{Ex:DRE}
\revise{
Following Example~\ref{Ex:TMI}, this example shows how the DRE of \opt{short}{\salgo}\opt{full}{\algo} works based on the diffusion model to solve \opt{short}{\sproblem}\opt{full}{\problem}. Assume that the seed group becomes $S^\mathcal{G} = \{(u_1,$ iPad$,1)\}$ after $\tau_1$ is promoted. To update the complementary and substitutable relevance in each user's dynamic perception, \opt{short}{\salgo}\opt{full}{\algo} employs the Monte Carlo method to generate different cases of users' adoption decisions according to their preferences for items and the influence strength. For example, suppose that, due to the seed $u_1$'s promotion at $t = 1$, $u_2$ adopts iPad and promotes to $u_5$ by her social influence. For the case that $u_5$ adopts iPad, the update of her adopting item (referring Fig.~\ref{Fig:simple_factors}) in turn changes her personal item network from Fig.~\ref{Fig:fig_algoex_u5} to Fig.~\ref{Fig:fig_algoex_new_u5}. Meanwhile, $u_5$'s preference for AirPods increases due to the adopted item $\{\text{iPad}\}$ and her changed personal item network (Fig.~\ref{Fig:fig_algoex_new_u5}), where the complementary relevance between iPad and AirPods increases from 0.05 to 0.1. By contrast, for the case that $u_5$ does not adopt iPad at $t=1$, her personal item network and her preference for AirPods remains unchanged. Based on the number of times these cases being observed, \opt{short}{\salgo}\opt{full}{\algo} computes the expectation of $u_5$' personal item network, as shown in Fig.~\ref{Fig:fig_algoex_expected_u5}.

After that, assume that $S^\mathcal{G} = \{(u_1,$ iPad$,1), (u_7,$ iPad$,2)\}$ after $\tau_2$ is also promoted. \opt{short}{\salgo}\opt{full}{\algo} now concentrates on $\tau_3$, where $N^{\tau_3}=\{(u_2, $ AirPods$), (u_4, $ iPhone$), (u_6, $ AirPods$)\}$, $d^{\tau_3} = 3$, and $I^{\tau_3} = \{$iPhone$, $ AirPods$\}$ (due to the items not yet promoted by the nominees in $N^{\tau_3}$). \opt{short}{DRE calculates the DR for iPhone and AirPods, i.e., $DR^{\tau_3}(S^\mathcal{G},$ iPhone$)$ and $DR^{\tau_3}(S^\mathcal{G},$ AirPods$)$, respectively, according to the updated (same as above) personal item networks.\footnote{\revise{The examples with detailed explanations and calculations are presented in \cite{online}.}} Suppose that $DR^{\tau_3}(S^\mathcal{G},$ AirPods$) > DR^{\tau_3}(S^\mathcal{G},$ iPhone$)$. }
\opt{full}{Given $\mathcal{W}=\{w_\text{iPhone} = 1, w_\text{iPad} = 1, w_\text{AirPods} = 0.5\}$, DRE calculates the DR for iPhone and AirPods as $DR^{\mathcal{W},\tau_3}(S^\mathcal{G},$ iPhone$) = 0.7+1=1.7$ and $DR^{\mathcal{W},\tau_3}(S^\mathcal{G},$ AirPods$) = 1.2+0.85=2.05$, respectively, according to the updated (same as above) personal item networks.\footnote{Detailed calculations are presented in Example~\ref{ex:detailDRE}.} Therefore, }DRE \opt{short}{then }extracts $\{(u_2,$ AirPods$), (u_6,$ AirPods$)\}$ for promotion first.
\exend
}
\end{example}

\label{para:algo_end1}
\opt{short}{\revise{After a set of nominees $\{(u,x_\text{p})\}$ are extracted by DRE, \opt{short}{\salgo}\opt{full}{\algo} finds the promotional timing for them. Since the target markets have been arranged in a promoting order (by TMI) and the items with higher DR have been promoted with higher priority (by DRE), the promotional timing $t_\text{p}$ for $\{(u,x_\text{p})\}$ is assigned right after the latest promotion in $S^\mathcal{G}$ to ensure the influence spread of the seeds in $S^\mathcal{G}$ is not reduced. Let $\hat{t}$ denote the latest promotion in $S^\mathcal{G}$, i.e., $\hat{t} = \max\{t \mid (u,x,t) \in S^\mathcal{G}\}$. \opt{short}{\salgo}\opt{full}{\algo} assigns $\{(u,x_\text{p})\}$ in the promotion $t_\text{p} = \min \{\hat{t}+1, T\}$ to be the seeds, i.e., $S^\mathcal{G} \cup \{(u,x_\text{p},t_\text{p})\}$. Then, \opt{short}{\salgo}\opt{full}{\algo} selects the next item with DRE and determines the timing for the corresponding nominees.} After all nominees in $\tau_k$ are assigned their promotional timings as the seeds, TMI moves on to the next target market $\tau_{k+1} \in \mathcal{G}$. It returns the seed group $S=\bigcup_\mathcal{G} S^{\mathcal{G}}$ as the solution after all target markets are examined.\footnote{$S^\mathcal{G}$ of different $\mathcal{G}$ can be derived in parallel due to the independency of different $\mathcal{G}$.}}
\label{para:algo_end2}


\opt{full}{
\textit{3. Timing Determination by Substantial Influence (TDSI):}
After a set of nominees $N_\text{p}=\{(u,x_\text{p})\}$ are generated by DRE, TDSI iteratively extracts the nominee and finds a proper promotional timing $t$ with the largest \textit{substantial influence (SI)}, where SI is exploited to measure the impact of assigning a nominee at some timing under the impact of the seed group $S^\mathcal{G}$. In other words, with SI, TDSI can find some $(u,x_\text{p},t)$ (by assigning some nominee $(u,x_\text{p})$ at some timing $t$) that has the largest impact under $S^\mathcal{G}$. TDSI then adds it to $S^\mathcal{G}$ and looks for the next seed repeatedly until all nominees in $N_\text{p}$ are assigned some timings and added to $S^\mathcal{G}$.

Specifically, for each $(u,x_\text{p},t)$, SI represents its impacts for immediate and subsequent adoptions according to the marginal adoption $MA^{\tau_k}(S^\mathcal{G}, (u,x_\text{p},t))$ and marginal likelihood $ML^{\tau_k}(S^\mathcal{G}, (u,x_\text{p},t))$, respectively. The marginal adoption is the increment of the total adoptions after $(u,x_\text{p},t)$ is added into $S^\mathcal{G}$, whereas the marginal likelihood is the increment of the likelihood for not-yet-adopted items to be adopted in the future promotions if $(u,x_\text{p},t)$ is added into $S^\mathcal{G}$. Let $\tau_1, \tau_2, \ldots, \tau_{k-1} \in \mathcal{G}$ denote the target markets promoted prior to $\tau_k$. Note that the seeds in $\tau_i$ ($i<k$) have been determined since TMI has arranged the promoting order of target markets. Similar to DR, SI of $(u,x_\text{p},t)$ is also derived under the impact of the seed group $S^\mathcal{G}$ as follows.
\begin{align}
\label{eq:si}
&\;SI^{\tau_k}(S^\mathcal{G}, (u,x_\text{p},t),T)\\ \nonumber
= & \; MA^{\tau_k}(S^\mathcal{G}, (u,x_\text{p},t))+ \frac{T-t+1}{T} \cdot ML^{\tau_k}(S^\mathcal{G}, (u,x_\text{p},t)),
\end{align}
where the marginal likelihood is weighted by the ratio of the remaining number promotions to the total number promotions, as more remaining promotions indicate more chances for future adoptions.
To derive SI, TDSI simulates the influence diffusion with $S^\mathcal{G}$ and $S^\mathcal{G} \cup \{(u,x_\text{p},t)\}$, respectively, by the Monte Carlo method. MA is the difference of adoptions between whether to have $(u,x_\text{p},t)$ as a seed. Likewise, ML is the difference of likelihood for adopting not-yet-adopted items in the future between whether to have $(u,x_\text{p},t)$ as a seed. To estimate the likelihood, for each case simulated by the Monte Carlo method, TDSI examines the probability for every user in $\tau_k$ to adopt each not-yet-adopted item in the next promotion. Accordingly, TDSI estimates the expectation of likelihood. Detailed derivations of marginal adoption and marginal likelihood are presented in Sec.~\ref{sec:extend_si}. They capture the changes in social influence strength and personal preferences, which is crucial to improve the lower bound in Theorem~\ref{the-sp3}.

TDSI iteratively extracts the nominees, finds the time leading to the largest SI, and adds it to the seed group $S^{\mathcal{G}}$, i.e., $S^{\mathcal{G}}=S^{\mathcal{G}} \cup \{(u,x_\text{p},t)\}$, where $(u,t) = \argmax_{u',t'}{SI^{\tau_k}(S^\mathcal{G}, (u',x_\text{p}, t'), T)}$. Instead of exploring every possible $t$ to maximize Eq.~\eqref{eq:si} for each nominee, TDSI only needs to search $t \in [\hat{t}, \min \{\hat{t}+1, \sum_{i\leq k}T^{\tau_i}\}]$, where $\hat{t}$ is the latest promotion in $S^\mathcal{G}$ (i.e., $\hat{t} = \max\{t \mid (u,x,t) \in S^\mathcal{G}\}$), $T^{\tau_k}$ is the promotional duration of $\tau_k$, proportional to the number of nominees in $\tau_k$ (i.e, $T^{\tau_k} = \frac{\left| N^{\tau_k}\right| \cdot T}{\sum_{\tau_i}{\left| N^{\tau_i}\right|}}$, where $N^{\tau_i}$ is the set of nominees in $\tau_i$ and $\tau_i, \tau_k \in \mathcal{G}$).
The limitation is satisfactory for the following reasons. On one hand, since 1) the target markets have been arranged in a promoting order (by TMI), 2) the items with higher DR have been promoted with priority (by DRE), and 3) the nominees with higher SI have become seeds (by TDSI), $t$ is required to be no earlier than the seeds in $S^\mathcal{G}$, i.e., $t \geq \hat{t}$. On the other hand, since the importance-aware influence function is non-monotonically increasing (proved in Sec.~\ref{sec:theoretical}), it is not necessary to explore the later timings for nominees if the seeds of any previous promotion have not been examined entirely. Thus, evaluating only the timing right after $\hat{t}$ is sufficient, because later promotional timings only reduce the extent of the marginal likelihood considered in SI. Moreover, since promoting items of different target markets in $\mathcal{G}$ are substitutable to each other (which will not be promoted simultaneously), each market $\tau_i$ is allowed a promotional duration $T^{\tau_i}$, and $\hat{t}+1$ cannot exceed the last promotional timing of $\tau_k$, i.e., $\sum_{i\leq k}T^{\tau_i}$. Thus, the search of $t$ is upper bounded by $\min \{\hat{t}+1, \sum_{i\leq k}T^{\tau_i} \}$. 

\begin{example}
Following Example~\ref{Ex:DRE}, $S^{\mathcal{G}} = \{(u_1,$ iPad$,1), (u_7,$ iPad$,2)\}$ and $(u_2,$ AirPods$)$ and $(u_6,$ AirPods$)$ are extracted by DRE. Given $T=5$, $T^{\tau_1} = 1$, $T^{\tau_2} = 1$, and $T^{\tau_3} = 3$, since the latest promotional timing is $\hat{t} = 2$ in $S^\mathcal{G}$, TDSI searches the promotional timings for the nominees in $[2, \min\{2+1, 1+1+3\}]=[2,3]$. Regarding the two nominees, TDSI computes $SI^{\tau_3}(S^\mathcal{G},(u_2,$ AirPods$,2),5)=1.2$, $SI^{\tau_3}(S^\mathcal{G},(u_2,$ AirPods$,3),5)=1.1$, $SI^{\tau_3}(S^\mathcal{G},(u_6,$ AirPods$,2),5)=2.3$, and $SI^{\tau_3}(S^\mathcal{G},(u_6,$ AirPods$,3),5)=1.9$. TDSI thus updates $S^\mathcal{G}=\{(u_1,$ iPad$,1),(u_7,$ iPad$,2),(u_6,$ AirPods$,2)\}$. Next, as $S^\mathcal{G}$ is updated, TDSI also updates the search of promotional timings in $[3, \min\{3+1,1+1+3\}]=[3,4]$ and computes $SI^{\tau_3}(S^\mathcal{G},(u_2,$ AirPods$,3),5)=1.1$ and $SI^{\tau_3}(S^\mathcal{G},(u_2,$ AirPods$,4),5)=1$ regarding the nominee $(u_2,$ AirPods$)$. $S^\mathcal{G}=\{(u_1,$ iPad$,1),(u_7,$ iPad$,2),(u_6,$ AirPods$,2), (u_2,$ AirPods$,3)\}$ are updated accordingly. Consequently, as both nominees have become seeds, TDSI stops. 
\exend
\end{example}

After deriving the promotional timing for a nominee $(u,x_\text{p})$, \algo\ selects the next item with DRE and finds the timing with TDSI. After all nominees in $\tau_k$ are assigned their promotional timings as the seeds, TMI moves on to the next target market $\tau_{k+1} \in \mathcal{G}$. It returns the seed group $S=\bigcup_\mathcal{G} S^{\mathcal{G}}$ as the solution after all target markets are examined.\footnote{$S^\mathcal{G}$ of different $\mathcal{G}$ can be derived in parallel due to the independency of different $\mathcal{G}$.}
}

\opt{short}{
\begin{theorem}\label{the-sp3}
\revise{\salgo\ is a $(1-\frac{1}{\sqrt{e}}-\epsilon)(\min \{ P^{c}_\text{minpref}, P^{c}_\text{minext}\})$ approximation algorithm for \sproblem} in $O(M\left|V\right|\left|I\right|k_{\max})$ time, where $P_\text{minpref}>0$ and $P_\text{minext}>0$ are the minimum preference and extra adoption probability, respectively. $c$ is the maximum hop of influence propagation, $M$ is the time to evaluate $\sigma$ depending on the evaluation error $\epsilon>0$,\footnote{Note that the technique of reverse influence sampling cannot support multiple promotions since the dependency among different promotions makes positive propagation irreversible.} and $k_{\max}$ is the maximum size of a feasible solution.
\end{theorem}
\revise{
\begin{proofsk}
\label{proofsk2}
\input{4_revision_algorithm_proofsketch}
\end{proofsk}
}
}
\subsection{Theoretical Results}
\label{sec:theoretical}

In the following, we first present several good properties of the importance-aware influence function $\sigma$ under some conditions.
\begin{definition}[Submodular function \cite{Lov1983}]\label{def1.1}
Given a ground set $U$, a set function $\rho: 2^U\mapsto \mathbb{R}$ is submodular if for any subset $X\subseteq Y$ and any element $e\in U\setminus Y$, 
\begin{equation}\label{eq-submodular}
  \rho(Y\cup \{e\})-\rho(Y)\leq \rho(X\cup \{e\})-\rho(X).
\end{equation}
\end{definition}

\begin{lemma}\label{lem-submodular}
The importance-aware influence function $\sigma$ is non-monotone increasing but submodular 
for $P_\text{pref}$, $P_\text{act}$, and $P_\text{ext}$ assigned at the beginning of all promotions.
\end{lemma}
\begin{proof}
Firstly, we first analyze the case with a single promotion. 
It is obvious that the importance-aware influence function $\sigma$ is monotone increasing even if we add a seed that promotes an item with importance $0$ and no longer promote any other user to adopt items. That is, $\sigma$ does not decrease.
For each user-item pair, the preference $P_\text{pref}(u,y)$ is regarded as the probability of the edges from user-item pair $(u',y)$ to user-item pair $(u,y)$, where $u'$ denotes $u$'s in-neighbors. Likewise, the extra adoption probability $P_\text{ext}(u,u',x,y)$ is also regarded as the probability of the edge from user-item pair $(u',x)$ to user-item pair $(u,y)$.
When $P_\text{pref}$, $P_\text{act}$ and $P_\text{ext}$ do not change, all the events of evaluating preference $P_\text{pref}$ and influence strength $P_\text{act}$, and extra adoption $P_\text{ext}$ happen independently. 
Then, the proof of the submodularity of $\sigma$ can be reduced to the proof of its submodularity in every arbitrary realization (i.e. a deterministic graph) of the stochastic graph.
Therefore, $\sigma$ is a coverage function, which is submodular \cite{Krause2012}. 

Our proof can be generalized to the triggering model easily by sampling the corresponding in-neighbor edge set of each user independently. For better understanding of the process of influence propagation, we first show the proof for the IC-based models. Following \cite{Kempe2003}, the proof of the submodularity of importance-aware influence function on the IC model is reduced to that in every arbitrary realization of the stochastic graph. Specifically, the realization is determined by flipping a coin on every edge independently according to its probability.
We show the submodularity of $\sigma(S_t)$ in a deterministic realized graph $G_{\text{SN}}'(V,E')$ by flipping the edge coin $P_\text{act}(u',u)$ of $G_{\text{SN}}$ with the probabilities of influence strength and the edge coin $P_\text{ext}(u,u',x,y)$ with the extra adoption probabilities independently first. 
If the edge coin $P_\text{act}(u',u)$ of $G_{\text{SN}}$ is up, i.e., the event of influence happens, then 
we flip the edge coin $P_\text{pref}(u,y)$ with the probabilities of preference. 
Then, The process of adoption propagation 
can be regarded as the influence propagation process. Upon the deterministic graph $G_{\text{SN}}'$, 
a user $u$ can be successfully influenced by its in-neighbor $u'$ who has adopted $x$, if and only if it does not adopt the item $x$ or the relevant item $y$ of $x$ (i.e. the item $y$ with $P_\text{ext}(u,u',x,y)>0$ has not been purchased). Notice that the two events happen independently.
Thus, for any nominee $(u,x)$, its influence propagation becomes a connected subgraph of $G_{\text{SN}}'$ rooted by $u$, which is regarded as the influence of seeding $u$ to promote item $x$. Note that each path of the connected subgraph may consist of different items due to the happening of the two events (e.g. $P_{x_1x_2x_2x_4\ldots}$), and each user in the path may adopt more than one item.   
For each seed group $S_t$ in the $t$-th promotion,
a realization of $\sigma(S_t)$ is exactly the union of the connected subgraphs that combine the influence produced by each nominee $(u,x)$ together. 

Next, we analyze the case of multiple promotions.
Even when $P_\text{pref}$, $P_\text{act}$ and $P_\text{ext}$ are static, there are some cases showing that the monotone increasing property of the importance-aware influence function $\sigma$ does not hold. This is because if a user $u$ successfully adopts item $x$, it cannot adopt item $x$ again in the later promotions, i.e., it cannot propagate influence of item $x$ again. 
As a consequence, putting a seed to promote $u$ to adopt item $x$ early may result in the influence of $u$ adopting $x$ not propagating in later promotions to increase the adoption of other users.
Now we show the non-monotone increasing property by the following observation.

Assume there are exactly two candidate items $y$ and $x$. We let $w_x=1$, $w_y=0$, $c_{u,x}>b$, $c_{u,y}>b$, $c_{v,x}>b$ for any $v\neq u$ (i.e., only $(v,y)$ can be nominees for any $v\neq u$), 
$P_\text{pref}(v,y)\equiv 1$, $P_\text{ext}(v,y,x)=P_\text{ext}(v,x,y)\equiv 0$ 
\footnote{The simplified symbol $P_\text{ext}(v,y,x)\equiv 0$ represents that it holds for arbitrary in-neighbors of $v$.}, 
$P_\text{pref}(v,x)\equiv 0$ for any $v\neq u$ (i.e., $v$ can only adopt item $y$ for any $v\neq u$), $P_\text{pref}(u,x)\equiv 0$, and $P_\text{ext}(u,y,x))\equiv 1$ (i.e., $u$ can only adopt item $x$ after its in-neighbor adopts item $y$). As a consequence, the total influence focuses on computing the probability of user $u$ adopting item $x$.
Assume there are exactly two promotions of influence propagation. 
Given a seed group $S$, we denote the probability of user $u$ adopting item $x$ in the first and second promotions as $Pr_1$ and $Pr_2$, respectively, according to the seed group $S$. Note that $\sigma(S)=Pr_1+(1-Pr_1)Pr_2$ and $Pr_2$ is conditioned on $Pr_1$.
If we add any nominee $(v,y)$ to the first promotion, the probability of $u$ adopting $x$ in the first promotion becomes $Pr'_1$, where $Pr'_1\geq Pr_1$. However, the probability of $u$ adopting $x$ in the second promotion becomes $Pr'_2$, which is at most $Pr_2$. This is because some nodes on the existing promoting paths, which make $u$ adopt $x$, in the second promotion cannot propagate influence since they have been promoted by $(v,y)$ in the first promotion, and the new nodes promoted the adopted items by $(v,y)$ cannot improve the probability of $u$ adopting $x$ in the second promotion. 
It implies $\sigma(S)=Pr_1+(1-Pr_1)Pr_2$ is larger than $\sigma(S\cup \{(v, y,1)\})=Pr'_1+(1-Pr'_1)Pr'_2$ under the condition of $(1-Pr_1)Pr_2-(1-Pr'_1)Pr'_2 > Pr'_1-Pr_1$. By letting $Pr_1=0$ (e.g. no seed in the first promotion), the condition becomes $Pr_2-(1-Pr'_1)Pr'_2 > Pr'_1$. If a sufficiently large $Pr_2$ (e.g. almost $1$) is achieved, then choosing $Pr'_2$ slightly smaller than $1$ is sufficient to satisfy the condition. 
The illustrative example is shown in Figure~\ref{counter-example}.
\begin{figure}
            \centering
            \includegraphics[height=1.7in,width=0.38\linewidth]{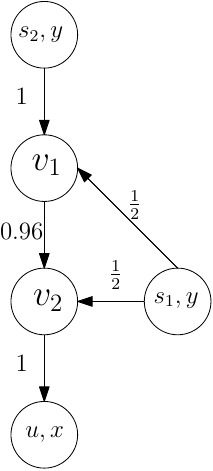}
    \caption{\small{Illustrative example of non-monotone increasing of $\sigma$, where $(s_1,y,1)$ and $(s_2,y,2)$ are seeds, and the weight of each edge is the influence strength.}}
   
    \label{counter-example}
\end{figure}
Note that $Pr_1=0$, $Pr'_1=0.74$, $Pr_2=0.96$ and $Pr'_2=0.48$ imply $\sigma(\{(s_2,y,2)\})=Pr_1+(1-Pr_1)Pr_2=0.96>\sigma(\{(s_1,y,1), (s_2,y,2)\})=Pr'_1+(1-Pr'_1)Pr'_2=0.74+(1-0.74)\times 0.48=0.8648$, which is not monotone increasing.

The realization of $\sigma(S)$ is not the union of total adopted items in each promotion directly.
However, 
the submodularity can still be proved in a realized deterministic graph, which is a specific disjoint union of each graph corresponding to a different promotion, and each graph is similar to the one in single promotion. Hence, The importance-aware influence function is also a coverage function, which is submodular \cite{Krause2012}.
Specifically,
in the first promotion, we first flip 
the edge coin $P_\text{act}(u',u)$ of $G_{\text{SN}}$ with the probabilities of influence strength and the edge coin $P_\text{ext}(u,u',x,y)$ with the extra adoption probabilities independently. If the edge coin $P_\text{act}(u',u)$ of $G_{SN}$ is up, then we flip the edge coin $P_\text{pref}(u,y)$ with the probabilities of preference independently. Afterward, we obtain the connected subgraphs of the deterministic graph rooted at each seed of $S_1$. 
After the first promotion, the promoted nodes in this promotion cannot be promoted the adopted items again, and it propagates the corresponding the influence later. 
Therefore, after flipping edge coins for the second promotion again, we obtain those connected subgraphs of the deterministic graph rooted at each seed of $S_2$ to avoid reaching each promoted user-item pair $(u,x)$ in the connected subgraphs rooted at each seed of $S_1$.\footnote{A user may have several candidate user-item pairs and the promoting path can reach the user-item pair which is not promoted in the first promotion through the living edges.} 
In addition, if $S_2$ has some nominees in the first promotion, these nominees can still try to promote their neighbors in the second promotion since they are chosen as new seeds again.
Similarly, we repeat the same process until all the random edges over the $T$ promotions are realized.
Notice that we flip edge coins of all the $T$ promotions before seeding since the probability of each edge in the later promotions is independent of that in the previous promotions. After the above realization, let $\bar \sigma(S_t \mid S_0,S_1,\ldots,S_{t-1})$ denote the influence of seeding $S_t$ conditioned on seeding $S_{t'}$ for $t'=0,1,\ldots,t-1$, where $S_0=\emptyset$.
Let $\bar \sigma(S)$ denote the importance-aware influence of a realization of $\sigma(S)$.
Then, $\bar \sigma(S)=\mathop{\cup} \limits_{t=1}^{T} \bar \sigma(S_t \mid S_0,S_1,\ldots,S_{t-1})$.
To prove the submodularity, it is sufficient to show that the function $\bar \sigma$ under each realization satisfies Inequality~\eqref{eq-submodular} based on the property that a user cannot be promoted again by another path, if it is promoted by some realized path rooted at some nominee in an early promotion. 
For any seed group set $S^X\subseteq S^Y$ and any seed $(u,x,t)\in U\setminus S^Y$, where $U$ consists of all the possible seeds, it is sufficient to prove that the marginal decreased influence of seed $(u,x,t)$ conditioned on the larger seed group is more than the smaller one. It is equivalent to show that all the users, who can be promoted by $(u,x,t)$ but have been promoted by $S^X$ earlier than $(u,x,t)$, can also be promoted by $S^Y$ earlier than $(u,x,t)$ to adopt the same item. Since $S^Y$ contains $S^X$, for each $(v,y,t')\in S^X$, it can either block $(u,x,t)$ similar to those in $S^Y$ or be replaced by some seed $(v',y',t'')\in S^Y\setminus S^X$ that does not have smaller influence than $(v,y,t')$ in $S^X$ (i.e. $(v',y',t'')$ may promote users earlier or promote more users). Therefore, 
\begin{equation}
  \bar \sigma (S^Y\cup \{(u,x,t)\})- \bar \sigma(S^Y)\leq \bar \sigma(S^X\cup \{(u,x,t)\})-\bar \sigma(S^X).
\end{equation}
The lemma follows. 
\end{proof}

Note that the general importance-aware influence function is not submodular. The reason is that, in some case, some users' $P_\text{pref}$ or $P_\text{ext}$ in an early promotion can be increased by seeds, but these seeds cannot bring the immediate adoption in this promotion. For a later promotion, however, the seeds can bring more immediate adoptions due to the increased $P_\text{pref}$ or $P_\text{ext}$. As a consequence, the sum of immediate adoptions under the influence of respective seeds
in each promotion is smaller than the overall adoptions under the influence of all seeds in all promotions.
It makes the importance-aware influence function non-submodular, i.e, violating the submodular Inequality~\eqref{eq-submodular}.

According to Lemma \ref{lem-submodular}, based on the approximation algorithm for the \textit{non-monotone increasing submodular maximization with knapsack constraint} problem (SMK), we have the following theorem.
\begin{theorem}\label{the-G-SMK} 
For $P_\text{pref}$, $P_\text{act}$, and $P_\text{ext}$ assigned at the beginning of all promotions, there is an $\alpha-\epsilon$--approximation algorithm for \problem, where $\alpha$ is the performance ratio of the approximation algorithm for SMK and
$\epsilon>0$ is the evaluation error.
\end{theorem}
Although there is an improved $\frac{1}{e}-\epsilon$ --approximation algorithm for SMK by continuous greedy strategy \cite{Feldman2011} (implying $\frac{1}{e}-\epsilon$--approximation for \problem\ for the current user preferences, social influence, and extra adoption probabilities), it takes much more time than $O(n^5)$ (roughly $O(n^8)$).
Due to the high complexity of the above approximation algorithm for SMK, we design a fast and simple approximation algorithm within $O(n^2)$ function calls, where $n$ is the size of the ground set. 
So far, for SMK, there is a $\frac{1}{6}$--approximation algorithm \cite{Gupta2010Constrained} taking $O(n^5)$ by enumeration of at most three potential solutions with the greedy strategy, where $n$ is the size of the ground set. Inspired by the work \cite{Gupta2010Constrained}, following a similar frame, our algorithm does not involve the enumeration, and the time complexity can be reduced to $O(n^2)$.
More specifically, there are two steps to achieve the goal. First, instead of finding a solution satisfying the condition of Lemma \ref{lem-greedy} combining the enumeration and greedy strategy, we find the desired solution satisfying the condition by TMI. 
TMI iteratively adds an element with the maximum marginal cost-performance ratio (i.e., MCP) into the solution until the budget is just violated or negative marginal influence occurs. 
Second, to achieve the desired approximation, following a similar approach in \cite{Gupta2010Constrained} by calling the greedy process with $O(n^2)$ function calls and a linear time algorithm \cite{Buchbinder2015A} for unconstrained submodular maximization problem (USM), we can obtain a solution with a good approximation ratio. However, the achieved solution may be not feasible. We find the feasible solution by deleting the violating element to obtain a new candidate solution. Then, the algorithm finds another seed group consisting of a single element that has the maximum importance-aware influence, and chooses the better one between them.
To analyze the performance ratio, we first derive the following two key lemmas. 
In the following, for better comparison with \cite{Gupta2010Constrained} and understanding, we still use $f$ as the importance-aware influence function of SMK.
\begin{lemma}[Lemma 2.2 in \cite{Gupta2010Constrained} ]\label{lem-19-8-5-1}
Given a submodular function, sets $C$, $S_1 \subset U$, let $C' = C \setminus S_1$, and $S_2 \subset  U \setminus S_1$.
Then $f(S_1 \cup C) + f(S_1 \cap C) + f(S_2 \cup C') \geq f(C)$.\footnote{For detailed definition of $S_1$ and $S_2$, please refer to Theorem 2.3 in \cite{Gupta2010Constrained}.}
\end{lemma}
The following lemma shows that TMI can find a solution with a good approximation ratio.
\begin{lemma}\label{lem-greedy}
For any feasible set $C$, TMI returns a set $S$ (just violating the budget constraint) disjointing with $C$ that satisfies $f(S) \geq \frac{1}{2}f(S \cup C)$ within $O(n^2)$ calls of function $f$, where $n$ is the size of the ground set.
\end{lemma}
\begin{proof}
Case 1: TMI stops when just violating the budget constraint.

In this case, no negative marginal influence occurs in TMI.
Let $S=\{e_1,e_2,\ldots,e_l\}$ be the output, where the elements are numbered in the order added into $S$ in TMI. 
Note that $S$ just violates the budget constraint after element $e_l$ is added, i.e, $c(S)=\sum_{e \in S}{c_e}>b$. We let $C=\{e'_1,e'_2,\ldots,e'_k\}$ and denote $S[:i]=\{e_1,e_2,\ldots,e_i\}$ and $C[:j]=\{e'_1,e'_2,\ldots,e'_j\}$ for $i=1,2,\ldots,l$ and $j=1,2,\ldots,k$, respectively. Note that $S[:l]=S$.
On one hand,
for any $i<j$, since
\begin{align*}
&\frac{f(S[:i])-f(S[:i-1])}{c_{e_i}} \\ 
& \hspace{1cm} \geq \frac{f(S[:i-1]\cup \{e_j\})-f(S[:1],\ldots,S[:i-1])}{c_{e_j}}    
\end{align*}
by TMI and $f(S[:i-1]\cup \{e_j\})-f(S[:i-1])\geq f(S[:j-1]\cup \{e_j\})-f(S[:j-1])$ by submodularity, we have  $$\frac{f(S[:i])-f(S[:i-1])}{c_{e_i}} \geq \frac{f(S[:j])-f(S[:j-1])}{c_{e_j}}.$$ Thus 
\begin{equation}\label{eq-19-8-7-1}
    \frac{f(S[:i])-f(S[:i-1])}{c_{e_i}} \geq \frac{f(S[:l])-f(S[:l-1])}{c_{e_l}}
\end{equation}
for any $i<l$.

On the other hand, by TMI,
\begin{align}
& \frac{f(S[:l])-f(S[:l-1])}{c_{e_l}} \nonumber \\
&\hspace{1cm}\geq \frac{f(S[:l-1]\cup \{e'_j\})-f(S[:l-1])}{c_{e'_j}}     \nonumber
\end{align}
for any $j=1,2,\ldots,k$.
Furthermore, by submodularity,
\begin{align}
& f(S[:l-1]\cup \{e'_j\})-f(S[:l-1]) \nonumber \\ 
& \hspace{1cm} \geq f(S[:l]\cup C[:j])-f(S[:l]\cup C[:j-1])     \nonumber
\end{align}
for any $j=1,2,\ldots,k$.
Therefore, 
\begin{align}
&\frac{f(S[:l])-f(S[:l-1])}{c_{e_l}} \nonumber \\
&\hspace{1cm}\geq \frac{f(S[:l]\cup C_j)-f(S[:l]\cup  C[:j-1])}{c_{e'_j}}     \nonumber
\end{align}
for any $j=1,2,\ldots,k$.
By Inequality~\eqref{eq-19-8-7-1}, we have 
\begin{align}\label{eq-19-8-7-2}
    &\frac{f(S[:i])-f(S[:i-1])}{c_{e_i}} \\
    &\hspace{0.5mm} \geq \frac{f(S[:l]\cup C[:j])-f(S[:l]\cup C[:j-1])}{c_{e'_j}}
\end{align}
for any $i=1,2,\ldots,l$ and $j=1,2,\ldots,k$.
After summing up the numerator and the denominator separately of the left-hand-side with $l$ terms and the right-hand-side with $k$ terms in Inequality~\eqref{eq-19-8-7-2}, respectively,
we have\footnote{In fact, even if the negative marginal influence occurs, we can continue to run the algorithm until the budget is just violated. Although the marginal influence may be negative, Inequality~\eqref{eq-19-8-7-3} still holds only if the corresponding denominator (i.e. cost function) is more than zero, implying the lemma still holds.}
\begin{equation}\label{eq-19-8-7-3}
    \frac{f(S[:l])}{\sum_{i=1}^l c_{e_i}} \geq \frac{f(S[:l]\cup C[:k])-f(S[:l])}{\sum_{j=1}^k c_{e'_j}}.
\end{equation}
Moreover, since $\sum_{i=1}^l c_{e_i}> b \geq \sum_{j=1}^k c_{e'_j}$ due to the feasibility of $C$, by Inequality~\eqref{eq-19-8-7-3}, 
$f(S[:l])\geq f(S[:l]\cup C[:k])-f(S[:l])$.
Hence $f(S[:l])\geq \frac{f(S[:l]\cup C[:k])}{2}$. Recall that $C[:k]=C$ and $S[:l]=S$, we have $f(S)\geq \frac{f(S\cup C)}{2}.$ Furthermore, note that TMI needs at most $O(n^2)$ function calls. 

Case 2: TMI stops when the negative marginal influence occurs.

In this case, assume $f(S[:i]\cup \{e_i\})-f(S[:i-1])<0$ just after $e_i$ is added. By Inequality~\eqref{eq-19-8-7-2}, we have $f(S[:l]\cup C[:j])-f(S[:l]\cup C[:j-1]) \leq f(S[:i])-f(S[:i-1])\leq 0$ for any $j=1,2,\ldots,k$. Thus, $f(S[:l]\cup C[:k])-f(S[:l])=\sum_{j=1}^k (f(S[:l]\cup C[:j])-f(S[:l]\cup C[:j-1]))<0$, implying $f(S[:l]\cup C[:k])-f(S[:l])<0$ achieving a strong result.
Following both cases, the lemma follows.

\end{proof}
Based on Lemma \ref{lem-19-8-5-1} and \ref{lem-greedy}, we have the following theorem.
\begin{theorem}[adapted from Theorem 2.3 in \cite{Gupta2010Constrained}]\label{the-SMK}
There is a $\frac{1}{12}$--approximation algorithm for SMK within $O(n^2)$ calls of function $f$.
\end{theorem}
\begin{proof}
Let $S^*$ be the optimal solution with $f(S^*)=opt$. Note that $S^*\setminus S_1$ and $S^*\setminus S_1\setminus S_2$ are both feasible solutions disjointing with $S_1$ and $S_2$, respectively. By Lemma \ref{lem-greedy}, $f(S_1)\geq \frac{1}{2}f(S_1\cup (S^*\setminus S_1))=\frac{1}{2}f(S_1\cup S^*)$ and $f(S_2)\geq \frac{1}{2}f(S_2\cup (S^*\setminus S_1\setminus S_2))=\frac{1}{2}f(S_2\cup (S^*\setminus S_1))$, respectively. If $f(S_1\cap S^*)\geq c \cdot opt$, then the $\frac{1}{2}$-approximate algorithm of USM on the ground set $S_1$ in linear time \cite{Buchbinder2015A} implies the approximation ratio is at least $\frac{c}{2} opt$. 
Otherwise,
$$f(S_1)\geq  \frac{1}{2}f(S_1\cup S^*)\geq  \frac{1}{2}f(S_1\cup S^*)+ \frac{1}{2}f(S_1\cap S^*)-\frac{c}{2} opt.$$
Combing with $f(S_2)\geq \frac{1}{2}f(S_1\cup (S^*\setminus S_1))$, we have
\begin{align*}
    2\max\{f(S_1),f(S_2)\} & \geq f(S_1)+f(S_2) \\
    & \geq  \frac{1}{2}(f(S_1\cup S^*)+ f(S_1\cap S^*)\\
    & +f(S_1\cup (S^*\setminus S_1)))-\frac{c}{2} opt \\
    & \geq  \frac{1}{2}f(S^*)-\frac{c}{2}opt \\
    & = \frac{1-c}{2}opt,
\end{align*}
where the last inequality comes from Lemma \ref{lem-19-8-5-1}. Thus, 
$$\max\{f(S_1),f(S_2)\}\geq \frac{1}{4}(1-c)opt.$$ Combining with the case $f(S_1\cap S^{*})\geq \frac{c}{2} opt$, the approximation ratio becomes
$\min\{\frac{c}{2},\frac{1-c}{4}\}opt$. By setting $c=\frac{1}{3}$, we achieve a $\frac{1}{6}$--approximation solution, denoted by $\bar S$. 
However, $\bar S$ may be infeasible
since all of $S_1$, $S_2$, and $S_1\cap S^*$ may just violate the budget constraint by the same element $e_1$ added into both $S_1$ and $S_1\cap S^*$, and $e_2$ added into $S_2$ defined in Lemma \ref{lem-greedy}.

In the following, we construct a feasible solution with the guaranteed performance ratio.
We denote $e_{\max}=\argmax_{e}f(\{e\})$ and $\hat S=\argmax_{\bar S\setminus \{e\},\{e_{\max}\}}\{f(\bar S\setminus \{e\}),f(\{e_{\max}\})\}$ as the output of the algorithm, where $e$ is the corresponding element of $\bar S$ that violates the budget constraint.
Consequently,
\begin{align*}
  f(\hat S) & =\max\{f(\bar S\setminus \{e\}),f(\{e_{\max}\})\} \\ 
  & \geq \max\{f(\bar S\setminus \{e\}),f(\{e\})\} \\
  & \geq \frac{1}{2}(f(\bar S\setminus \{e\})+f(\{e\})) \\
  & \geq \frac{1}{2}(f(\bar S\setminus \{e\} \cup \{e\})+f(\emptyset))\\
  & =\frac{1}{2}f(\bar S)\geq \frac{1}{12}opt, 
\end{align*}
where the third inequality is based on the submodularity.
The theorem follows.
\end{proof}

Afterwards, we apply the proposed algorithm to the special case. 
Note that $f$ becomes the importance-aware influence function $\sigma$ with $\left|V\right|\left|I\right|T$ variables, for seeding some users and choosing some items to be promoted at some promotional timings. Let $k_{\max}$ denote the maximum size of a feasible solution. Since we study the strategy between the multiple promotions and the single promotion under the same budget, we assume $k_{\max}$ of the multiple promotions and the single promotion is the same, which is bounded by $\left|V\right|\left|I\right|$ without the budget constraint. Notice that $k_{\max}$ can be computed in linear time by iteratively adding the minimum cost seed until the budget constraint is violated. 
\begin{theorem}\label{the3}
For $P_\text{pref}$, $P_\text{act}$, and $P_\text{ext}$ assigned at the beginning of all promotions,
there is a $\frac{1}{12}-\epsilon$--approximation algorithm in
$O(M\left|V\right|\left|I\right|T\cdot k_{\max})$ time for \problem, 
where $M$ is the time to evaluate $\sigma$ depending on the evaluation error $\epsilon>0$, and $k_{\max}$ is the maximum size of a feasible solution.
\end{theorem}
\begin{proof}
By Theorem \ref{the-G-SMK} and \ref{the-SMK}, the approximation ratio follows.
We analyze the time complexity as follows. Let $M$ be the time to evaluate $\sigma$ by Monte Carlo sample depending on the evaluation error $\epsilon>0$. 
By lemma \ref{lem-greedy}, it needs at most $O(n^2)$ function calls in TMI. For \problem, $n=\left|V\right|\left|I\right|T$ and it needs at most $k_{\max}$ iterations to use all the budget, i.e., stop the algorithm. Thus, the number of function calls is $O(\left|V\right|\left|I\right|Tk_{\max})$ in TMI. Note that TMI brings the main function calls. Thus, the total number of function calls of the designed algorithm becomes $O(\left|V\right|\left|I\right|Tk_{\max})$, implying the time complexity is $O(M\left|V\right|\left|I\right|Tk_{\max})$.
\end{proof}

Next, let us analyze the performance ratio of \algo\ in the following more general case.
\begin{theorem}\label{the-sp3}
\algo\ is a 
$(1-\frac{1}{\sqrt{e}}-\epsilon)(\min \{ P^{c}_\text{minpref}\cdot P^{c}_\text{minact}, P^{c}_\text{minext}\})$ approximation algorithm for \problem\ in $O(M\left|V\right|\left|I\right|k_{\max})$ time, where 
$P_\text{minpref}>0$, $P_\text{minext}>0$ and $P_\text{minact}>0$ are the minimum preference, extra adoption probability and influence strength, 
respectively. 
$c$ is the maximum hop of influence propagation,
$M$ is the time to evaluate $\sigma$ depending on the evaluation error $\epsilon>0$, and $k_{\max}$ is the maximum size of a feasible solution.
\end{theorem}

\begin{proof}

Since the importance-aware influence function of the general problem is neither monotone increasing nor submodular, we first find the relation between the general problem and a restricted problem (defined later) whose importance-aware influence function is monotone increasing and submodular. Consequently, a solution found by the algorithm in the restricted problem can be regarded as an approximation solution to the general problem. 
We first prove that the output value of \algo\ is at least $\beta-\epsilon$ (defined later) times of the value of the restricted problem . 
Afterwards, according to the evaluation of the range of parameters changing, we derive the importance-aware influence function gap between the restricted problem and original problem is at most $\gamma$ (defined later), implying the total gap is $(\beta-\epsilon)\gamma$.

Specifically, let $\gamma=\min \{ P^{c}_\text{minpref}\cdot P^{c}_\text{minact}, P^{c}_\text{minext}\}$.
Let $\hat {S}_t$ and ${S^*}_t$ denote the seeds in the $t$-th promotion of the optimal solution of the restricted problem (i.e., $P_\text{pref}\equiv P_\text{minpref}$, $P_\text{ext}\equiv P_\text{minext}$, $P_\text{act}\equiv P_\text{minact}$, and $c$-hop influence propagation) and the optimal solution of the general problem, respectively. 
First, we construct an auxiliary solution $\bar N$ such that $\sigma(\bar N)=\max\{\sigma(N_\text{first}), \sigma(\{e_{\max}\})\}$, where $N_\text{first}$ consists of all nominees in $N$ promoting items in the first promotion, i.e., $N_\text{first}=\{(u,x,1) \mid (u,x) \in N\}$, and $e_{\max}$ is the seed bringing the maximum expected importance-aware influence. Note that $\sigma(\bigcup_\mathcal{G}S^\mathcal{G})\geq \sigma(\bar N)$ since the solution $\bar N$ has been considered by \algo.
Hence, our goal is to prove
$\sigma(\bar N)\geq (\beta-\epsilon)\gamma \sigma(S^*)$ sufficiently by two stages, where $\beta$ will be decided later.
Before presenting the formal proof, we first consider the case that each seed user can adopt each item once in all the promotions.
Notice that if $P_\text{pref}$, $P_\text{ext}$ and $P_\text{act}$ do not change, it is not necessary to assign the seeds in different promotions; it may decrease the probability of adoption for the multiple promotions since the promoted users in early promotions cannot propagate the influence in later promotions, where these users may become the bridges to promote the items with larger importance to others. 
Thus, if $P_\text{pref}$, $P_\text{ext}$ and $P_\text{act}$ do not change in the restricted problem, there exists an optimal solution for the restricted problem such that only the first promotion has seeds (i.e., $\hat {S}_t =\emptyset$ for $t>1$).

For the first stage, when $P_\text{pref}$, $P_\text{act}$ and $P_\text{ext}$ are static, by the above observation, we can focus on the single promotion. Moreover, by the proof of Lemma \ref{lem-submodular}, the importance-aware influence function is monotone increasing submodular. Consequently, if there exists a $\beta$--approximation algorithm for \textit{monotone increasing submodular maximization with knapsack constraint} problem (MSMK), then
$\sigma(\bar N)\geq (\beta-\epsilon) \sigma(\hat {S}_1)$ if $P_\text{pref}\equiv P_\text{minpref}$, $P_\text{ext}\equiv P_\text{minext}$ and $P_\text{act}\equiv P_\text{minact}$ do not change for the single promotion, where $\epsilon$ is the evaluation error of Monte Carlo sampling. Furthermore, when $P_\text{pref}$, $P_\text{ext}$ and $P_\text{act}$ can change, $\sigma(\bar N)$ is no smaller than $\sigma(\bar N)$ with $P_\text{pref}\equiv P_\text{minpref}$, $P_\text{ext}\equiv P_\text{minext}$ and $P_\text{act}\equiv P_\text{minact}$ unchanged, since higher $P_\text{pref}$, $P_\text{ext}$ and $P_\text{act}$ lead to more influence for a single promotion. Therefore, if $P_\text{pref}$, $P_\text{ext}$ and $P_\text{act}$ can change, $\sigma(\bar N)$ is at least $(\beta-\epsilon) \sigma(\hat {S}_1)$ with $P_\text{pref}\equiv P_\text{minpref}$, $P_\text{ext}\equiv P_\text{minext}$ and $P_\text{act}\equiv P_\text{minact}$ unchanged.
For the second stage, it is sufficient to prove when $P_\text{pref}\equiv P_\text{minpref}$, $P_\text{ext}\equiv P_\text{minext}$ and $P_\text{act}\equiv P_\text{minact}$ do not change, $\sigma(\hat {S}_1)$ is at least $\gamma \sigma(S^*)$ with varying $P_\text{pref}$, $P_\text{ext}$ and $P_\text{act}$. 
Note that, if $P_\text{pref}\equiv P_\text{minpref}$, $P_\text{ext}\equiv P_\text{minext}$ and $P_\text{act}\equiv P_\text{minact}$ do not change,
$\sigma(\hat {S}_1)\geq  \sigma(S^*_\text{first})$, where $S^*_\text{first}$ consists of all nominees of $S^*$ and assigns them in the first promotion, i.e., $ S^*_\text{first} = \{(u,x,1) \mid (u,x,t) \in S^*\}$, and it is a feasible solution of the restricted problem.
Moreover,
$\sigma(S^*_\text{first})$ with $P_\text{pref}\equiv 1$, $P_\text{ext}\equiv 1$ and $P_\text{act}\equiv 1$ is no smaller than $\sigma({S^*})$ with varying $P_\text{pref}$, $P_\text{ext}$ and $P_\text{act}$, since $P_\text{pref}$, $P_\text{ext}$ and $P_\text{act}$ can be increased to at most $1$ for the multiple promotions.
Therefore, we only need to prove that, $\sigma(S^*_\text{first})$ is no smaller than $\gamma \sigma(S^*_\text{first})$ with $P_\text{pref}\equiv 1$, $P_\text{ext}\equiv 1$, and $P_\text{act}\equiv 1$, if $P_\text{pref}\equiv P_\text{minpref}$, $P_\text{ext}\equiv P_\text{minext}$ and $P_\text{act}\equiv P_\text{minact}$ do not change.
To prove it, for each promoted user-item pair, it can be promoted by the promoting paths from different influence propagation hop $j$ ($j\in \{ 1,2,\ldots,c\}$). 
Then, we denote $X_j$ as the probability that all the promoting paths succeed to promote the user-item pair in the $j$-th hop with varying $P_\text{pref}$, $P_\text{ext}$ and $P_\text{act}$.
For each influence propagation hop $j$, when $P_\text{pref}\equiv P_\text{minpref}$, $P_\text{ext}\equiv P_\text{minext}$ and $P_\text{act}\equiv P_\text{minact}$ do not change, the total adoption probability of paths in hop $j$ is at least $\min \{P_\text{minpref}^{j}\cdot P^{j}_\text{minact}, P_\text{minext}^{j}\}X_j$.
Thus, the total adoption probability of paths with $P_\text{pref}\equiv P_\text{minpref}$, $P_\text{ext}\equiv P_\text{minext}$ and $P_\text{act}\equiv P_\text{minact}$ unchanged is at least $\min \{ P_\text{minpreft}^{c}\cdot P^{c}_\text{minact}, P_\text{minext}^{c}\}$ of that with $P_\text{pref}\equiv 1$, $P_\text{ext}\equiv 1$ and $P_\text{act}\equiv 1$. Thus, $\sigma(S^*_\text{first})$ is no smaller than $\gamma \sigma(S^*_\text{first})$ with $P_\text{pref}\equiv 1$, $P_\text{ext}\equiv 1$ and $P_\text{act}\equiv 1$, if $P_\text{pref}\equiv P_\text{minpref}$, $P_\text{ext}\equiv P_\text{minext}$ and $P_\text{act}\equiv P_\text{minact}$ do not change.
As mentioned above, the approximation ratio of the case depends on the approximation $\beta$ of MSMK. In particular, \algo\ can guarantee the performance ratio for MSMK by Nominee Selection and checking all possible single candidate seed in a subroutine of \algo, i.e., the algorithm in \cite{krause2008note} which has $\frac{1-e^{-1}}{2}$ approximation \cite{nguyen2013budgeted} for MSMK. In fact, this approximation can be further improved to $1-\frac{1}{\sqrt{e}}$ by an analysis similar to that in \cite{khuller1999budgeted} for MSMK. Furthermore, since the improvement of other subroutines in \algo\ does not worsen the result, the approximation holds, implying $\beta=1-\frac{1}{\sqrt{e}}$. 

Last, we analyze the time complexity.
Recall that $M$ is the time to evaluate $\sigma$ by Monte Carlo sample depending on the evaluation error $\epsilon>0$. It is sufficient to show the total number of $\sigma$ function calls. Note that Nominee Selection and TDSI take the main function calls. 
In Nominee Selection, it involves at most $\left|V\right|\left|I\right|k_{\max}$ function calls. 
In TDSI, the primary cost is finding the promotional timing for each nominee. If there are at most $\left|\bar I\right|$ items chosen in Nominee Selection, $\left|\bar I\right|$ items are distributed to at most $k_{\max}$ nominees chosen in Nominee Selection. We assume the $i$-th item is associated with $k_i$ users, where $\sum_{i}k_i\leq k_{\max}$. For the $i$-th item, it takes at most $k^2_i$ function calls since TDSI searches at most two promotional timings for each nominee associated with the item. Therefore, TDSI involves at most $(k^2_1+k^2_2+,\ldots,k^2_{\left|\bar I\right|})\leq (\sum_{i}k_i)^2\leq k^2_{\max}\leq \left|V\right|\left|I\right|k_{\max}$.
Thus, the total number function calls of \algo\ is $O(\left|V\right|\left|I\right|k_{\max})$, implying the time complexity is $O(M\left|V\right|\left|I\right|k_{\max})$.
The theorem follows.
\end{proof}

\section{Detailed Derivation for the Dynamics}
\label{sec:extend}

\subsection{Update Mechanism of Main Factors in \problem}
\label{sec:extend_problem}

The details of deriving and updating each factor are elaborated as follows.

(1) Relevance measurement. KG and meta-graphs with personal weightings are exploited to find the personal relevance between two items \cite{shi2019semrec,gu2019relevance,Huang2016Meta,zhang2016collaborative}. KG \cite{gu2019relevance,zhang2016collaborative} is a heterogeneous information network $G_\text{KG}=(\mathcal{V},\mathcal{E},\Phi,\Psi)$ with node and edge sets ($\mathcal{V}$ and $\mathcal{E}$) and two type-mapping functions ($\Phi$ and $\Psi$), e.g., $\Phi(\textsf{iPhone})=\textsc{item}$ and $\Psi((\textsf{iPhone},\textsf{Bluetooth}))=\textsc{support}$. A meta-graph is a schema represented as a graph of node types with certain connections (e.g., Fig.~\ref{Fig:new_ex_meta}) specified by users \cite{Huang2016Meta} or learned from historical data of users \cite{gu2019relevance}. For a meta-graph $m$, its instances are the subgraphs of KG that exactly match the schema of $m$. The relevance between item $x$ and item $y$ as defined by meta-graph $m$, i.e., $s(x,y \mid m) \in [0,1]$, is correlated to the number of $m$'s instances with end items $x$ and $y$ (e.g., SCSE \cite{Huang2016Meta}). For example, in Fig.~\ref{Fig:new_ex_meta}, $s(x,y \mid m_1)$ is correlated to the number of common \textsc{feature}s supported by $x$ and $y$ in KG.

As each relationship can be captured by various meta-graphs with different connections between \textsc{item}s, we realize the personal item networks using the meta-graphs with personal weightings \cite{shi2019semrec}. Formally, let $W_\text{meta}(u,m,\zeta_t)$ denote a user $u$'s weighting on a meta-graph $m$ at the end of step $\zeta_t$, i.e., the significance of $m$ to $u$ for describing an item relationship. According to \cite{mcauley2015inferring,zhao2017improving,shi2019semrec,gu2019relevance}, adopted items usually change users' perceptions of item relationships. Therefore, we update $u$'s weightings on meta-graphs (e.g., by SemRec and RelSUE \cite{shi2019semrec,gu2019relevance}) after all adoption decisions at step $\zeta_t$ are made. Then, the complementary and substitutable relevance between $x$ and $y$ in $u$'s perception, denoted as $r^{\mathtt{C}}(u,x,y,\zeta_t)$ and $r^{\mathtt{S}}(u,x,y,\zeta_t)$ respectively, are immediately derived (e.g., by SemRec \cite{shi2019semrec}) according to $s(x,y \mid m^{\mathtt{C}})$ and $s(x,y \mid m^{\mathtt{S}})$ of every meta-graph $m^{\mathtt{C}}$ and $m^{\mathtt{S}}$, respectively.\footnote{In this paper, we aim to study IM with dynamic personal perceptions of item relationships modeled by KG and meta-graphs. The computation of relevance is not our focus. Interested readers are referred to previous studies, e.g., \cite{shi2019semrec,gu2019relevance}.}  Consequently, $u$'s personal item network is deduced as $G_\text{PIN}(u,\zeta_t) = (V^\textsc{item}, E^{\mathtt{C}}, E^{\mathtt{S}})$, where $V^\textsc{item}$ is the item set, and an edge $(x,y) \in E^{\mathtt{C}}$ (or $E^{\mathtt{S}}$) exists and carries the complementary (or substitutable) relevance between $x$ and $y$ if $r^{\mathtt{C}}(u,x,y,\zeta_t) >0$ (or $r^{\mathtt{S}}(u,x,y,\zeta_t) >0$). When $u$ and $\zeta_t$ are clear from context, we write $r^{\mathtt{C}}_{x,y}$ and $r^{\mathtt{S}}_{x,y}$ for short.

(2) Preference estimation. Following \cite{zhao2017improving,banerjee2019maximizing,xin2019relational}, $u$'s preference for $x$ depends on her adopted items and personal item network. Thus, after all adoption decisions are made and the personal item network is changed, the preferences for not-yet-adopted items are updated to reflect users' interests in those items. Previous works derive users' preferences for not-yet-adopted items by learning the embedding of users and items (e.g., RSC and RCF \cite{zhao2017improving,xin2019relational}) or by statistics inference  \cite{banerjee2019maximizing}. Thus, let $A(u,\zeta_t)$ be the set of items that $u$ has adopted so far after all adoption decisions are made at step $\zeta_t$. Equipped with $A(u,\zeta_t)$ and $G_\text{PIN}(u,\zeta_t)$, we follow the above research to derive $u$'s rating on each not-yet-adopted item $y$ as $u$'s updated preference for $y$ at $\zeta_t$, i.e., $P_\text{pref}(u,y,\zeta_t)$.

(3) Influence learning. The influence strength is correlated to the similarity of users, i.e., their personal item networks and their adopted items, since similar users are inclined to become closer and easier to influence each other. Therefore, after users' personal item networks and adopted items update, the influence strength changes according to statistic models \cite{zhang2019learning} and user embedding (e.g., DeepInf and DANSER \cite{qiu2018deepinf,wu2019dual}). Following the above works, after the sets of adopted items $A(u,\zeta_t)$ and $A(v,\zeta_t)$ are obtained and personal item networks $G_\text{PIN}(u,\zeta_t)$ and $G_\text{PIN}(v,\zeta_t)$ are updated, we update the influence strength from $u$ to $v$ at step $\zeta_t$ as $P_\text{act}(u,v,\zeta_t)$.

(4) Item associations. 
When users are promoted an item $x$ by social influence, an extra adoption of a relevant item $y$ is triggered due to item associations according to the probability of $u$ being promoted and preferring $x$ as well as the relationships and relevance between $x$ and $y$, by learning and comparing the embedding of $u$, $x$, and $y$ (e.g., CKE, RSC, and RCF \cite{zhao2017improving,zhang2016collaborative,xin2019relational}). As a result, when $u$ is promoted an item $x$ by $u'$, at step $\zeta_t$, $u$'s extra adoption probability for $x$'s relevant item $y$ is $P_\text{ext}(u,u',x,y,\zeta_t)$, derived according to $P_\text{act}(u',u,\zeta_t -1)$, $P_\text{pref}(u,x,\zeta_t -1)$, and $u$'s personal item network $G_\text{PIN}(u,\zeta_t -1)$.

\subsection{Derivation of DR in DRE of \algo}
\label{sec:extend_dr}

As formulated in Eq.~\eqref{eq:dr}, for an item $x$ to be promoted in $\tau_k \in S^{\mathcal{G}}$, its DR consists of $x$'s proactive impact (PI) and reactive impact (RI). The former is the probability of increasing users' preferences for other items by $x$, while the latter is the probability of increasing users' preferences for $x$ by other items. 
The adoption of $x$ increases (decreases) the preferences for the items complementary (substitutable) to $x$ \cite{frank1991microeconomics}. Given $S^{\mathcal{G}}$, the likelihood of regarding $x$ and $y$ as complementary (substitutable) for each user is proportional to the complementary (substitutable) relevance between $x$ and $y$, i.e., $\mathcal{L}^{\mathtt{C},\tau_k}(x,y,S^\mathcal{G}) = \frac{\bar{r}^{\mathtt{C}}_{x,y}}{\bar{r}^{\mathtt{C}}_{x,y}+\bar{r}^{\mathtt{S}}_{x,y}}$ and $\mathcal{L}^{\mathtt{S},\tau_k}(x,y,S^\mathcal{G}) = \frac{\bar{r}^{\mathtt{S}}_{x,y}}{\bar{r}^{\mathtt{C}}_{x,y}+\bar{r}^{\mathtt{S}}_{x,y}}$, where $\bar{r}^{\mathtt{C}}_{x,y}$ and $\bar{r}^{\mathtt{S}}_{x,y}$ are the average complementary and substitutable relevance between $x$ and $y$ over all users in $\tau_k$ after the promotion of $S^{\mathcal{G}}$, respectively.\footnote{The update of relevance is described in Sec.~\ref{sec:extend_problem}.} Therefore, PI is recursively formulated as follows. 
\begin{align}
\label{eq:fi}
&PI^{\mathcal{W},\tau_k}(S^\mathcal{G},x,d) \\ \nonumber
&= \sum_{y} \Big(
\mathcal{L}^{\mathtt{C},\tau_k}(x,y,S^\mathcal{G})  \bar{r}^{\mathtt{C}}_{x,y}   w_y
- \mathcal{L}^{\mathtt{S},\tau_k}(x,y,S^\mathcal{G})  \bar{r}^{\mathtt{S}}_{x,y}  w_y\\ \nonumber
& \hspace{3cm} 
+ PI^{\mathcal{W},\tau_k}(S^\mathcal{G},y,d-1)\Big),
\end{align}
where $y$ represents each item relevant to $x$, and $\mathcal{W}$ is the set of item importance. The first two terms are the likelihood to increase and decrease the preferences of the users in $\tau_k$ \textit{for $y$} (weighted by the corresponding relevance between $x$ and $y$ and \textit{$y$'s item importance}). The last term $PI^{\mathcal{W},\tau_k}(S^\mathcal{G},y,d-1)$ recursively captures the likelihood to increase or decrease the preferences (of users in $\tau_k$) \textit{for other items} via item impact propagation from $y$, where $PI^{\mathcal{W},\tau_k}(S^\mathcal{G},y,0)=0$.\footnote{Here it is $d-1$ because item impact has propagated 1-hop from $x$ to $y$. }

Similarly, RI evaluates the item impact propagation from any other item $z$ to $x$ according to $\bar{r}^{\mathtt{C}}_{z,x}$ and $\bar{r}^{\mathtt{S}}_{z,x}$ as follows. 
\begin{align}
\label{eq:bi}
& RI^{w_x,\tau_k}(S^{\mathcal{G}},x,d) \\ \nonumber
&= \sum_{z}\Big(
\mathcal{L}^{\mathtt{C},\tau_k}(z,x,S^\mathcal{G})  \bar{r}^{\mathtt{C}}_{z,x}   w_x
- \mathcal{L}^{\mathtt{S},\tau_k}(z,x,S^\mathcal{G})  \bar{r}^{\mathtt{S}}_{z,x}  w_x
\\ \nonumber
& \hspace{3cm} 
+ RI^{w_x,\tau_k}(S^{\mathcal{G}},z,d-1)\Big), \\ \nonumber
\end{align}
where $z$ is each item relevant to $x$ and $RI^{w_x,\tau_k}(S^\mathcal{G},y,0)=0$. Note that RI is derived from $w_x$ (rather than $\mathcal{W}$ in PI), since RI only focuses on the preferences for $x$.

\begin{example}
\label{ex:detailDRE}
For Example~\ref{Ex:DRE}, we show the detailed calculation of the DR of iPhone as follows.
\begin{align*}
&\;DR^{\mathcal{W},\tau_3}(S^\mathcal{G}, \text{iPhone}) \\
= & \; \big(\frac{0.2 \cdot 0.2 \cdot 1}{0.4+0.2} - \frac{0.4 \cdot 0.4 \cdot 1}{0.4+0.2} + PI^{\mathcal{W},\tau_3}(S^\mathcal{G},\text{iPad},2)\big) \\
& \;\;\;\;\; + \big(1 \cdot 0.4 \cdot 0.5 + PI^{\mathcal{W},\tau_3}(S^\mathcal{G},\text{AirPods},2)\big) \\
& \;\;\;\;\; +\big(\frac{0.2 \cdot 0.2 \cdot 1}{0.4+0.2}-\frac{0.4 \cdot 0.4 \cdot 1}{0.4+0.2}+RI^{w_\text{iPhone},\tau_3}(S^\mathcal{G},\text{iPad},2)\big) \\
& \;\;\;\;\; + \big(1 \cdot 0.4 \cdot 1+RI^{w_\text{iPhone},\tau_3}(S^\mathcal{G}, \text{AirPods},2)\big) \\
= & \; 0.7+1=1.7.
\hspace{6cm}\blacksquare
\end{align*}
\end{example}

\subsection{Derivation of SI in TDSI of \algo}
\label{sec:extend_si}

As formulated in Eq.~\eqref{eq:si}, SI of a candidate seed $(u,x_\text{p},t)$ is measured by the marginal adoption (MA) and the marginal likelihood (ML). The marginal adoption of $(u,x_\text{p},t)$ after the promotion of $S^\mathcal{G}$ is
\begin{align}
MA^{\tau_k}(S^\mathcal{G}, (u,x_\text{p},t)) = \sigma^{\tau_k}(S^\mathcal{G} \cup \{(u,x_\text{p},t)\}) - \sigma^{\tau_k}(S^\mathcal{G}),
\end{align}
where $\sigma^{\tau_k}(S^\mathcal{G})$ is the importance-aware influence in $\tau_k$ under $S^\mathcal{G}$. The marginal likelihood of $(u,x_\text{p},t)$ after the promotion of $S^\mathcal{G}$ is
\begin{align}
ML^{\tau_k}(S^\mathcal{G}, (u,x_\text{p},t)) = \pi^{\tau_k}(S^\mathcal{G} \cup \{(u,x_\text{p},t)\}) - \pi^{\tau_k}(S^\mathcal{G}),
\end{align}
where $\pi^{\tau_k}(S^\mathcal{G})$ is the likelihood of all users in $\tau_k$ to adopt the rest of not-yet-adopted items in the future under $S^\mathcal{G}$. As the probability for a user to adopt an item depends on the social influence strength she receives and her preference for the item, $\pi^{\tau_k}(S^\mathcal{G})$ is derived from the sum of probabilities of all $u \in \tau_k$ to adopt their not-yet-adopted items $y$.
\begin{align}
\pi^{\tau_k}(S^\mathcal{G}) = \sum_{v \in \tau_k} \sum_{y \notin A(v,\zeta_{\hat{t}}^\text{last})} AIS(v,y,\zeta_{\hat{t}}^\text{last}) \cdot P_\text{pref}(v,y,\zeta_{\hat{t}}^\text{last}),
\end{align}
where $\zeta_{\hat{t}}^\text{last}$ is the last step of the latest promotion $\hat{t}$ in $S^\mathcal{G}$ (i.e., $\hat{t} = \max\{t \mid (u,x,t) \in S^\mathcal{G}\}$).  $A(v,\zeta_{\hat{t}}^\text{last})$ is the set of items adopted by $v$ so far after the promotions of $S^\mathcal{G}$. $AIS(v,y,\zeta_{\hat{t}}^\text{last})$ is the aggregated influence probability for $y$ to be promoted to $v$ in the next promotion after the promotions of $S^\mathcal{G}$.\footnote{For example, under  the IC, $AIS(v,y,\zeta_{\hat{t}}^\text{last}) = 1-\prod_{v' \in N^\text{in}(v) \wedge y \notin A(v',\zeta_{\hat{t}}^\text{last})} \big(1 - P_\text{act}(v',v,\zeta_{\hat{t}}^\text{last})\big)$, while under the LT, $AIS(v,y,\zeta_{\hat{t}}^\text{last}) = \sum_{v' \in N^\text{in}(v) \wedge y \in A(v',\zeta_{\hat{t}}^\text{last})} P_\text{act}(v',v,\zeta_{\hat{t}}^\text{last})$, where $N^\text{in}(v)$ is the set of $v$'s in-neighbors.} $P_\text{pref}(v,y,\zeta_{\hat{t}}^\text{last})$ is $v$'s preference for $y$ after the promotions of $S^\mathcal{G}$.\footnote{The updates are described in Sec.~\ref{sec:extend_problem}.}

\opt{full}{
\subsection{\algo\ for Adaptive IM}
\label{sec:adaptive}
For the adaptive IM without a predefined budget allocation, \algo\ carefully determines the current promotion budget by avoiding the antagonism of the substitutable relationship. Specifically, to find $S_t$ after the propagation of the $(t-1)$-th promotion is observed, for $t < T$, TMI is modified by selecting only one nominee with the largest MCP at a time. Accordingly, \algo\ exploits TMI multiple times to gradually find nominees and target markets until the identified overlapping target markets promote substitutable items. Then, \algo\ rejects the latest identified nominee that causes the antagonism of the substitutable relationship. Next, \algo\ exploits DRE and TDSI alternatively to determine whether the identified nominees for each $\mathcal{G}$ are suitable for the $t$-th promotion. Note that the search for possible promotional timings in TDSI is limited to $t$ and $t+1$. Once the candidate seed with the largest SI is assigned to $t+1$, \algo\ ends the searching of $S^\mathcal{G}_t$ since the remaining nominees in $\mathcal{S}$ are suitable for later promotions as well. After all $\mathcal{G}$ are examined, $S_t = \bigcup_{\mathcal{G}} S^\mathcal{G}_t$ is found as the seeds for the $t$-th promotion. 
On the other hand, for $t = T$, \algo\ exploits TMI to select the best nominees under the remaining budget and assigns them to $T$ as $S_T$.
}

\section{Experiments}

\subsection{Experiment Setup}
\label{sec:exp_setup}

The experiment includes four datasets, where each one consists of a KG and a social network:\footnote{The KGs are heterogeneous information networks (HINs) in the datasets, where the HINs contain diverse node types like items, categories, brands, etc. 
\opt{full}{Please refer to \cite{hung2016social,liu2014exploiting,Zhao2017Metagraph} for more details.}}
i)~\textit{Douban} \cite{hung2016social}, ii)~\textit{Gowalla}\opt{short}{,\footnote{\url{https://www.yongliu.org/datasets}.}}\opt{full}{ \cite{liu2014exploiting},} iii)~\textit{Yelp}\opt{short}{,\footnote{\url{https://www.yelp.com/dataset}.}}\opt{full}{ \cite{Zhao2017Metagraph},} and iv)~\textit{Amazon}\opt{short}{.\footnote{\url{https://jmcauley.ucsd.edu/data/amazon}.}}\opt{full}{ \cite{Zhao2017Metagraph}.}
Since there are no social relationships in \textit{Amazon}, we supplement it with Pokec\footnote{\url{https://snap.stanford.edu/data/soc-Pokec.html}.} 
according to the user profiles\opt{full}{ \cite{wang2015prediction}}. To capture the complementary and substitutable relationships between items, the meta-graphs are generated according to \cite{mcauley2015inferring}, and the relevance of a certain relationship regarding a meta-graph is derived according to \cite{Huang2016Meta}. For the diffusion models, \opt{short}{the two factors, relevance measurement (including the learning of personal weightings on meta-graphs and the constructions of personal item networks) and preference estimation are learned and updated based on \cite{shi2019semrec} and \cite{zhao2017improving}, respectively.}\opt{full}{the four factors, relevance measurement (including the learning of personal weightings on meta-graphs and the constructions of personal item networks), preference estimation, influence learning and item associations are learned and updated based on \cite{shi2019semrec}, \cite{zhao2017improving}, \cite{zhang2019learning}, and \cite{zhao2017improving}, respectively. To set up the \opt{short}{\sproblem}\opt{full}{\problem} problem, the item importance of \textit{Douban}, \textit{Yelp}, and \textit{Amazon} is distributed following the prices on their websites, while that of \textit{Gowalla} is randomly assigned (since its website is no longer online).} \revise{The statistics of the datasets are listed in Table~\ref{T:dataset}.} Following \opt{short}{\cite{Nguyen2016Cost}}\opt{full}{\cite{Nguyen2016Cost,aslay2017revenue}}, the costs of hiring users to promote items are set proportional to users' out-degree and their preferences for items, since users who are more influential and who prefer the item less may need more incentive to be seeds. In the implementation of \algo, we exploit the submodularity \opt{full}{(similar to CELF++ \cite{goyal2011celf++}) }to speed up the nominee selection, and follow \cite{chen2017people} and \cite{chen2010scalable} to cluster nominees and explore influenced users, respectively, in TMI.

\begin{table}[t]
    \caption{\revise{The statistics of datasets.}}
    \label{T:dataset}
    \centering
    \begin{tabular}{|c||c|c|c|c|}
        \hline
        Dataset & \textit{Douban} & \textit{Gowalla} & \textit{Yelp} & \textit{Amazon} \\ \hline
        \# of node types & 3 & 3 & 6 & 6 \\
        \# of nodes & 7.6M & 3.2M & 251K & 260K \\
        \# of users & 5.5M & 407K & 17K & 1.6M \\
        \# of items & 2.1M & 2.8M & 22K & 20K \\
        \# of edge types & 3 & 3 & 6 & 6 \\
        \# of edges & 100M & 42M & 1.6M & 1.4M \\
        \# of friendships & 86M & 4.4M & 140K & 30.6M \\
        Directed friendship? & No & No & No & Yes \\
        Avg. initial influence strength & 0.011 & 0.092 & 0.121 & 0.050 \\
        \opt{full}{Avg. item importance & 2.1 & 0.5 & 1.6 & 1.8 \\ }
        \hline
    \end{tabular}
\end{table}

\revise{We compare \opt{short}{\salgo}\opt{full}{\algo} with OPT (derived from a brute-force approach) and four state-of-the-art approaches: BGRD \cite{banerjee2019maximizing}, HAG \cite{hung2016social}, PS \cite{teng2018revenue}, and DRHGA \cite{huang2020competitive} as the baselines.\footnote{Codes and datasets are available on \url{https://tinyurl.com/y26fx2mp}.} We extend \cite{huang2020competitive,teng2018revenue,hung2016social,banerjee2019maximizing} to consider different costs of selecting a user to promote an item by selecting from the user-item pairs or the users that satisfy the remaining budget.
\label{para:baselines}Furthermore, since they cannot be directly applied to our problem, we augment \cite{huang2020competitive,teng2018revenue,hung2016social,banerjee2019maximizing} with CR-Greedy \cite{sun2018multi} to support multiple promotions and determine the promotion timings of the user-item pairs as the seeds in each baseline.} The performance metrics include the 1) influence spread $\sigma$ (Def.~\ref{def:inf})\revise{\footnote{\revise{As \problem\ is an optimization problem (instead of a learning problem) to maximize the influence spread, we follow previous Influence Maximization (IM) research \cite{Kempe2003,Nguyen2016Cost,chen2010scalable,guo2020influence} to compare different algorithms in the experiment, by deriving the influence spread according to the diffusion models (instead of learning the influence spreads) when evaluating different sets of seeds.}}} and 2) execution time. We perform a series of sensitivity tests in terms of the budget $b$ and the number of promotions $T$. \revise{To verify our algorithm, we further conduct an empirical study on course promotion in viral marketing for the course selection system. \opt{short}{The complete experiments for \problem\ (including the case study on \textit{Amazon}) are shown in \cite{online} due to the space constraint. }}We conduct all experiments on an HP DL580 server with an Intel 2.10GHz CPU and 1TB RAM.
Each simulation result is averaged over 100 samples (i.e., $M=100$).

\subsection{Performance Comparison}
\label{sec:exp_comparison}

\label{para:opt}
First, we compare all approaches and OPT on small datasets sampled from \textit{Amazon} with 100 users. 
Fig.~\ref{Fig:opt_b} shows the \opt{full}{importance-aware }influence under different budgets. \revise{\opt{short}{\salgo}\opt{full}{\algo} has the closest performance to OPT, and outperforms BGRD, HAG, PS, and DRHGA, because TMI of \opt{short}{\salgo}\opt{full}{\algo} carefully selects influential nominees by MCP, and DRE of \opt{short}{\salgo}\opt{full}{\algo} then prioritizes nominees based on dynamic perceptions of item relationships.} In contrast, the baselines neglect the changes in item relationships and do not promote items beneficial to each other over time. Fig.~\ref{Fig:opt_T} compares the \opt{full}{importance-aware }influence under various numbers of promotions. \opt{short}{\salgo}\opt{full}{\algo} creates a larger influence spread as $T$ increases because TMI avoids promoting substitutable items to the same users in near promotions\opt{short}{.}\opt{full}{, and TDSI finds the promotional timings by carefully evaluating the subsequent adoptions.} All baselines do not incorporate the item impact propagation to achieve a larger influence spread as $T$ grows even a sophisticated algorithm based on CR-Greedy \cite{sun2018multi} is employed to schedule promotions at different timings.

\begin{figure}[t]
    \centering
    \subfigure[Diff. budgets.]{
        \centering
        \includegraphics[width=0.235\textwidth]{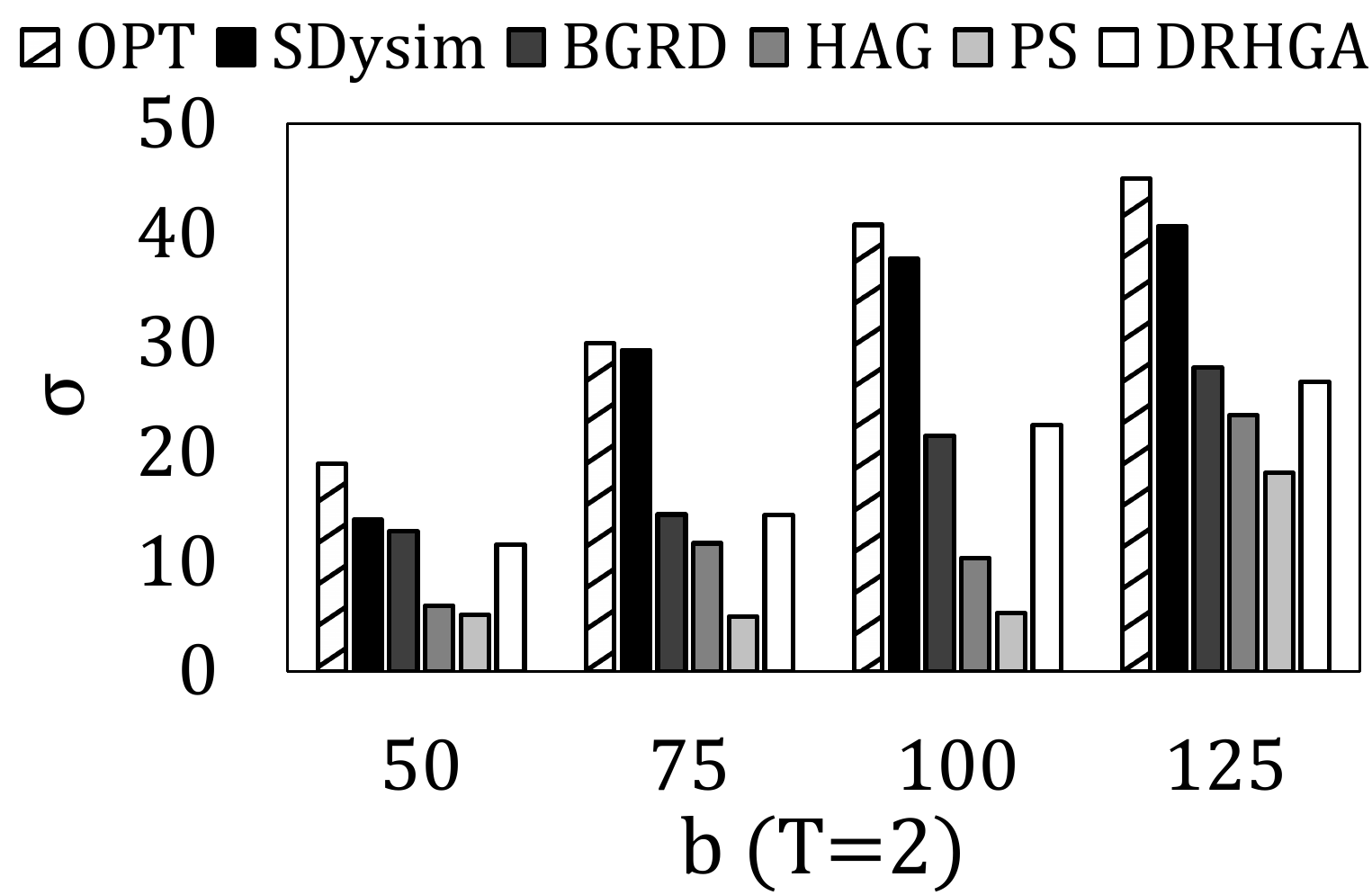}
        \label{Fig:opt_b}
    }%
    \subfigure[Diff. numbers of promotions.]{
        \centering
        \includegraphics[width=0.235\textwidth]{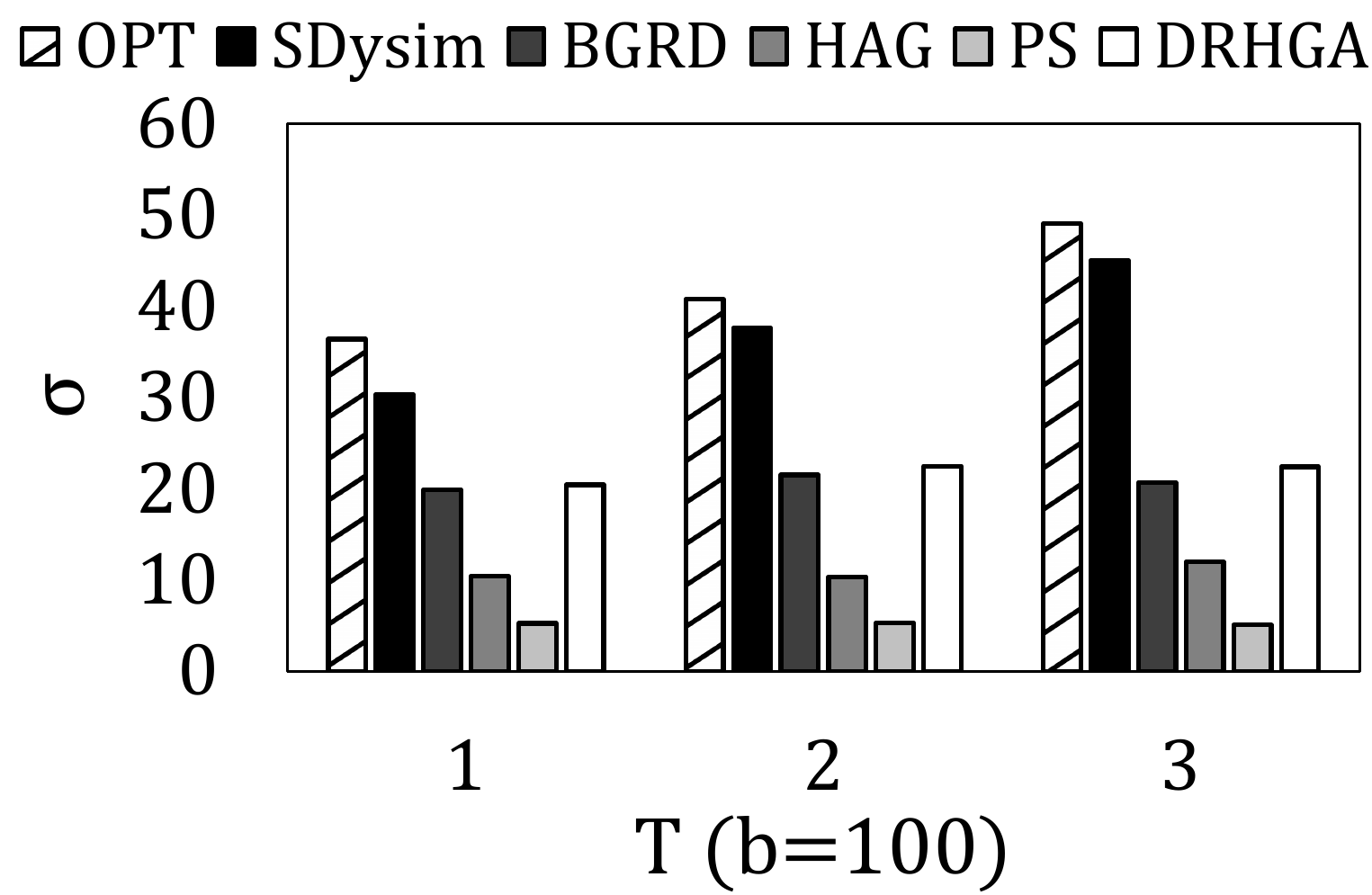}
        \label{Fig:opt_T}
    }
    \caption{\revise{Comparisons with optimal solutions.}}
    \label{Fig:opt}
\end{figure}

\begin{figure*}[t]
    \centering
    \subfigure[Influence (\textit{Yelp}).]{
        \centering
        \includegraphics[width=0.24\textwidth]{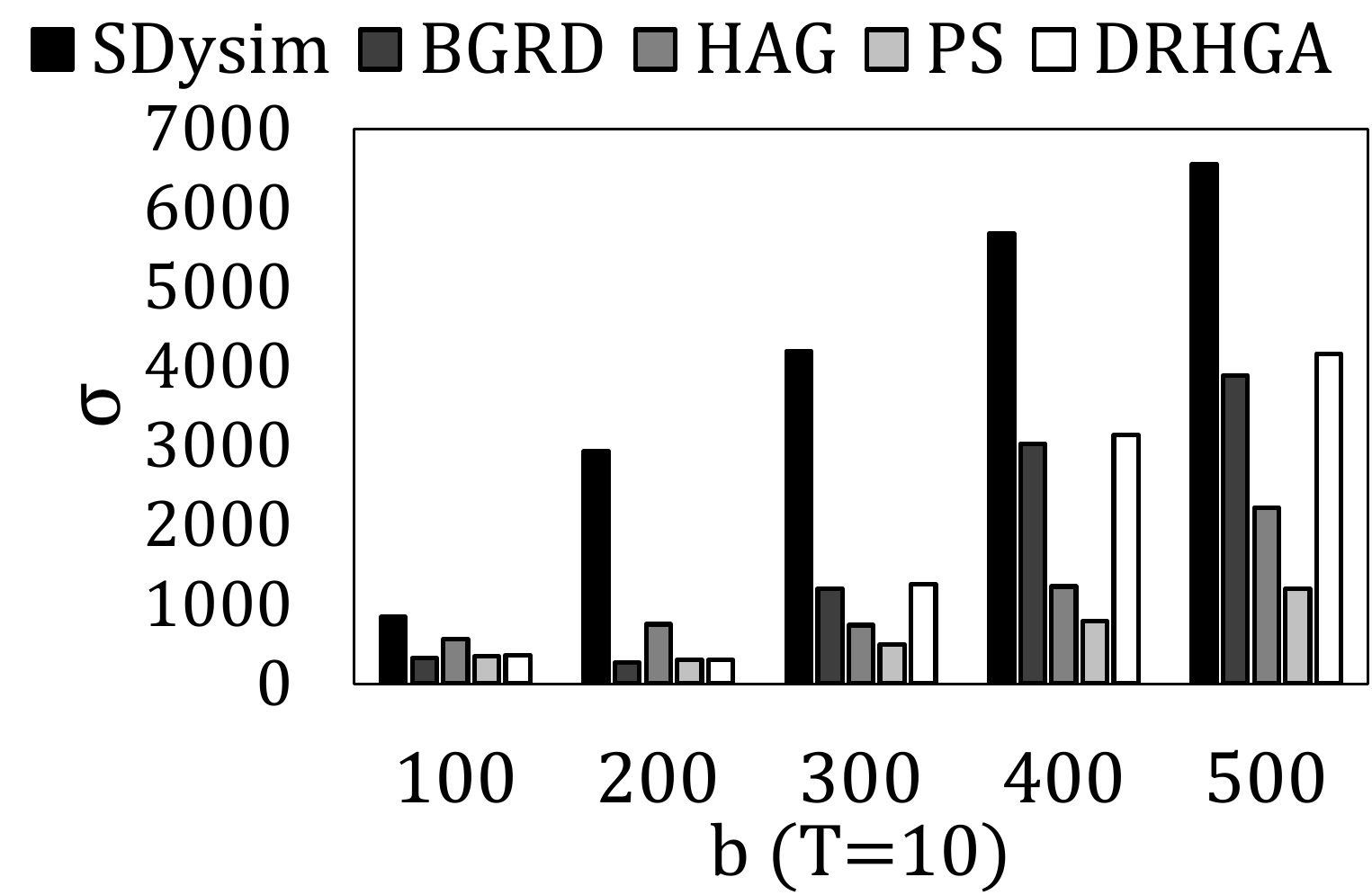}
        \label{Fig:yelp_sigma_b}
    }%
    \subfigure[Influence (\textit{Amazon}).]{
        \centering
        \includegraphics[width=0.24\textwidth]{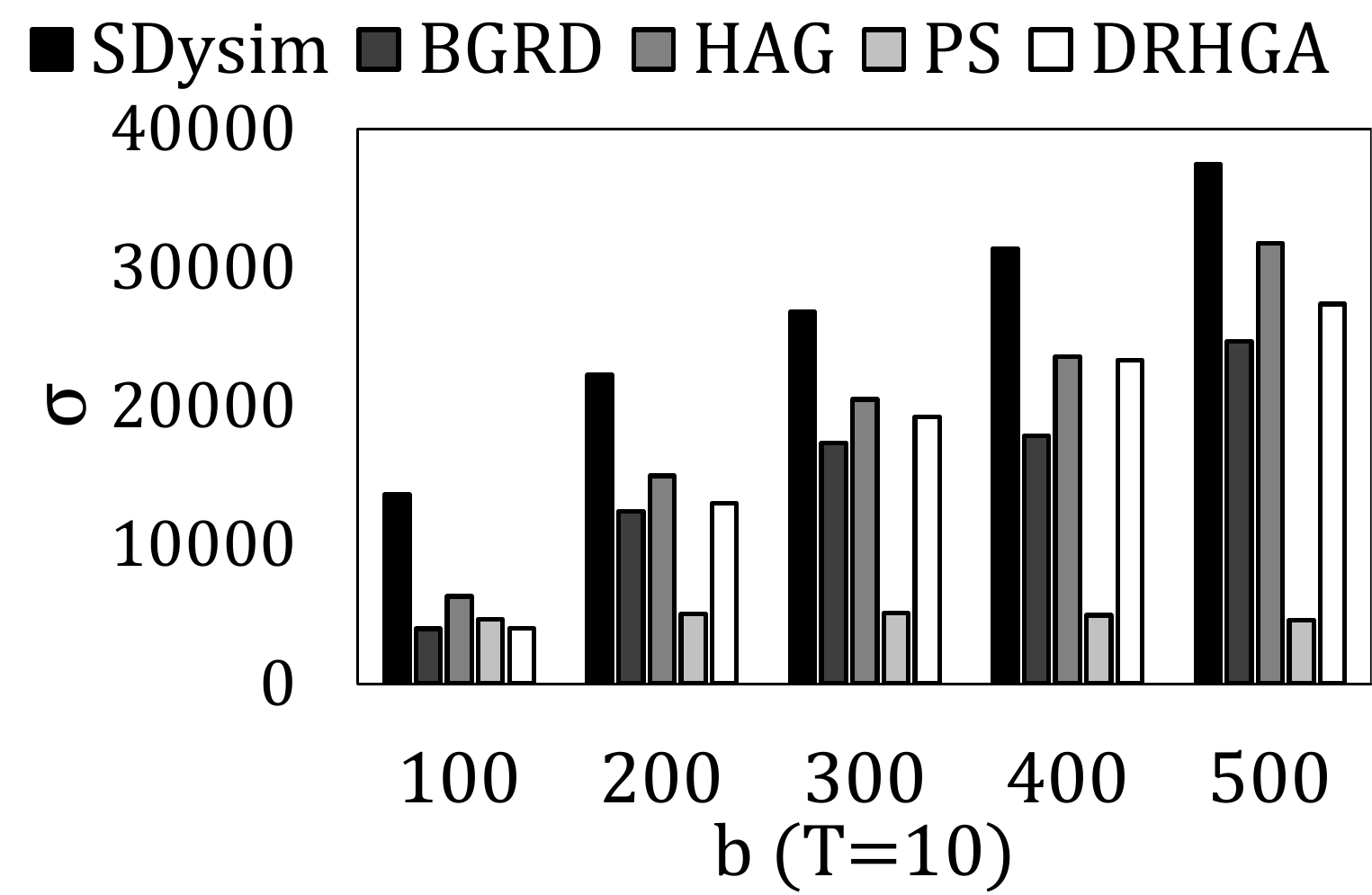}
        \label{Fig:amazon_sigma_b}
    }%
    \subfigure[Influence (\textit{Douban}).]{
        \centering
        \includegraphics[width=0.24\textwidth]{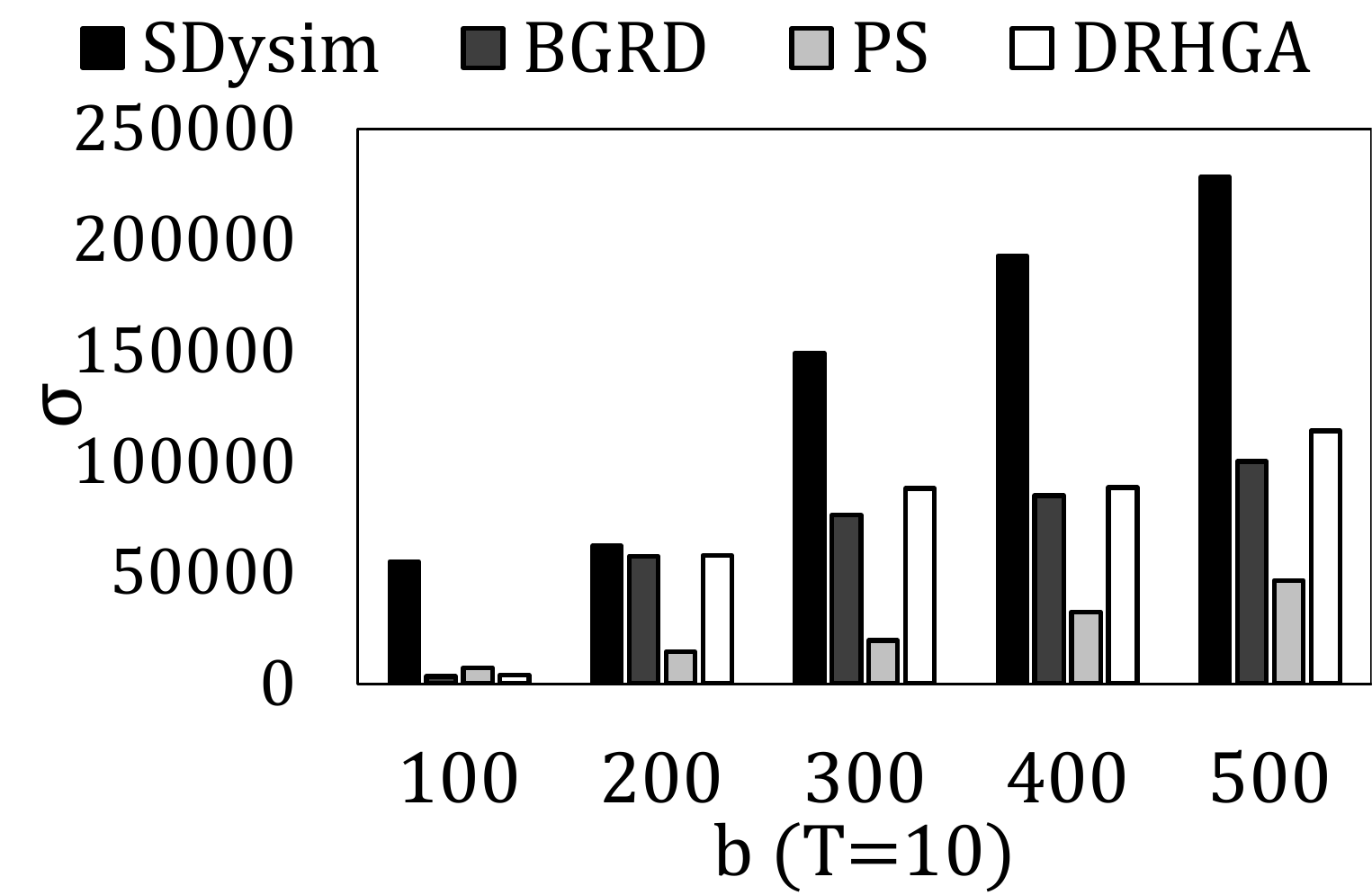}
        \label{Fig:douban_sigma_b}
    }%
    \subfigure[Execution time (\textit{Amazon}).]{
        \centering
        \includegraphics[width=0.24\textwidth]{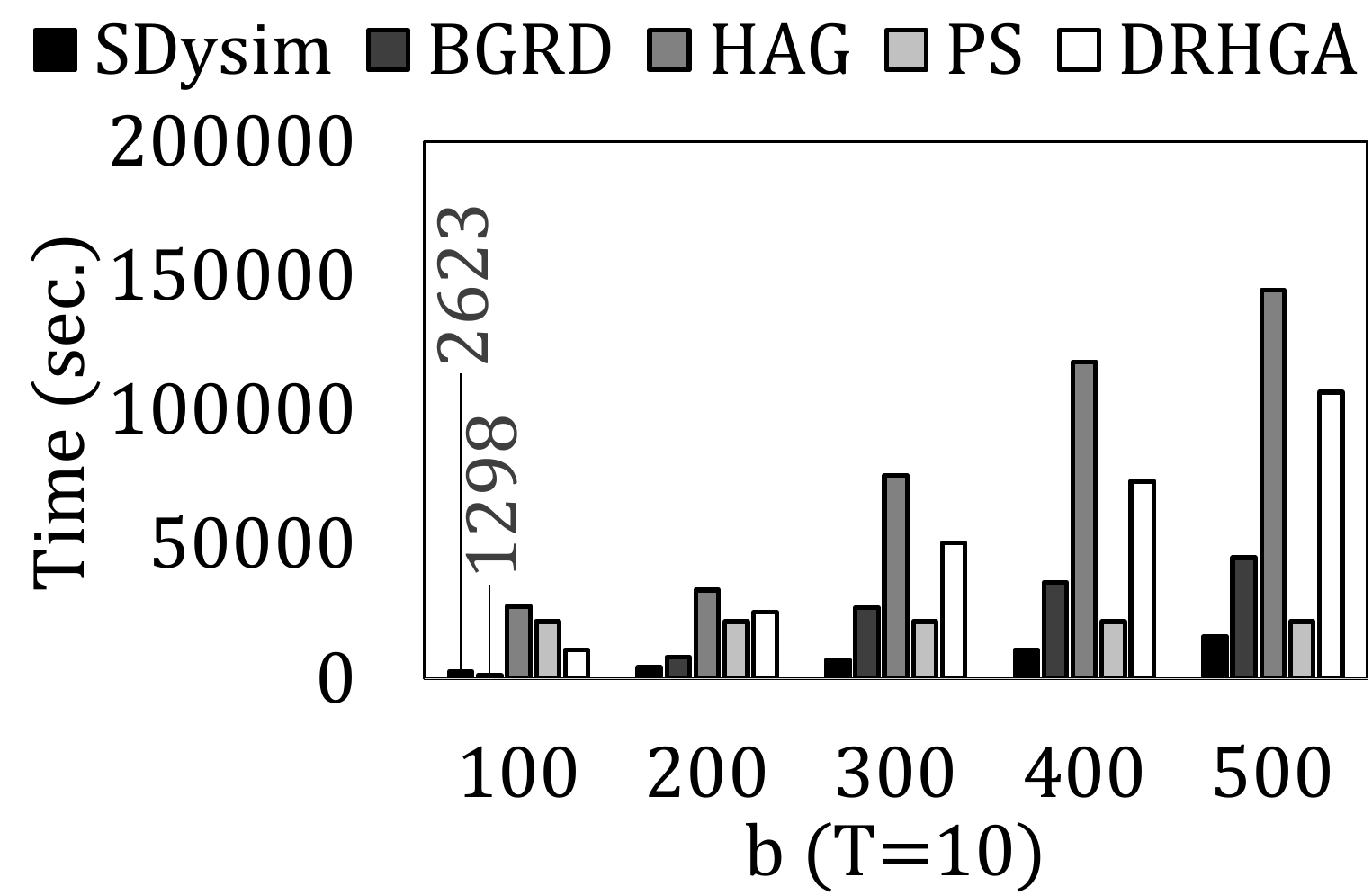}
        \label{Fig:amazon_time_b}
    }\hfill
    \subfigure[Influence (\textit{Yelp}).]{
        \centering
        \includegraphics[width=0.24\textwidth]{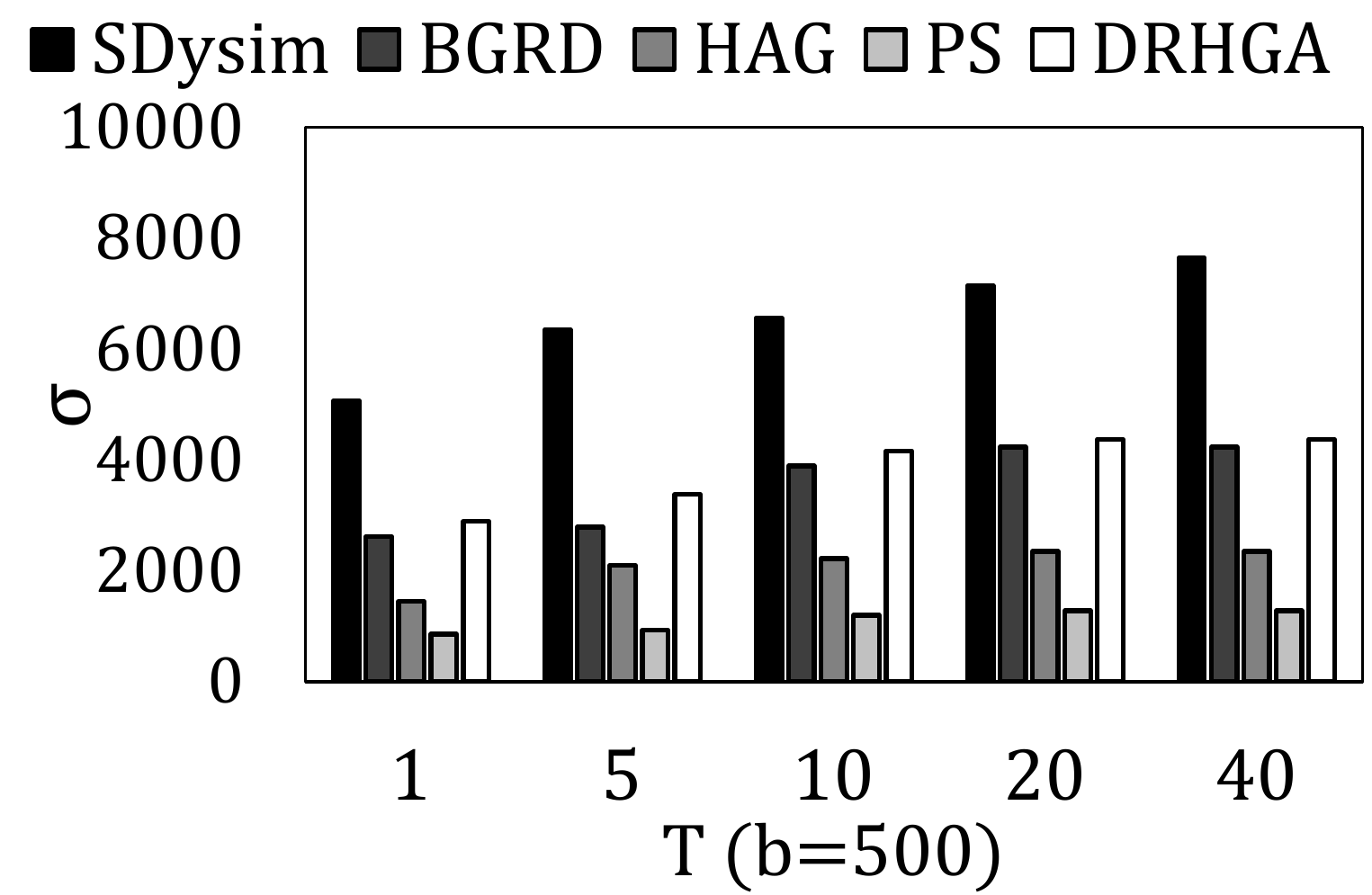}
        \label{Fig:yelp_sigma_T}
    }%
    \subfigure[Influence (\textit{Amazon}).]{
        \centering
        \includegraphics[width=0.24\textwidth]{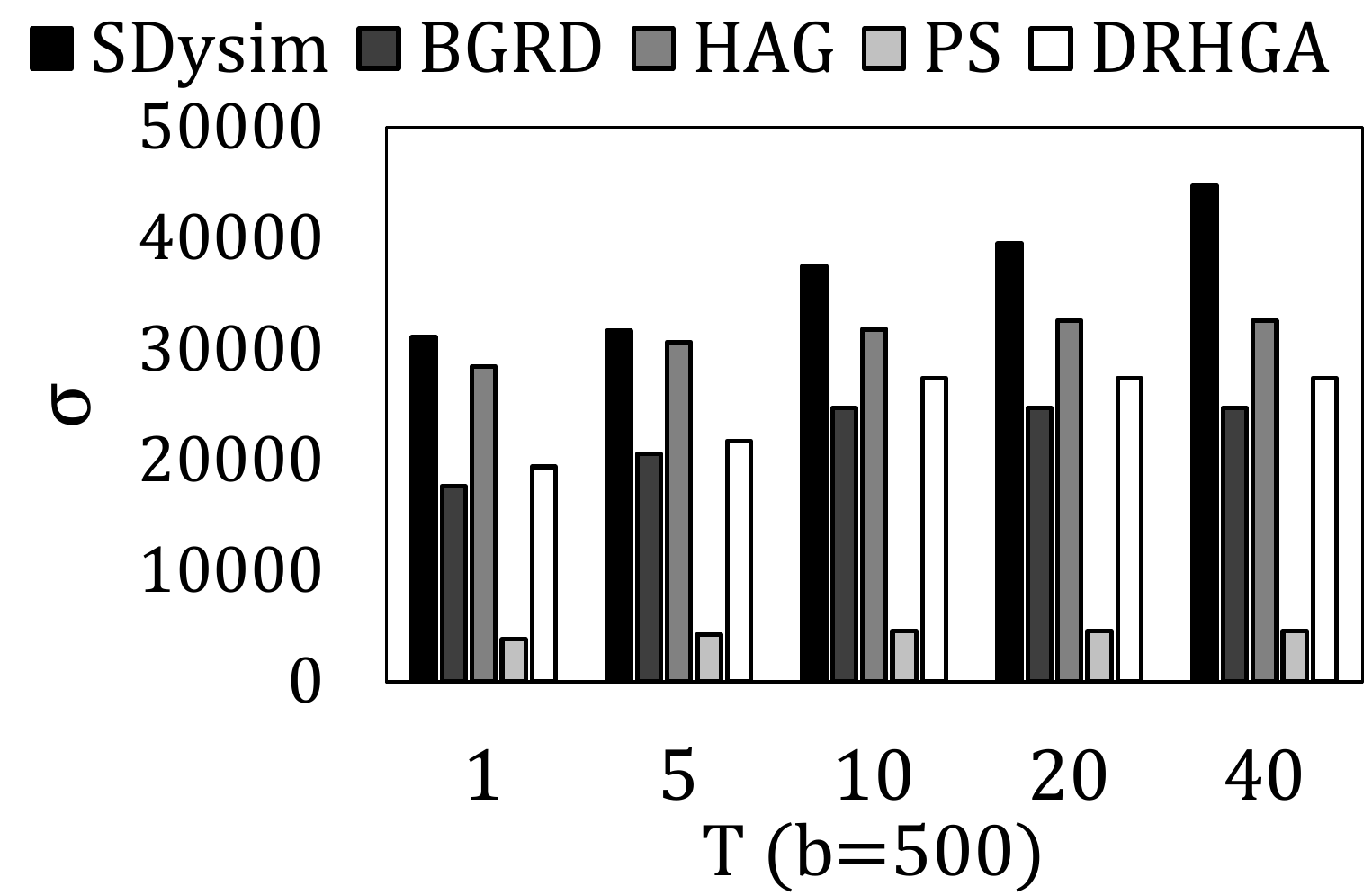}
        \label{Fig:amazon_sigma_T}
    }%
    \subfigure[Execution time (\textit{Amazon}).]{
        \centering
        \includegraphics[width=0.24\textwidth]{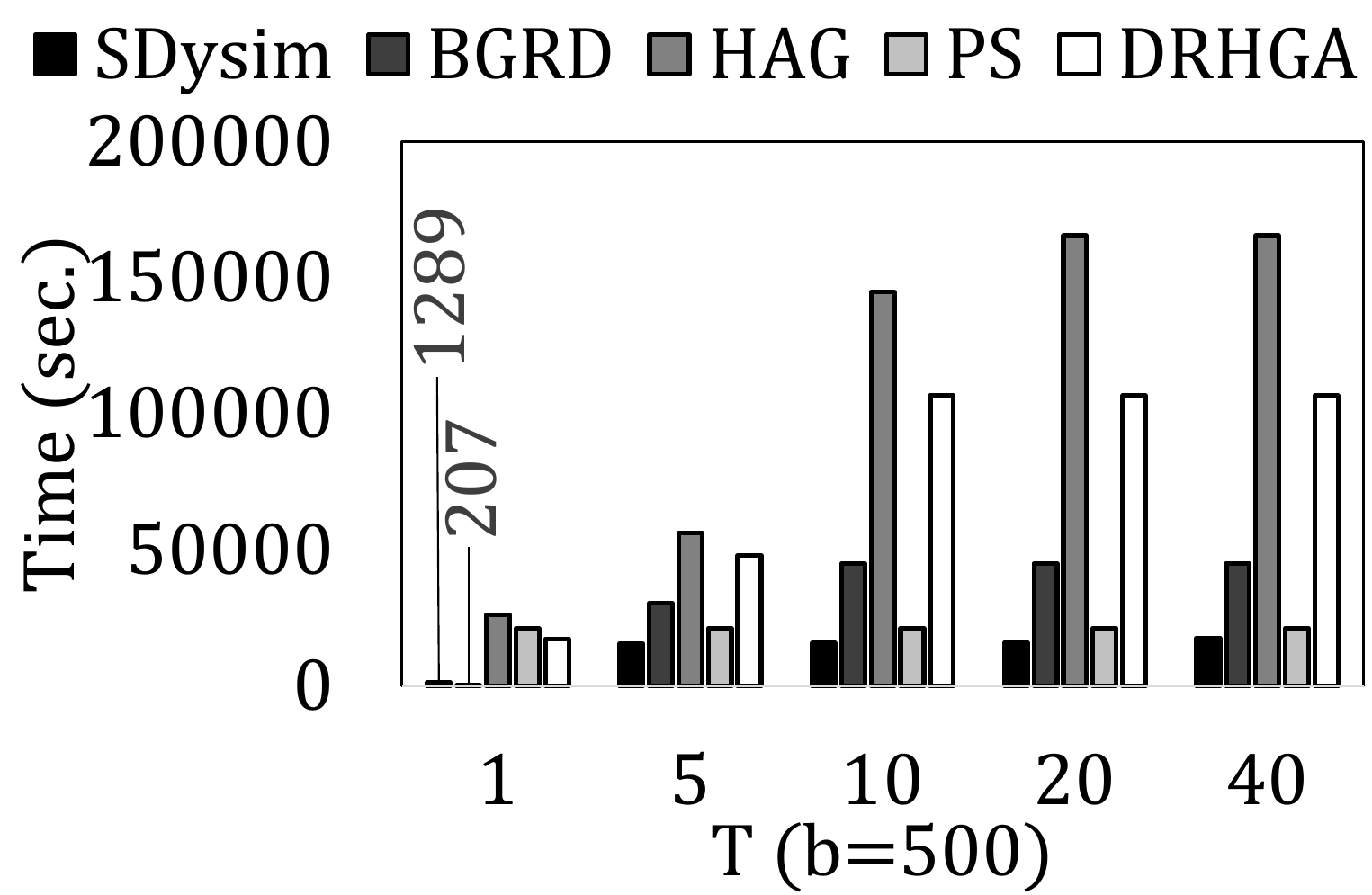}
        \label{Fig:amazon_time_T}
    }%
    \subfigure[Execution time.]{
        \centering
        \includegraphics[width=0.24\textwidth]{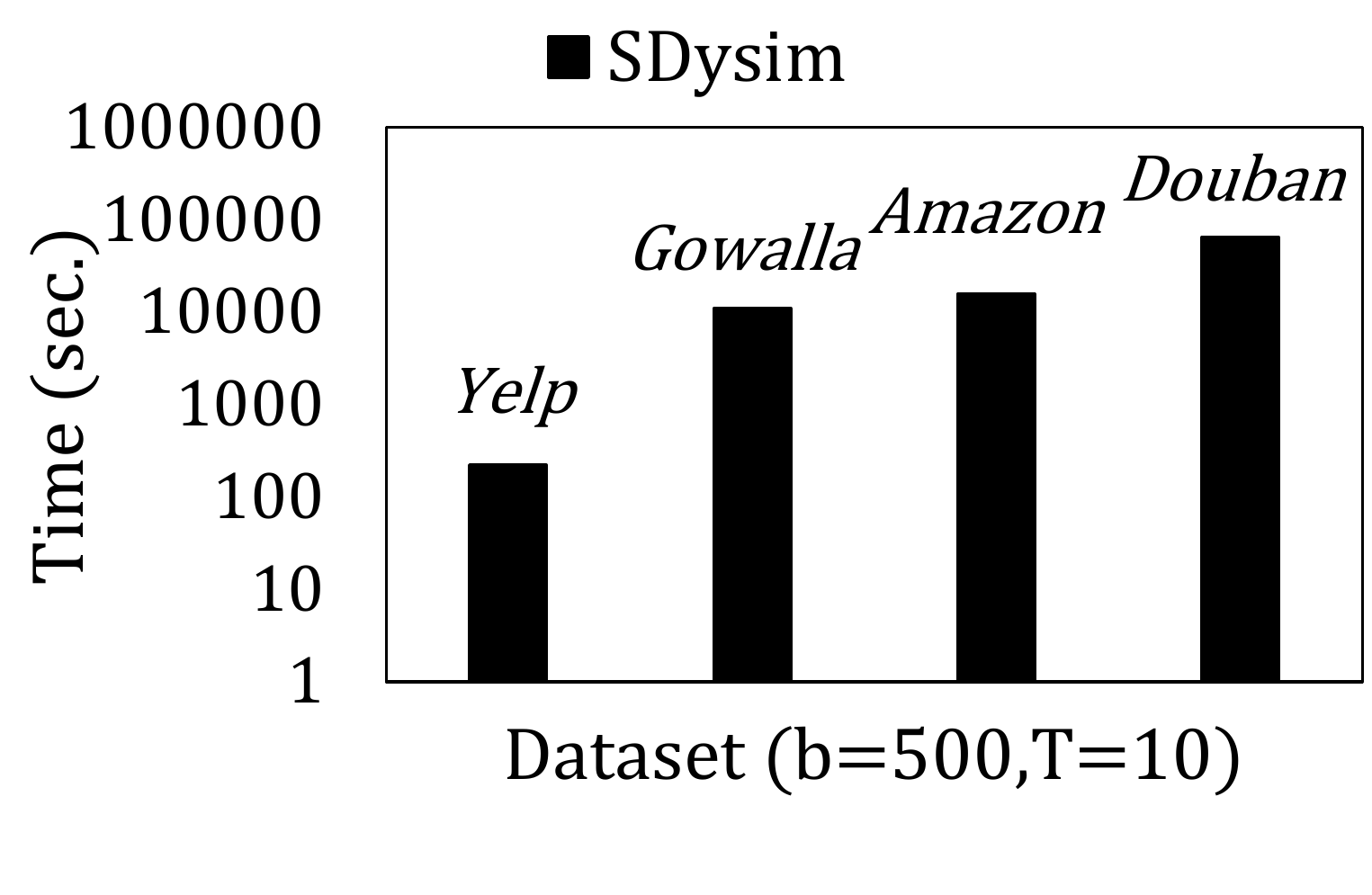}
        \label{Fig:all_time_T}
    }\hfill
    \caption{\revise{Comparisons on large datasets.}}
    \label{Fig:large}
\end{figure*}

Figs.~\ref{Fig:yelp_sigma_b}-\ref{Fig:douban_sigma_b} compare the \opt{full}{importance-aware }influence in large datasets under different budgets.\footnote{Fig.~\ref{Fig:douban_sigma_b} doesn't include HAG due to execution time longer than 12 hours.} \revise{For all datasets, \opt{short}{\salgo}\opt{full}{\algo} achieves the largest influence spread, followed by DRHGA, BGRD, HAG, and PS, because \opt{short}{\salgo}\opt{full}{\algo} is able to exploit the changes in users' preferences\opt{short}{.}\opt{full}{ and social influence strength.}} PS fails to obtain a large influence spread because it only estimates the influence of a seed alone and cannot utilize the impact of items from other promotions to find seeds. BGRD usually achieves smaller than half of the \opt{full}{importance-aware }influence compared with \opt{short}{\salgo}\opt{full}{\algo}, because it neglects the substitutable relationship and regards all items as a bundle to be promoted. 
\revise{Although DRHGA also promotes all items, it is usually better than BGRD since DRHGA is able to select appropriate users to promote each item, instead of regarding all items as a bundle in BGRD. However, as DRHGA does not choose items to be promoted, it still generate a smaller influence spread compared with \opt{short}{\salgo}\opt{full}{\algo}.}
HAG outperforms BGRD in \textit{Yelp} with low budgets and in \textit{Amazon} when the budget is relatively low to the social network size. This is because HAG greedily selects the most influential combination of user-item pairs as the seeds, instead of the most influential user to promote a bundle of items, making the solutions of HAG more cost-effective. BGRD fails to achieve a large influence spread for a large $b$ in \textit{Douban} since items (e.g., songs and books) in \textit{Douban} are usually complementary, but BGRD still allocates the budget to the same users to promote a bundle of complementary items.

Figs.~\ref{Fig:yelp_sigma_T}-\ref{Fig:amazon_sigma_T} present the \opt{full}{importance-aware }influence in large datasets under different numbers of promotions with the maximal $T$ as 40 (following \cite{sun2018multi}). \revise{\opt{short}{\salgo}\opt{full}{\algo} achieves the largest influence spread for all $T$ with significant increments as $T$ grows, because TMI of \opt{short}{\salgo}\opt{full}{\algo} first arranges the promoting order of target markets, and  \opt{short}{\salgo\ then exploits DR to prioritize items to be promoted for each target market.}\opt{full}{\algo\ then exploits SI (aware of the changes in preferences and social influence strength) to determine proper promotional timings of nominees for each target market.}} In contrast, the influence spreads grow slowly for the baselines, especially when $T \geq 20$, because they cannot arrange the promoting order holistically and fail to utilize more promotions to properly gain more adoptions.

Figs.~\ref{Fig:amazon_time_b} and~\ref{Fig:amazon_time_T}-\ref{Fig:all_time_T} compare the execution time under different budgets and different numbers of promotions, respectively. As shown in Fig.~\ref{Fig:amazon_time_b}, when $b$ varies, \opt{short}{\salgo}\opt{full}{\algo} requires the least execution time for most cases. HAG suffers from finding numerous combinations of seeds for a large budget. PS requires much time to search for maximum influence paths to evaluate the influence of a user. \revise{Although DRHGA only selects users, it takes more time than BGRD since the selection process is repeated for each item.} As $b$ becomes larger, the execution time of \opt{short}{\salgo}\opt{full}{\algo} only slightly increases since TMI quickly selects influential nominees by MCP according to the cost and increment on \opt{full}{important-aware }influence for each candidate nominee. PS is less sensitive to $b$ since it employs a discounting strategy to estimate a seed's influence under the impact of selected seeds. \revise{On the other hand, as shown in Fig.~\ref{Fig:amazon_time_T}, \opt{short}{\salgo}\opt{full}{\algo} requires a low overhead to find promotional timings \opt{short}{since it assigns the promotions by TMI and DRE, which are less sensitive to $T$}\opt{full}{due to an efficient search with pruning in TDSI}, whereas the baselines greedily assigning the promotional timings tend to suffer from larger $T$.} To show the scalability of \opt{short}{\salgo}\opt{full}{\algo}, Fig.~\ref{Fig:all_time_T} compares the execution time of \opt{short}{\salgo}\opt{full}{\algo} on different datasets (in the order of the number of users in the social network). The time increases not only as the number of users increases but also as that of items increases (e.g., so the time on \textit{Gowalla} and \textit{Amazon} are similar) due to the propagation of item impact.
\label{para:time_comparison}

\subsection{Ablation Study}
\label{sec:ablation_study}
\revise{Fig.~\ref{Fig:ablation_test} compares \opt{short}{\salgo}\opt{full}{\algo}, \opt{short}{\salgo}\opt{full}{\algo} without target markets (i.e., w/o TM), and \opt{short}{\salgo}\opt{full}{\algo} without item priority (i.e., w/o IP). We have the following three observations. First, the influence spread is smaller when target markets are not identified, since the selected nominees may promote substitutable items to the same users in consecutive promotions, which detracts from users' preferences for the posterior items to be promoted. By contrast, \opt{short}{\salgo}\opt{full}{\algo} effectively avoids the antagonism of the substitutable relationship by identifying and prioritizing the target markets.
Second, the influence spread of \opt{short}{\salgo}\opt{full}{\algo} without item priority is also smaller than that of \opt{short}{\salgo}\opt{full}{\algo}, because all items in a target market are promoted simultaneously, and therefore the promotion of an item is hardly facilitated by promoting its complementary items first. In contrast, \opt{short}{\salgo}\opt{full}{\algo} determines the item priority by exploiting DR, which carefully measures the impact from previously promoted items on an item and also the potential impact from this item on other items in subsequent promotions. 
Third, as $T$ increases, the gaps between \opt{short}{\salgo}\opt{full}{\algo} and \opt{short}{\salgo}\opt{full}{\algo} w/o TM/IP increase. This is because the number of promotions in \opt{short}{\salgo}\opt{full}{\algo} w/o TM/IP is limited, i.e., at most the number of items/target markets, implying that more promotions are not beneficial for a larger influence spread. By contrast, \opt{short}{\salgo}\opt{full}{\algo} effectively schedules the promotional timings of different target markets and different items to exploit the propagation of item impacts.
}

\begin{figure*}
    \centering
    \subfigure[Diff. $b$ (\textit{Yelp}.)]{
        \centering
        \includegraphics[width=0.23\textwidth]{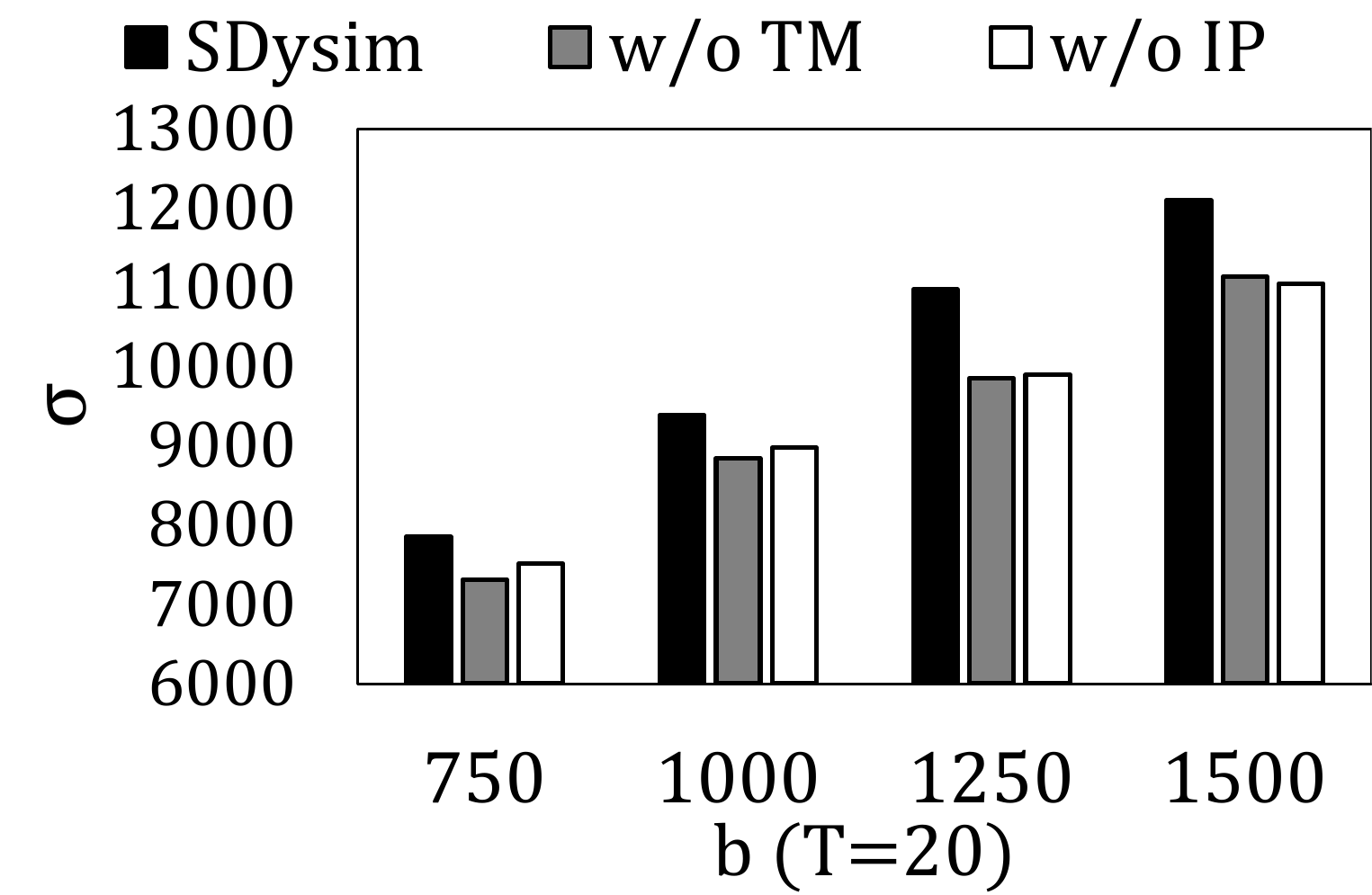}
        \label{Fig:yelp_ablation_test_budget}
    }%
    \subfigure[Diff. $T$ (\textit{Yelp}).]{
        \centering
        \includegraphics[width=0.23\textwidth]{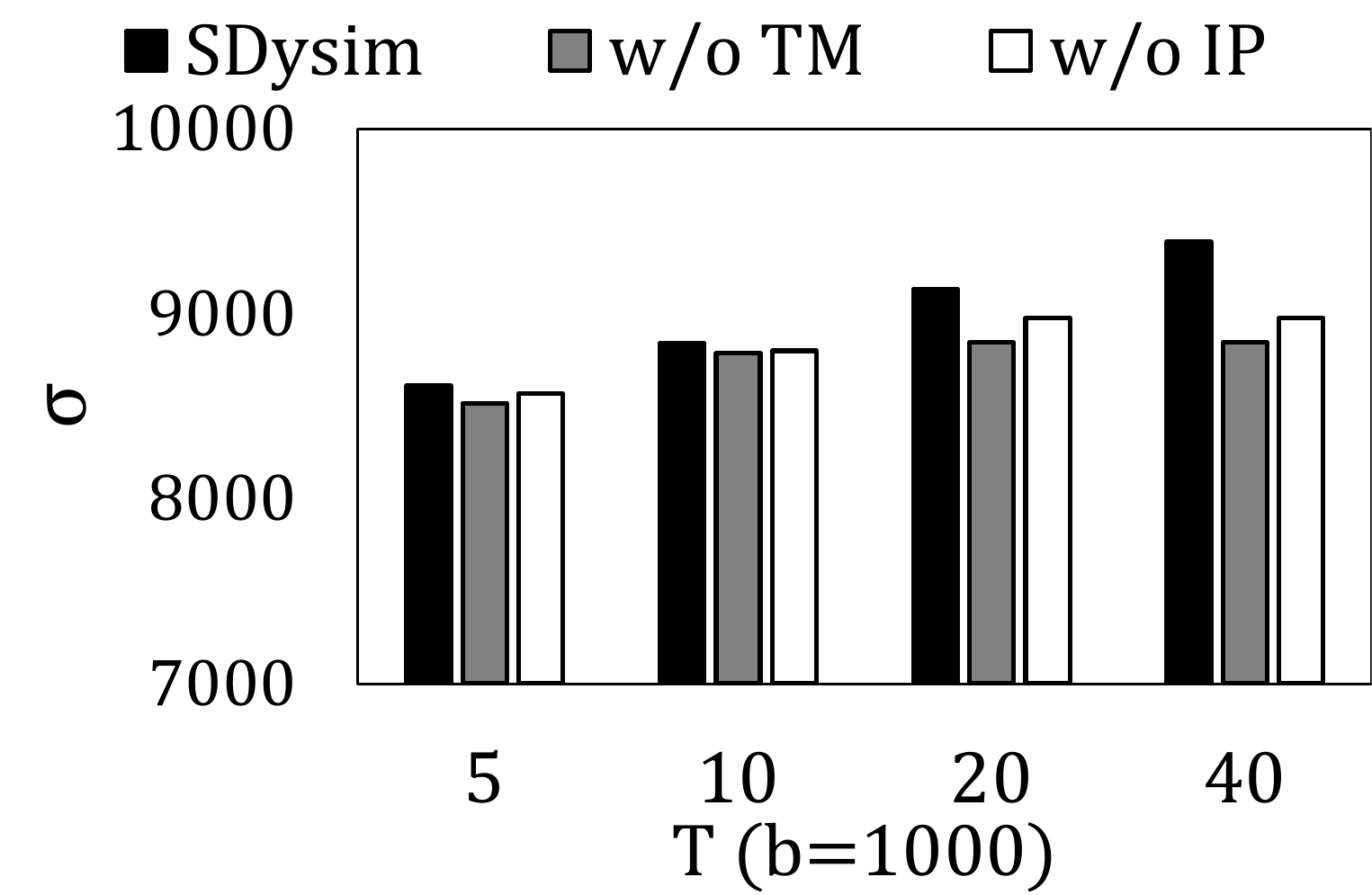}
        \label{Fig:yelp_ablation_test_T}
    }%
    \subfigure[Diff. $b$ (\textit{Amazon}.)]{
        \centering
        \includegraphics[width=0.23\textwidth]{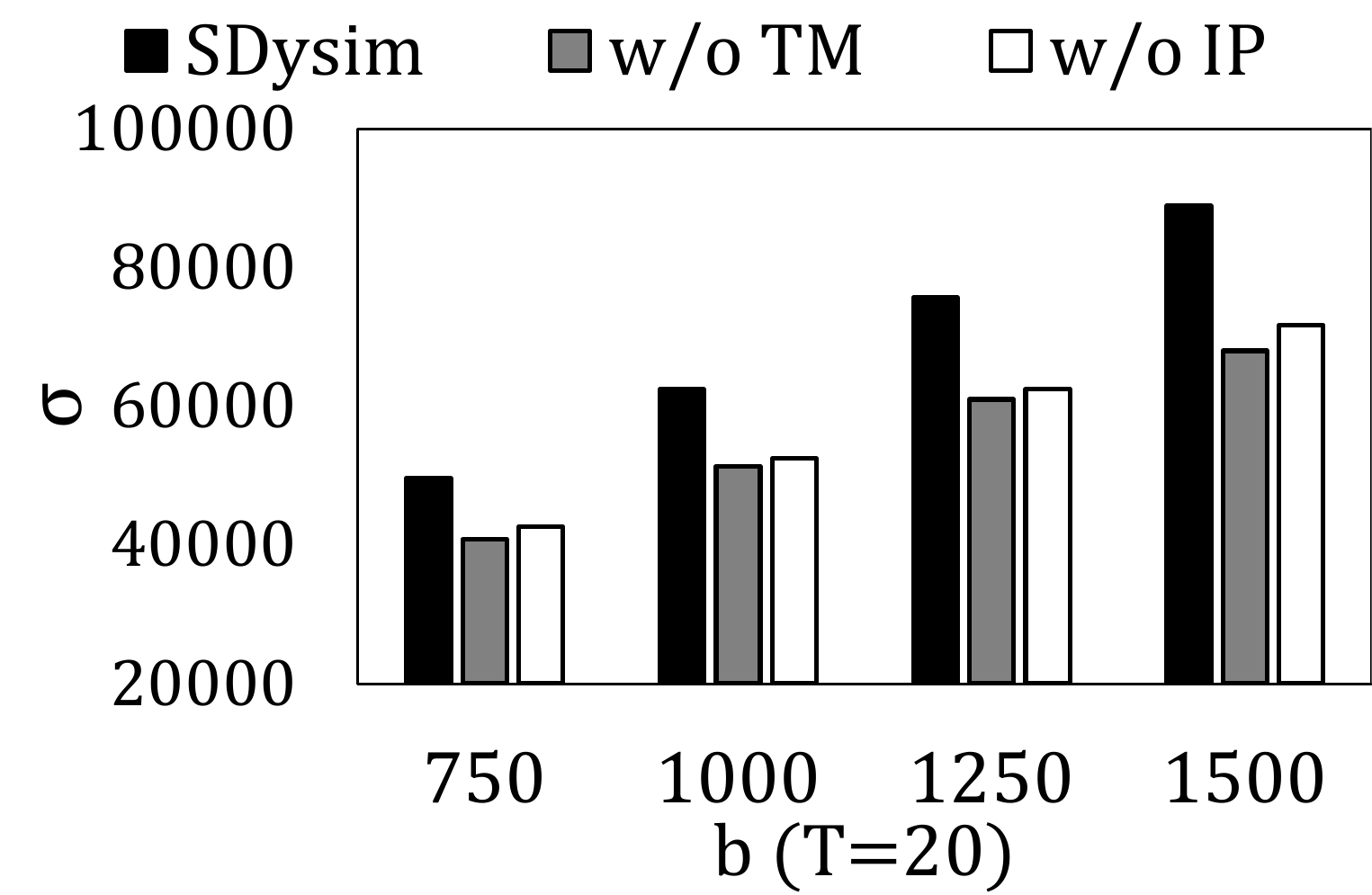}
        \label{Fig:amazon_ablation_test_budget}
    }%
    \subfigure[Diff. $T$ (\textit{Amazon}).]{
        \centering
        \includegraphics[width=0.23\textwidth]{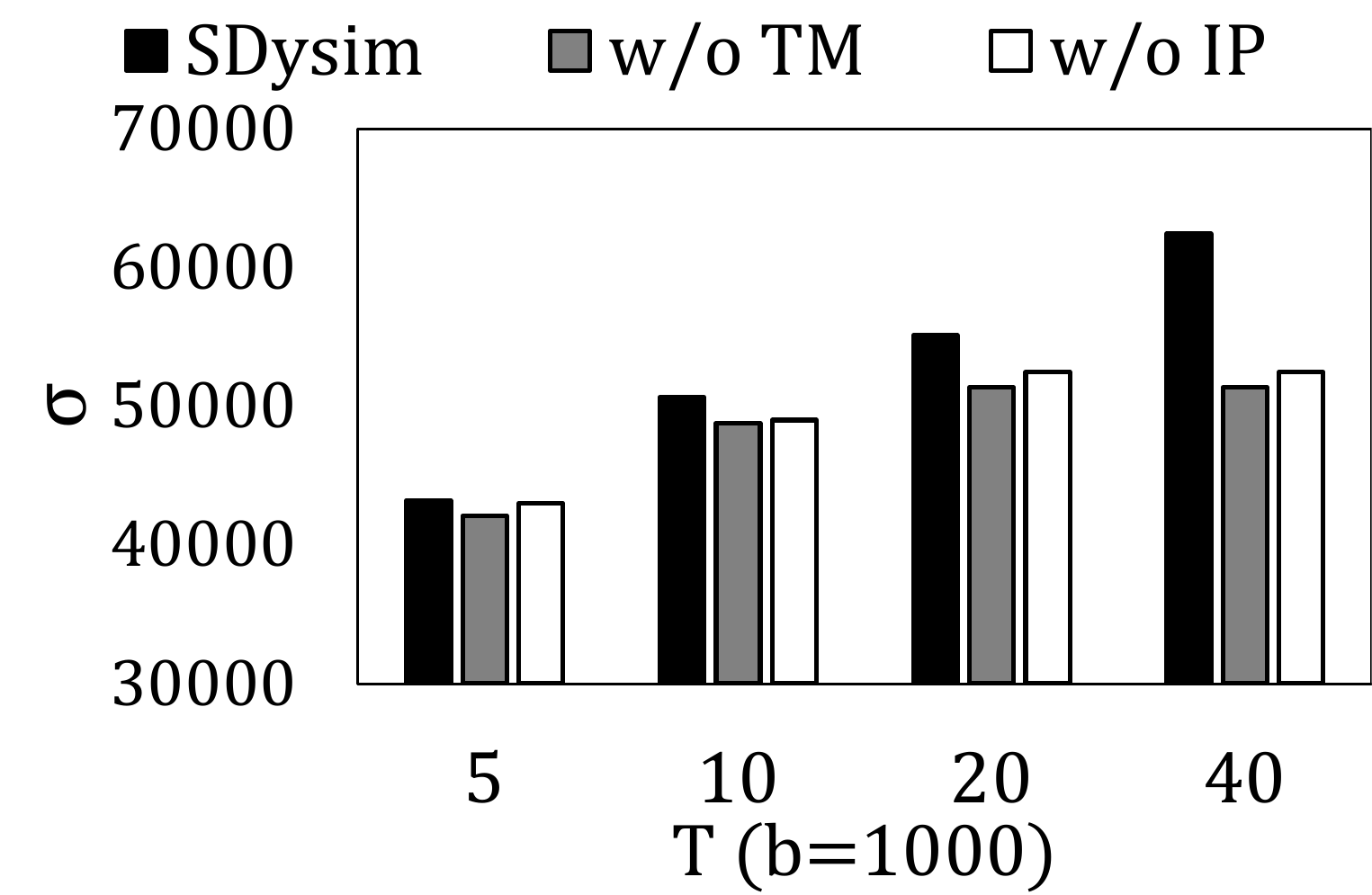}
        \label{Fig:amazon_ablation_test_T}
    }\hfill
    \caption{\revise{The ablation study.}}
    \label{Fig:ablation_test}
\end{figure*}

\subsection{Comparison of Different Market Orders}
\label{sec:market_order}
\revise{To compare with Antagonistic Extent (AE), we leverage the following metrics to evaluate additional promotional orders of target markets: profitability (PF) \cite{simkin1998prioritising}, size of the market (SZ) \cite{simkin1998prioritising}, relative market share (RMS) \cite{farris2010marketing}, and random (RD).
PF and SZ are two of the most common criteria to prioritize target markets in the marketing research field. PF is the expected adoptions under the promotion from the corresponding nominees minus the cost of the nominees. SZ is the number of customers in the target market. A target market with a larger PF or SZ is preferred to be promoted earlier.
RMS is widely used to assess the value of a firm's item in the product management field. RMS of an item $x$ is defined as the ratio of $x$'s market share to the largest market share of its substitutable item, where the market share is evaluated by the number of users preferring the item most. The target market that promotes items with a higher RMS is prioritized. }

\revise{Fig.~\ref{Fig:market_order} manifests that AE and PF usually achieve the largest influence spread, followed by SZ, RMS, and RD. AE and PF outperform the others since AE prioritizes the target markets that have less substitutable relationship on the subsequent target markets, while PF prioritizes the target markets with more profits to ensure their influence spread. When there exist excessively large target markets (e.g., identified by plenty of nominees), PF is suggested as the ordering metric, since PF can accurately prioritize these large target markets to maximize the influence. In general cases, AE is usually a better metric to prioritize the target markets since the impact from the substitutable items promoted by prior target markets is minimized. By contrast, SZ, RMS, and RD, without carefully examining the relationships of items promoted in other target markets, cannot avoid the antagonism of the substitutable relationship. The results manifest that promoting target markets with a smaller AE or a larger PF earlier in TMI is beneficial to achieve a larger influence spread. }

\begin{figure*}
    \centering
    \subfigure[Diff. $b$ (\textit{Yelp}.)]{
        \centering
        \includegraphics[width=0.23\textwidth]{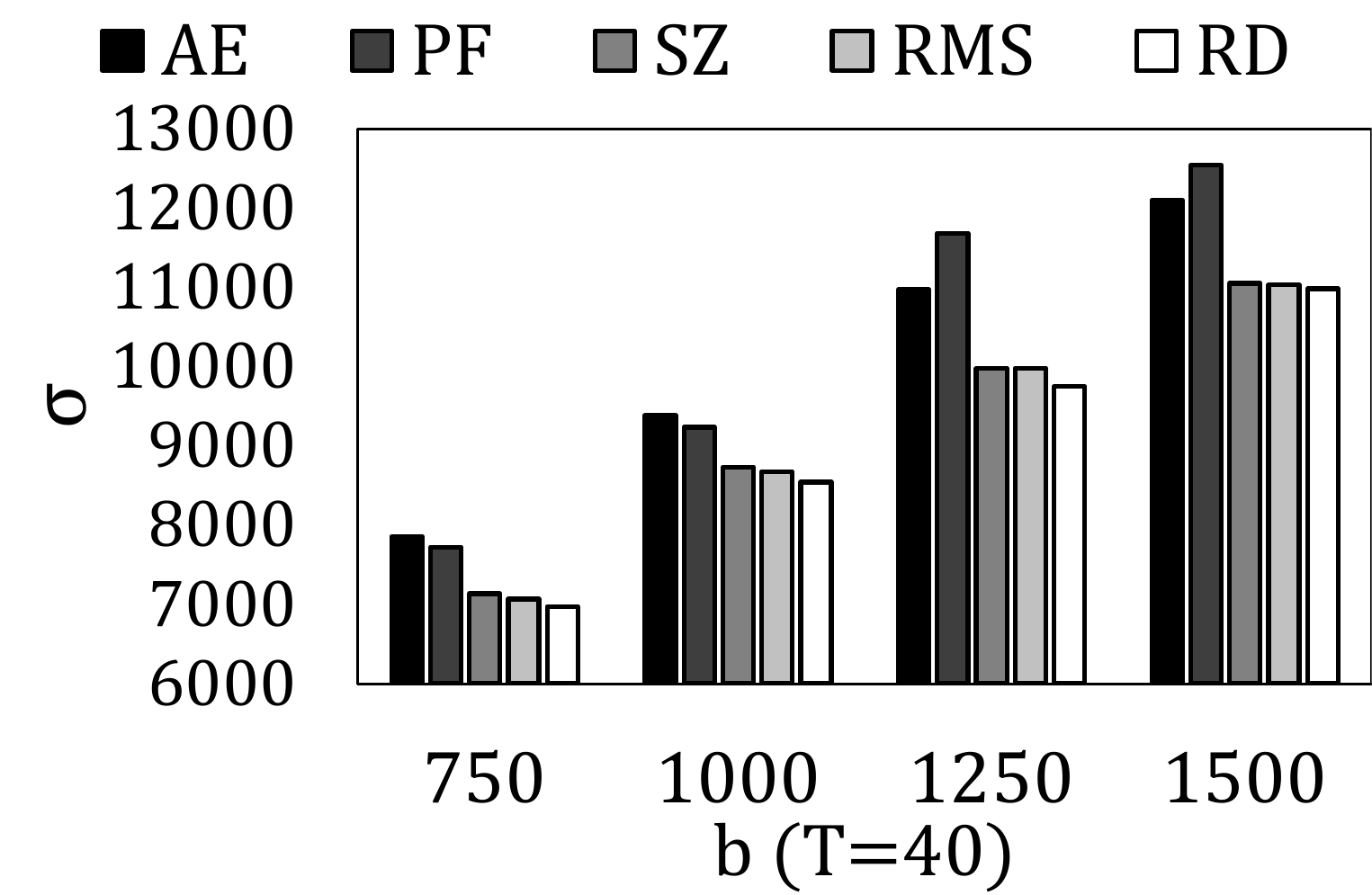}
        \label{Fig:yelp_market_order_budget}
    }%
    \subfigure[Diff. $T$ (\textit{Yelp}).]{
        \centering
        \includegraphics[width=0.23\textwidth]{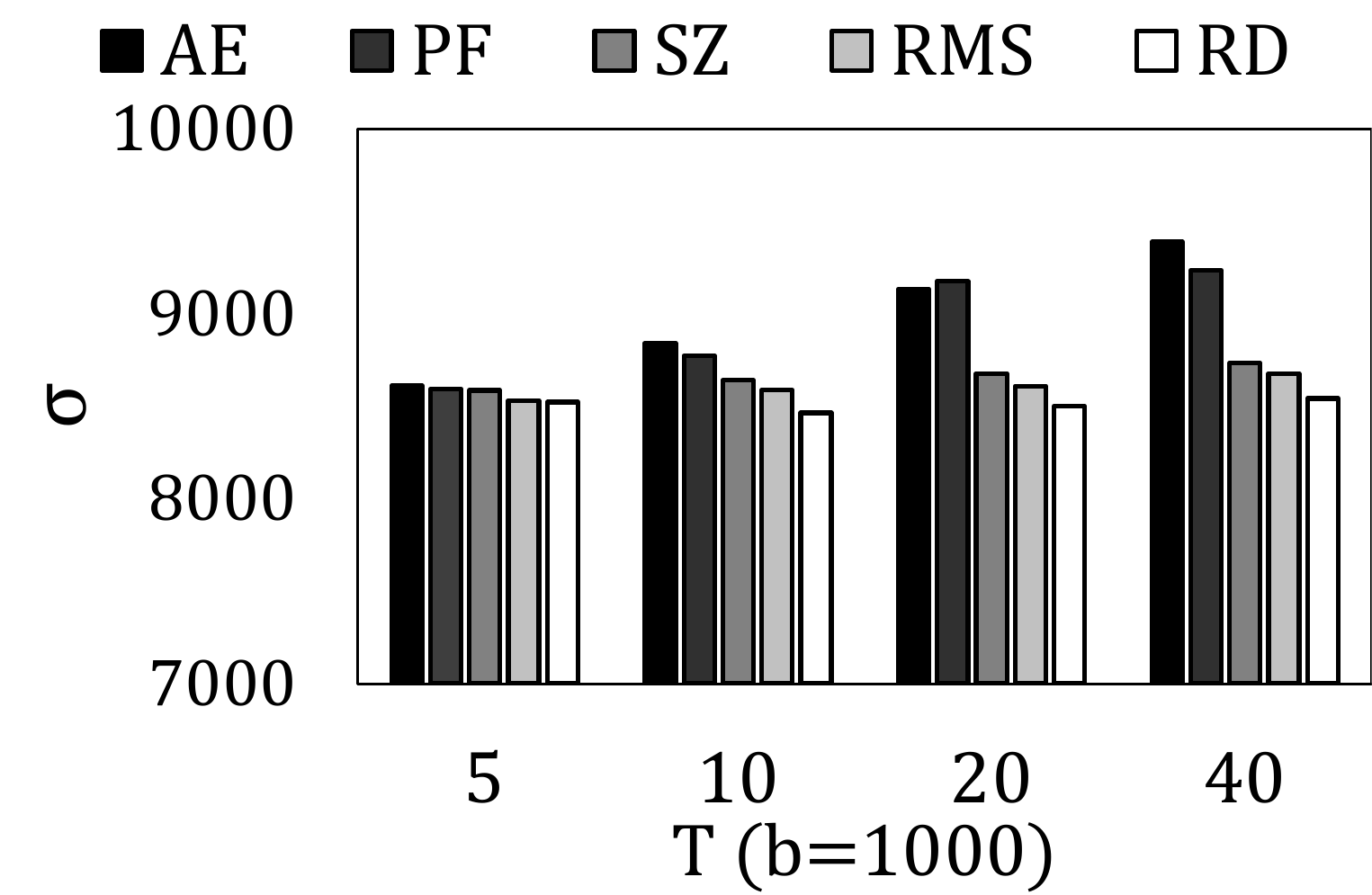}
        \label{Fig:yelp_market_order_T}
    }%
    \subfigure[Diff. $b$ (\textit{Amazon}.)]{
        \centering
        \includegraphics[width=0.23\textwidth]{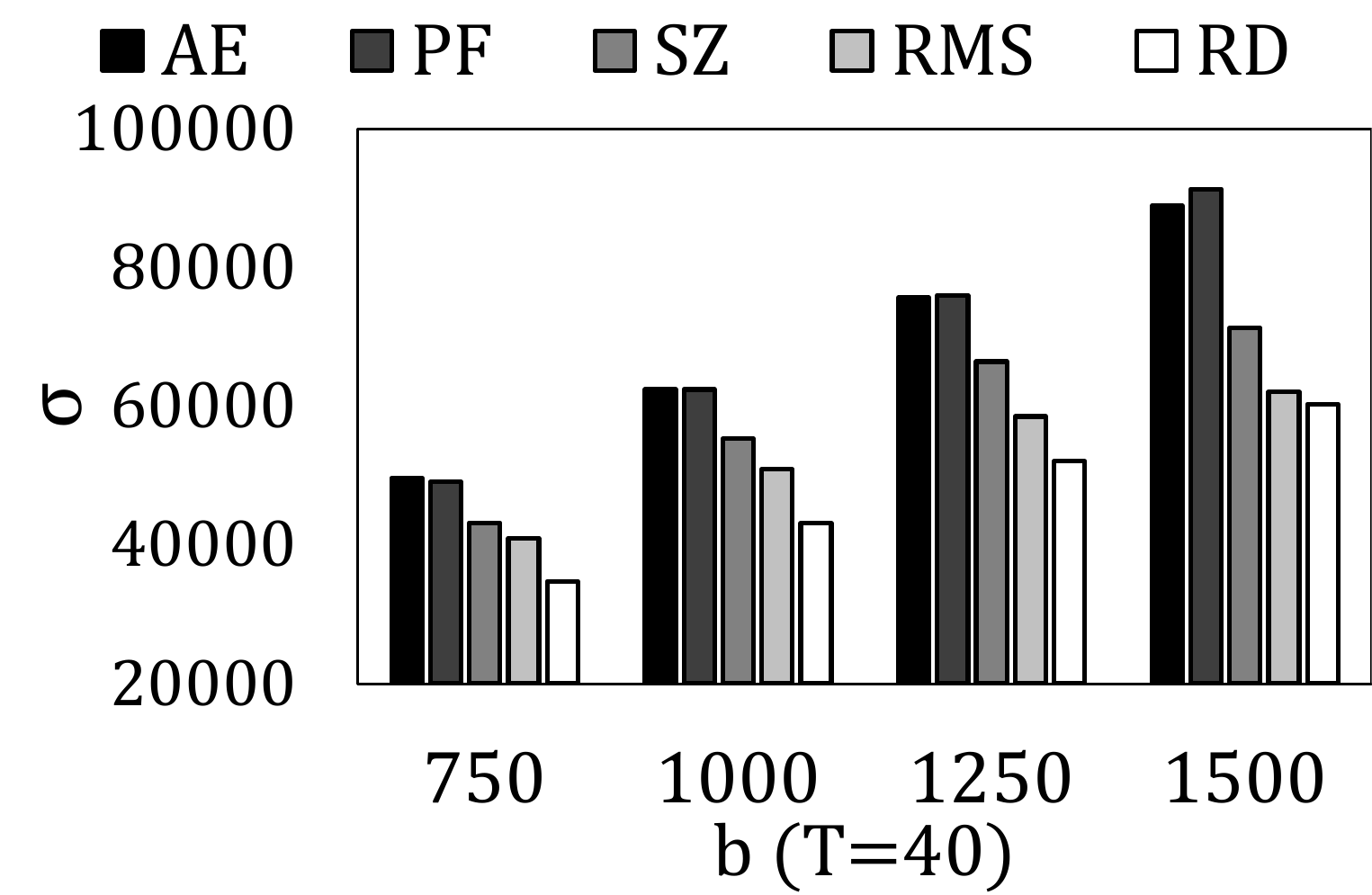}
        \label{Fig:amazon_market_order_budget}
    }%
    \subfigure[Diff. $T$ (\textit{Amazon}).]{
        \centering
        \includegraphics[width=0.23\textwidth]{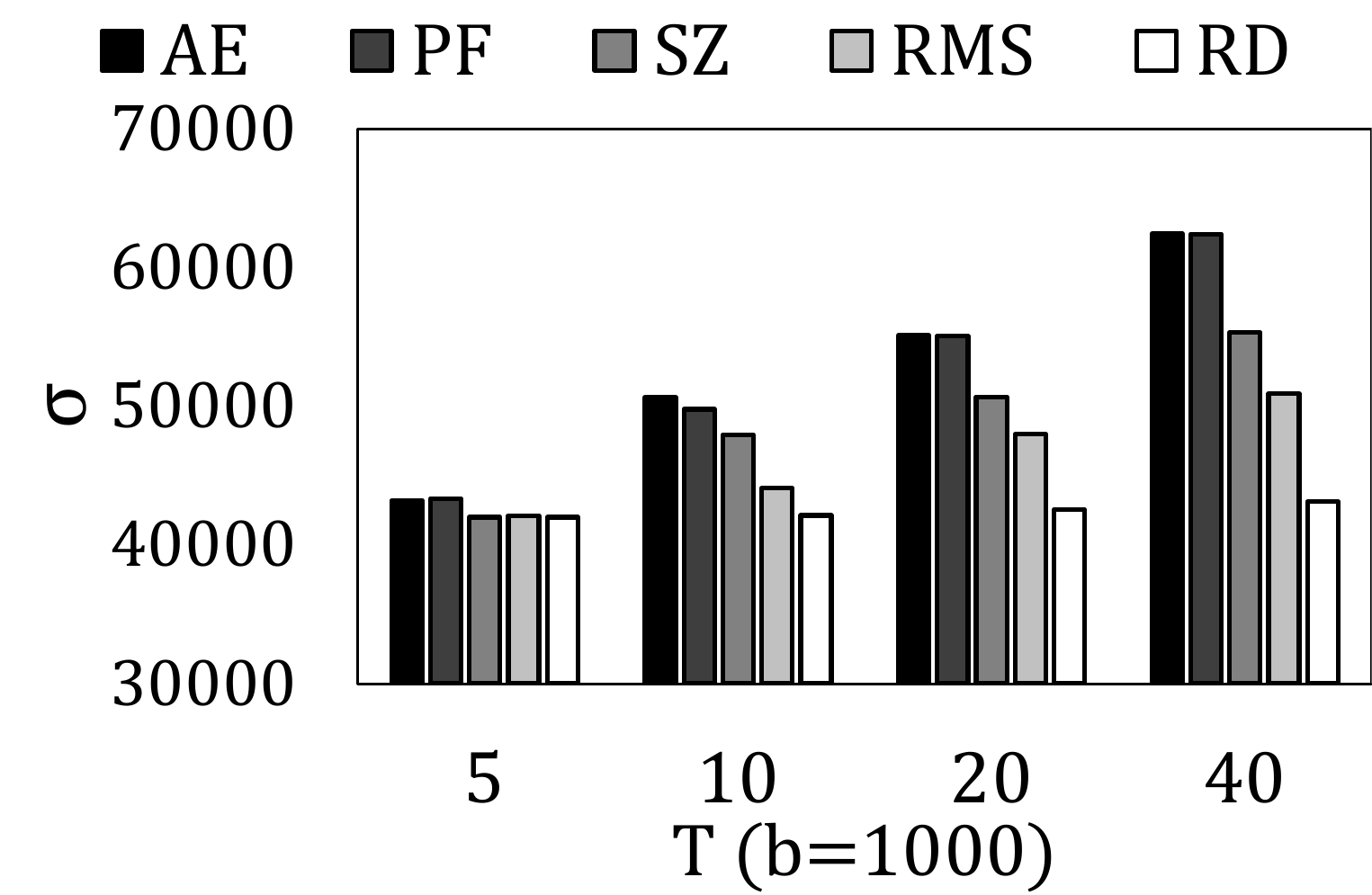}
        \label{Fig:amazon_market_order_T}
    }\hfill
    \caption{\revise{Comparisons of different market orders.}}
    \label{Fig:market_order}
\end{figure*}

\subsection{Empirical Study}
\label{sec:user_study}
In this study, we have recruited five classes for promoting courses by viral marketing to evaluate the effectiveness of \opt{short}{\salgo}\opt{full}{\algo} in real-world settings. There were 30 elective courses for computer science college students, including artificial intelligence (AI), objective-oriented programming (OOP), and big data, to name a few. The goal of the campaigns is to encourage the students in Taiwan University to select those courses, i.e., maximizing the total number of students selecting the elective courses. The statistics of all classes are presented in Table~\ref{T:class}. 

To construct KG of these courses, we crawled their syllabuses from Taiwan University, and extracted keywords of courses, related compulsory courses, and research fields of teachers. The meta-graphs were defined according to the curriculum guidelines in Taiwan.\footnote{\revise{\url{https://cirn.moe.edu.tw/Upload/file/32077/83646.pdf}.}} 
Following \cite{Nguyen2016Cost}, the costs of hiring users to promote courses are set as users' out-degree over their initial preferences for courses, since users who are more influential and who prefer the course less may need more incentive to be the seeds. 

To evaluate the effectiveness of different approaches, we have launched campaigns based on the following approaches: 1) \opt{short}{\salgo}\opt{full}{\algo}, 2) BGRD \cite{banerjee2019maximizing}, 3) HAG \cite{hung2016social}, and 4) PS \cite{teng2018revenue}. 
In this study, the budget and the number of promotions were set to 50 and 3, respectively. For \opt{short}{\salgo}\opt{full}{\algo}, \opt{short}{relevance measurement (including the learning of personal weightings on meta-graphs and the constructions of personal item networks) and preference estimation are updated according to \cite{shi2019semrec} and \cite{zhao2017improving}, respectively.}\opt{full}{relevance measurement (including the learning of personal weightings on meta-graphs and the constructions of personal item networks), preference estimation, influence learning and item associations are learned and updated based on \cite{shi2019semrec}, \cite{zhao2017improving}, \cite{zhang2019learning}, and \cite{zhao2017improving}, respectively.} TMI of \opt{short}{\salgo}\opt{full}{\algo} follows \cite{chen2017people} and \cite{chen2010scalable} to cluster nominees and explore influenced users, respectively. 

Fig.~\ref{Fig:user_study} reports the total number of students selecting the elective courses for different approaches in each class. For all classes, \opt{short}{\salgo}\opt{full}{\algo} induces the most students who selected those courses, followed by BGRD, HAG, and PS. These results validate that \opt{short}{\salgo}\opt{full}{\algo} is able to encourage students to select those courses by carefully evaluating the dynamic changes in the relationships between courses. For instance, we observe that a student in Class A initially regarded the complementary relevance between AI and software design for cloud computing (SDCC) as 0.1. After he selected AI and big data, the complementary relevance between AI and SDCC increased to 0.6 (derived according to \cite{shi2019semrec}). He then selected SDCC accordingly. In Class D, another student initially reported that the influence from one of her classmates is 0.2. During the promotions, both of them selected cloud computing and IoT, which increased the classmate's influence to this student to 0.7 (derived according to \cite{zhang2019learning}). Then, this student selected big data after being informed that the classmate selected big data as well. By contrast, BGRD, HAG, and PS do not capture the dynamic changes in the relationships between courses and the ripple effect, resulting in fewer students selecting the elective courses in the end.

Besides, although BGRD is able to select influential students in each class, all courses are promoted as a bundle without considering their relationships. For example, in Class B, BGRD selects a student to promote python and C++ in a bundle, but the two courses were usually regarded as substitutable for most students (i.e., the average substitutable relevance between python and C++ was 0.7). We observe that more than two-thirds of the students who selected python did not select C++ when they were promoted C++ by their classmates. Similar to BGRD, HAG does not examine the substitutable relationship when promoting courses. In Class B, HAG also promoted OOP and C++ to the same set of students. However, more than half of the students selected only one of OOP and C++, indicating the waste of simultaneous promotions for substitutable items. PS induces the fewest students to select the elective courses, since it does not facilitate students to promote multiple courses and cannot properly utilize the course promotion from other seeds. For example, in Class C, PS selected a student to promote deep learning (DL) to a set of students who were very interested in DL (i.e., their average initial preference for DL was 0.9). As DL and natural language processing (NLP) were regarded as highly complementary for this set of students (i.e., the average complementary relevance between DL and NLP was 0.75), a good strategy is to let the students selected by PS to promote NLP as well. However, PS did not promote any other course to this set of students in Class C.
The above results lead to conclusions consistent with the experiments in Sec.~\ref{sec:exp_comparison}, indicating that exploring the dynamic personal perceptions of item relationships, dynamic preference for items, dynamic social influence strength, and item associations is the cornerstone of influence maximization under a sequence of promotions on relevant items.

\begin{table}[t]
    \caption{\revise{The statistics of recruited classes.}}
    \label{T:class}
\centering
    \begin{tabular}{|c||c|c|c|c|c|}
        \hline
        Class ID & A & B & C & D & E \\ \hline
        \# of users & 33 & 26 & 22 & 20 & 20 \\
        \# of edges & 293 & 420 & 387 & 227 & 308 \\
        \hline
    \end{tabular}
\end{table}

\begin{figure}[t]
    \centering
    \includegraphics[width=0.30\textwidth]{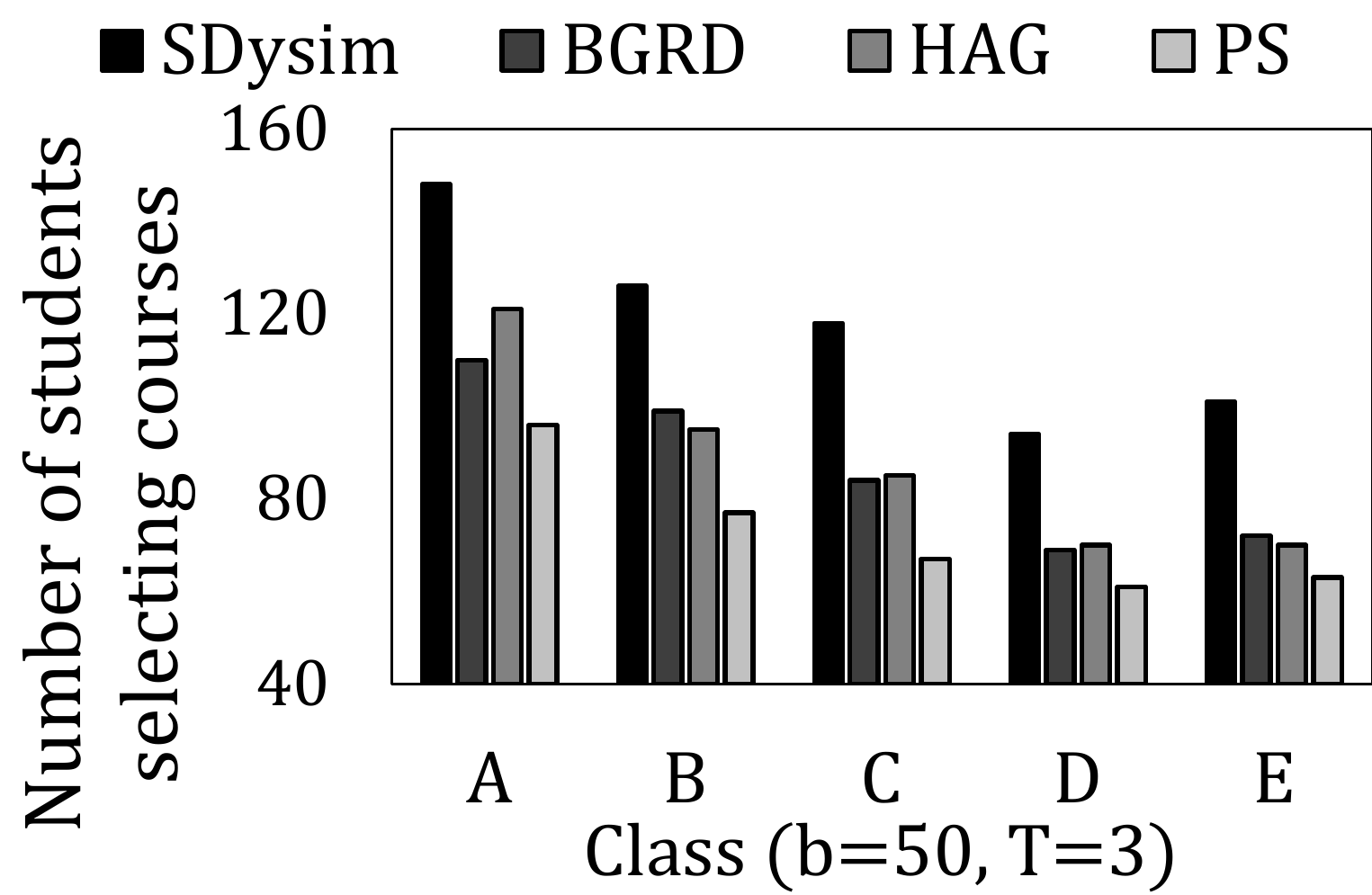}
    \caption{\revise{The empirical study: the total number of students selecting the elective courses for each classes.}}
    \label{Fig:user_study}
\end{figure}

\opt{full}{\subsection{Case Study}
From the experiments, we find several interesting cases in \textit{Amazon} and study them in detail as follows.

1) For $b=500$ and $T=20$, we find that User \#277 inclines to adopt items with greater importance via item associations. She is promoted and adopts a cinema camera Canon EOS C500 and two of its lenses in the 5th, 6th, and 9th promotions, respectively.
Then, in the 13th promotion, if she is promoted a camera Premier DS-3081S, she adopts another camera Canon EOS M via item associations (instead of Premier DS-3081S) with a very high probability. We observe that the average substitutable relevance between Premier DS-3081S and Canon EOS M is increased from 0.7 to 0.93 if she adopts the two lenses in separate promotions after adopting Canon EOS C500, since she may regard the lenses with similar functions as substitutable. Her perception of the substitutable relationship triggers the adoption of a more expensive camera Canon EOS M, as cameras are high-end items, and users usually adopt only one among the substitutable items. The changes in perceptions facilitating item associations indeed lead to larger importance-aware influence in subsequent promotions.

2) For $b=200$, User \#16900 has different purchase decisions between settings $T=1$ and $T=10$. 
We observe that her preferences for Kindle and Kindle Unlimited service are 0.61 and 0.32, respectively, before any promotion.
For $T=1$, she usually adopts only Kindle if she is promoted both items at the same time. 
For $T=10$, she is likely to be promoted and adopt Kindle in the 2nd promotion. We find that as Kindle and Kindle Unlimited service are complementary to each other, the adoption of Kindle increases her preference of Kindle Unlimited service to 0.58 on average. If she is then promoted Kindle Unlimited service in the 3rd promotion, she adopts it with a high probability. Exploiting the item relationships and the changes in personal preferences in multiple promotions is able to achieve more adoptions.

3) For $b=300$ and $T=10$, we find two different results when User \#2236 promotes some item to User \#186644. If User \#2236 attempts to promote Kindle to User \#186644 in the 2nd promotion, she probably fails. If both users adopt Garmin nuvi 50 promoted by their common friend in the 5th promotion, User \#2236 is very likely to successfully promote Kindle Voyage to User \#186644 in the 7th promotion. In the case that User \#2236 and User \#186644 both adopt Garmin nuvi 50 in the 5th promotion, they become more similar and socially closer accordingly. This makes the influence strength from User \#2236 to User \#186644 increased from 0.39 to 0.47 on average, which further helps User \#2236 promote Kindle Voyage to User \#186644 in the 7th promotion.

\subsection{Sensitivity Tests}
\label{sec:sensitivity}


\begin{figure*}[t]
    \centering
    \subfigure[\textit{Yelp}.]{
        \centering
        \includegraphics[width=0.24\textwidth]{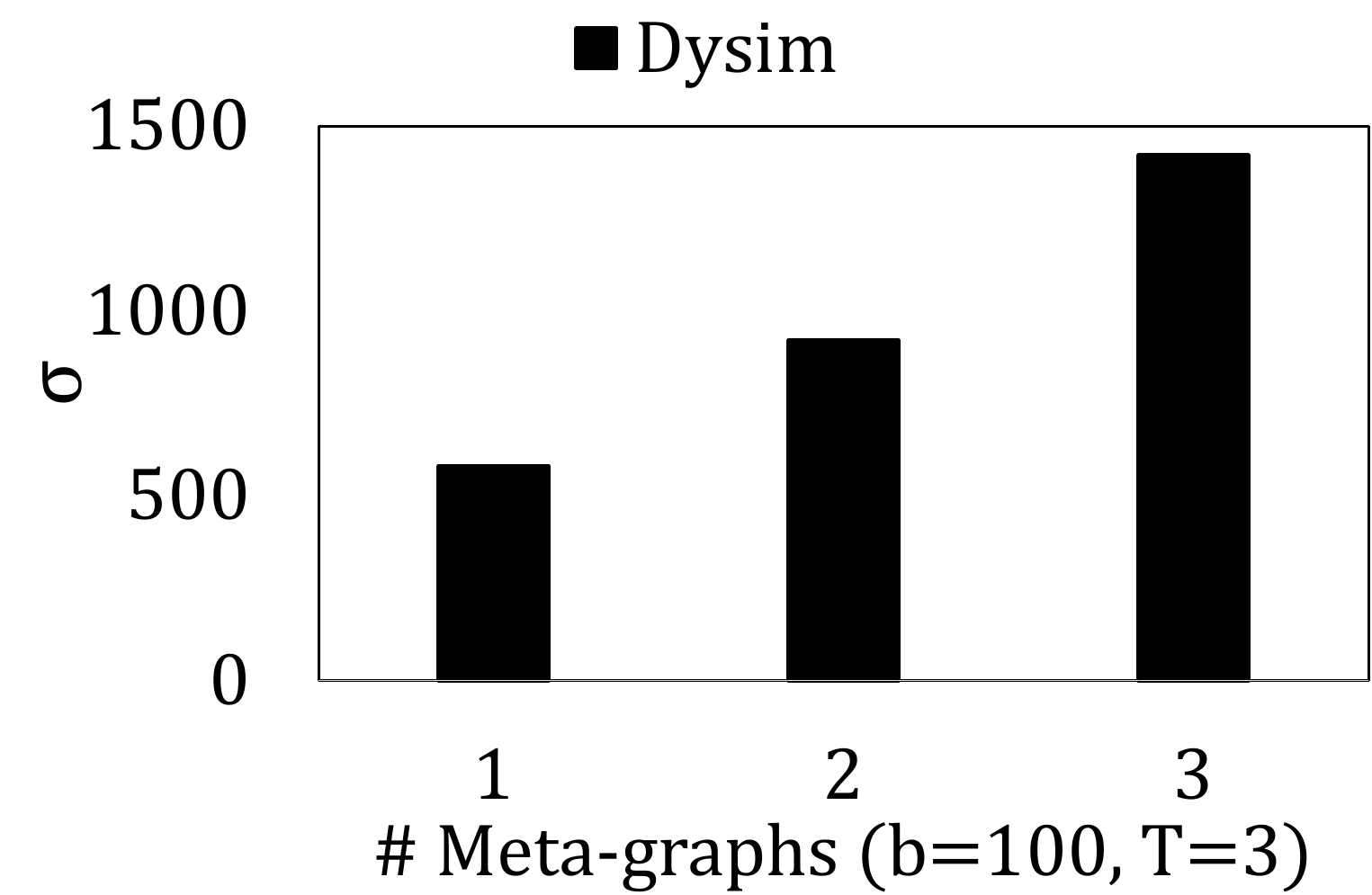}
        \label{Fig:yelp_meta}
    }%
    \subfigure[\textit{Gowalla}.]{
        \centering
        \includegraphics[width=0.24\textwidth]{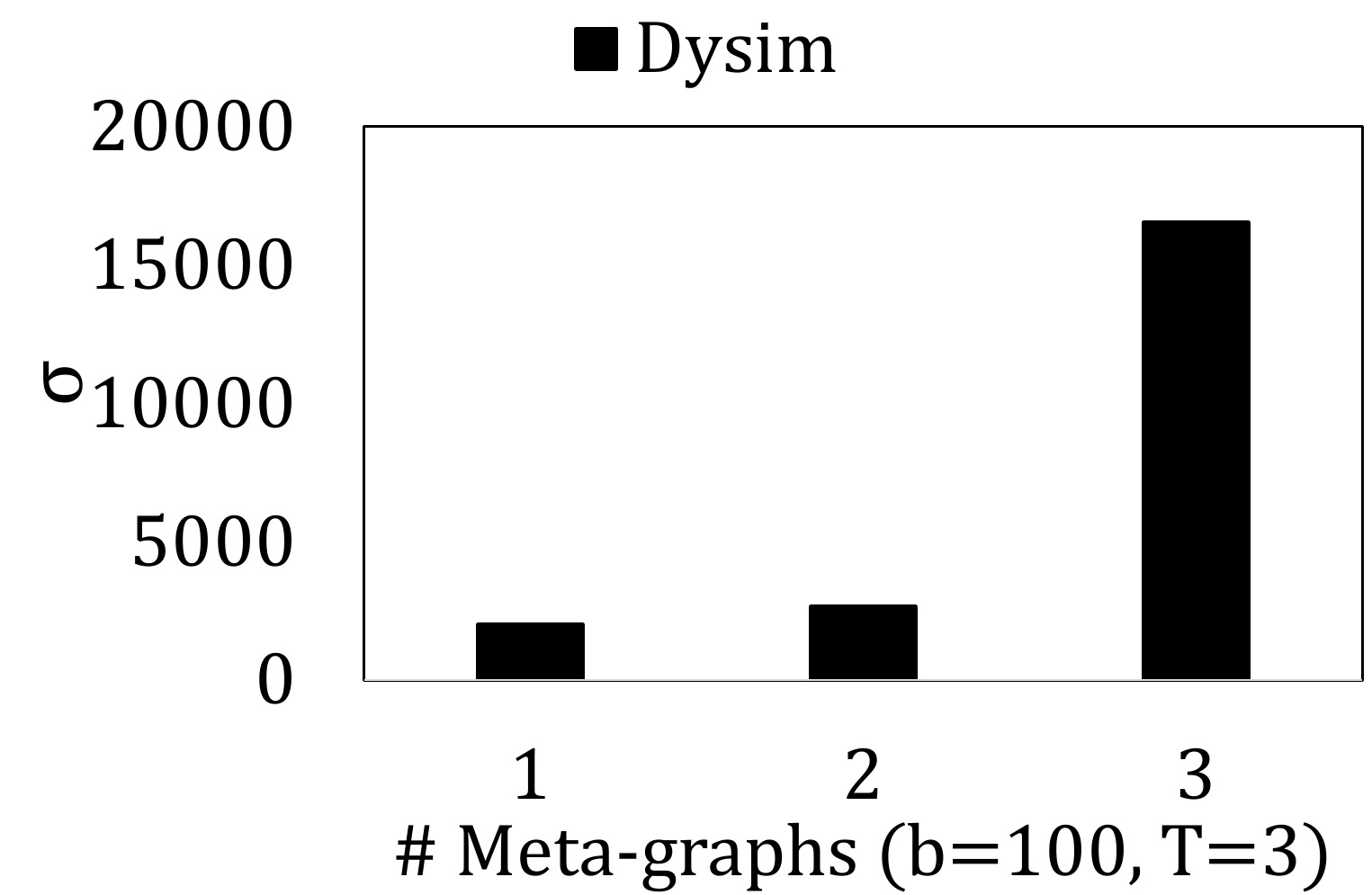}
        \label{Fig:gowalla_meta}
    }%
    \subfigure[\textit{Amazon}.]{
        \centering
        \includegraphics[width=0.24\textwidth]{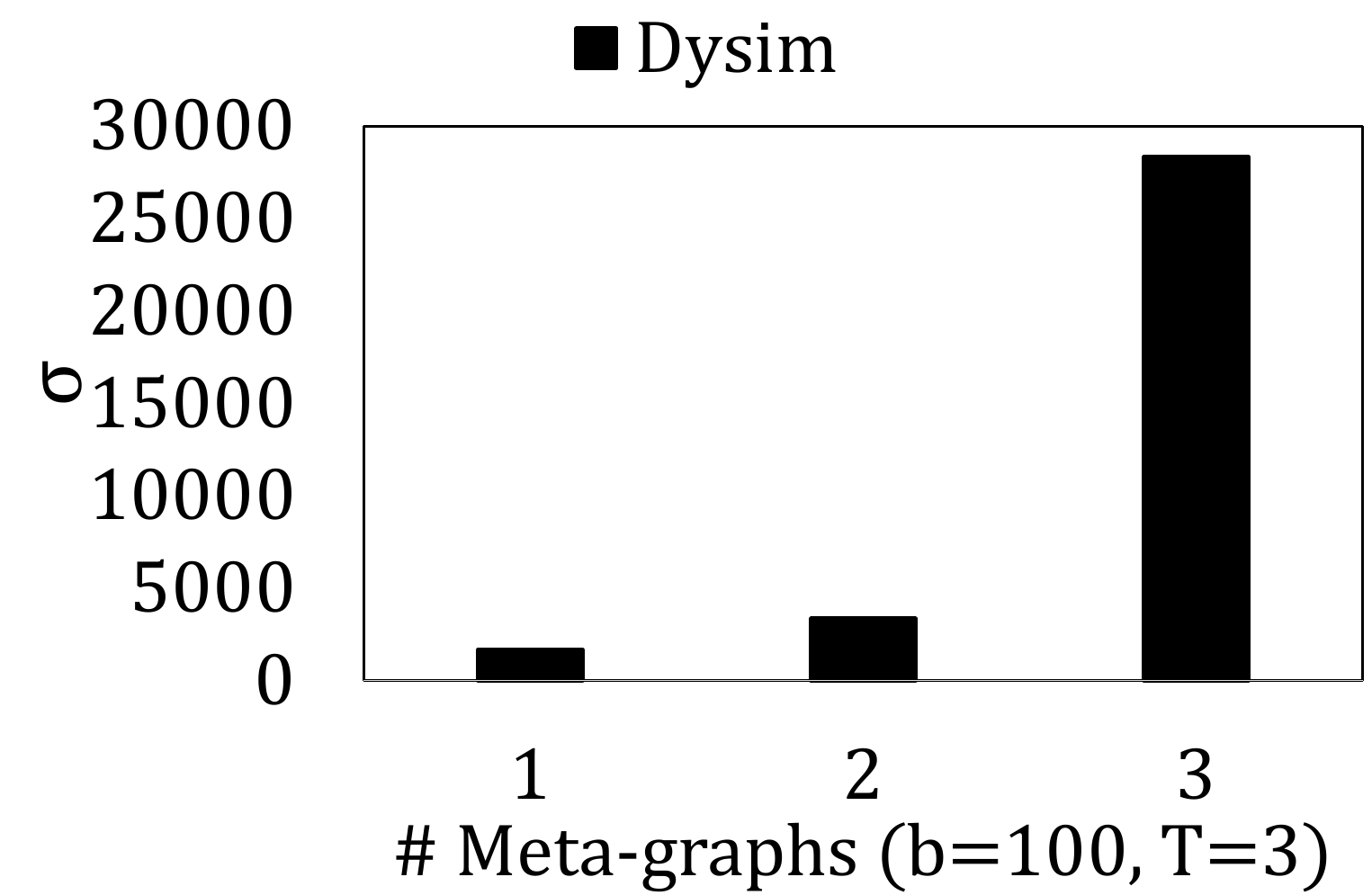}
        \label{Fig:amazon_meta}
    }%
    \subfigure[\textit{Douban}.]{
        \centering
        \includegraphics[width=0.24\textwidth]{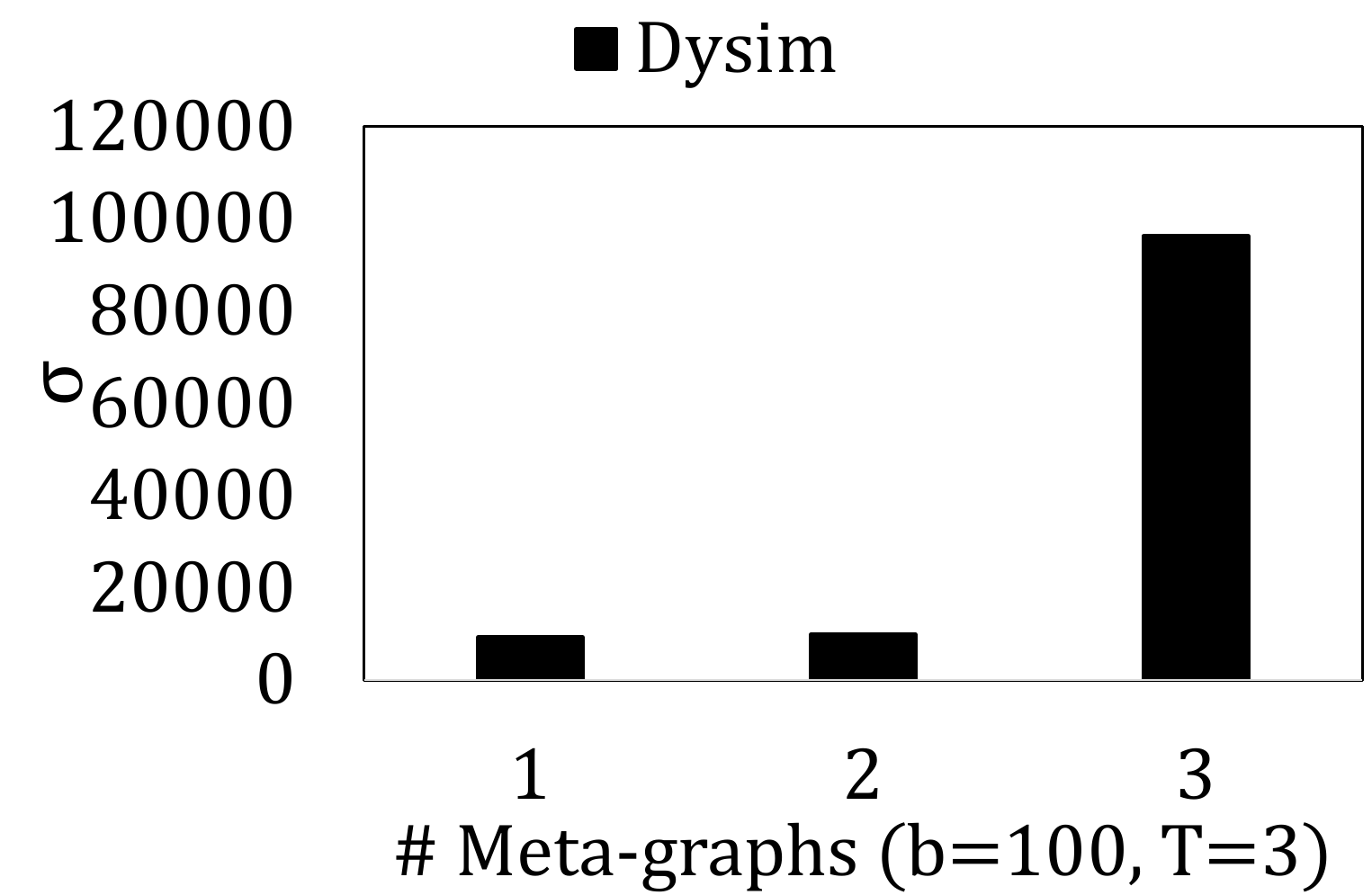}
        \label{Fig:douban_meta}
    }\hfill
    \caption{Sensitivity tests for the number of meta-graphs.}
    \label{Fig:large_meta}
\end{figure*}

\begin{figure*}[t]
    \centering
    \subfigure[\textit{Yelp}.]{
        \centering
        \includegraphics[width=0.24\textwidth]{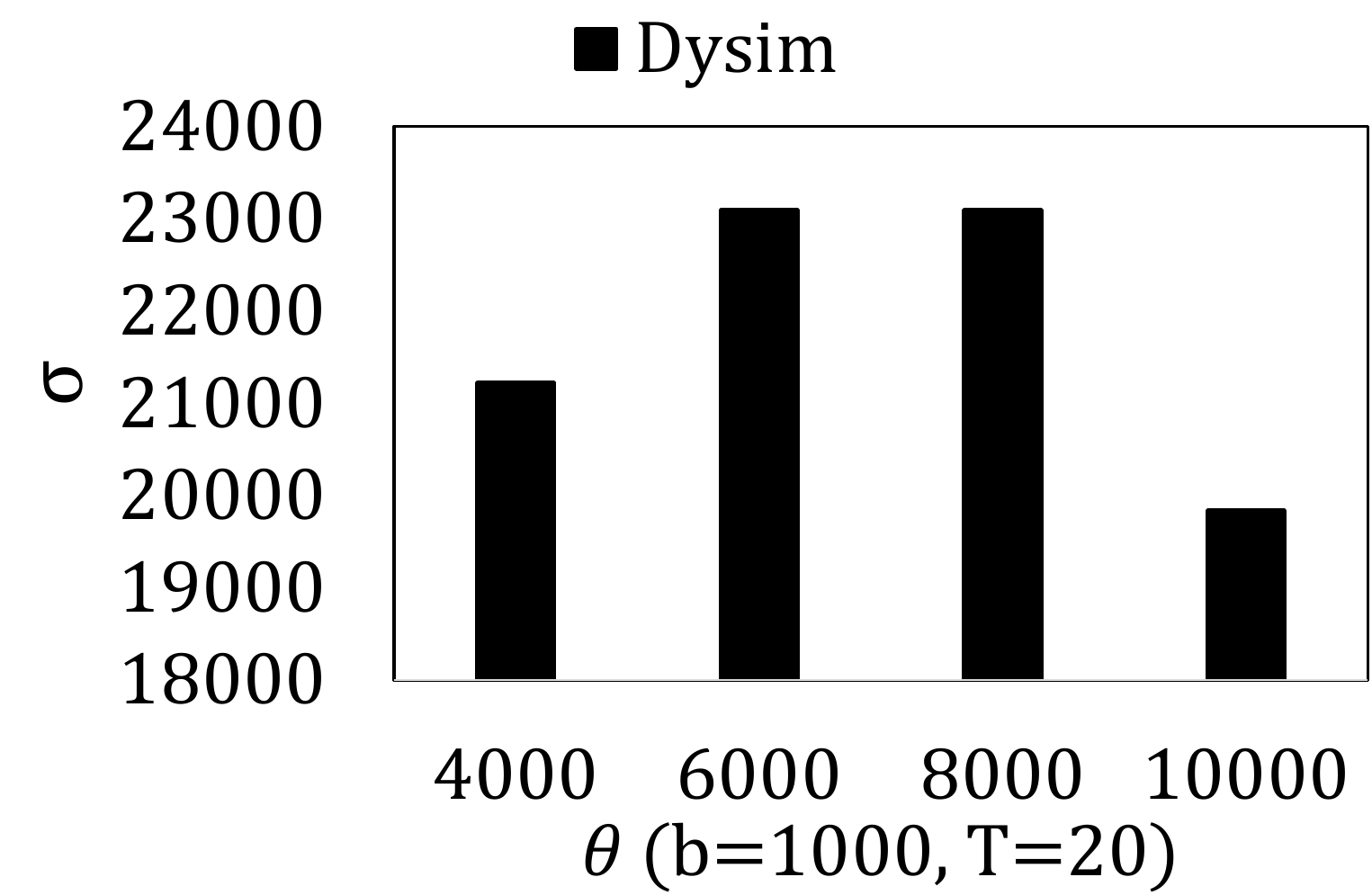}
        \label{Fig:yelp_theta}
    }%
    \subfigure[\textit{Gowalla}.]{
        \centering
        \includegraphics[width=0.24\textwidth]{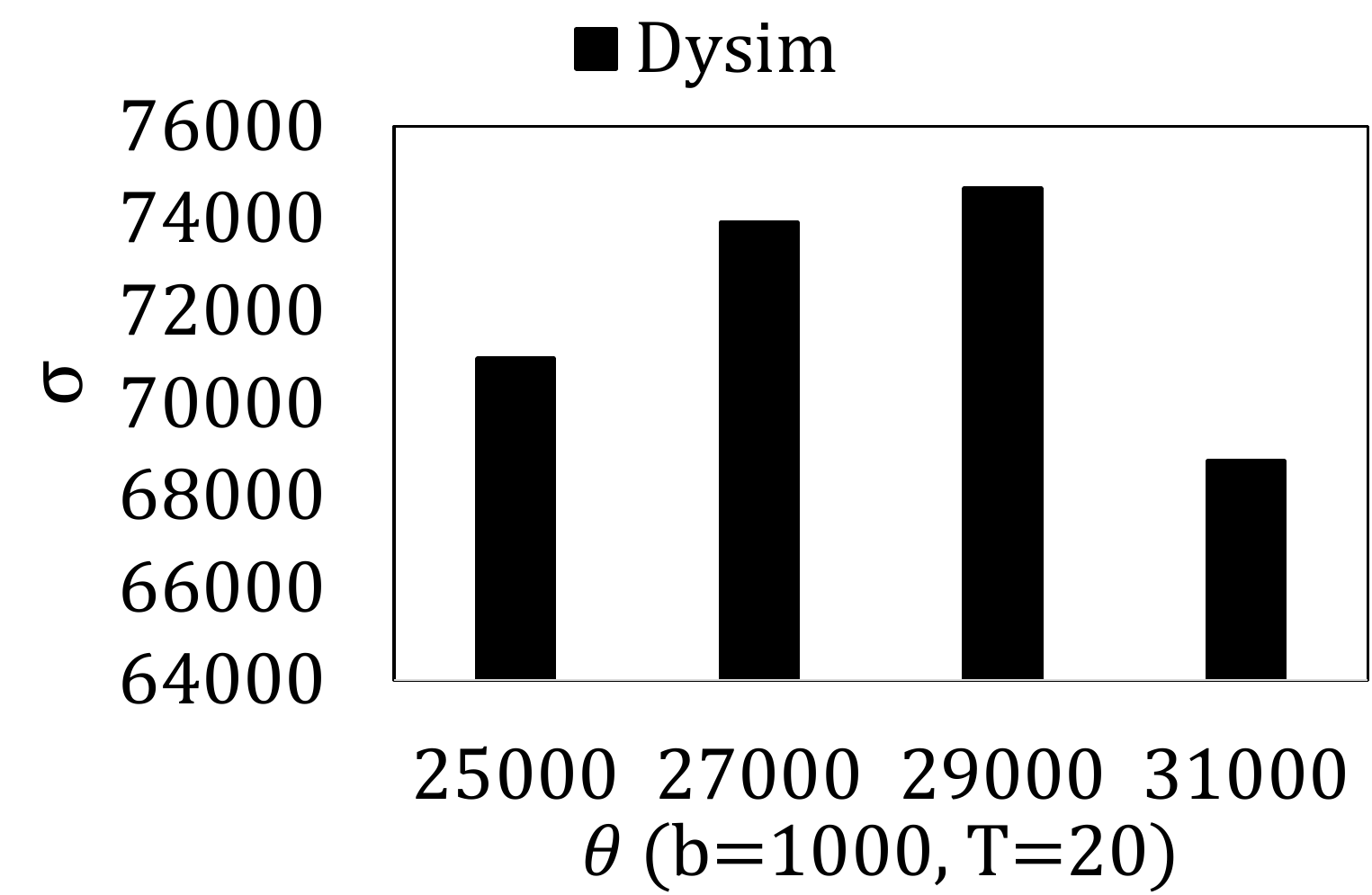}
        \label{Fig:gowalla_theta}
    }%
    \subfigure[\textit{Amazon}.]{
        \centering
        \includegraphics[width=0.24\textwidth]{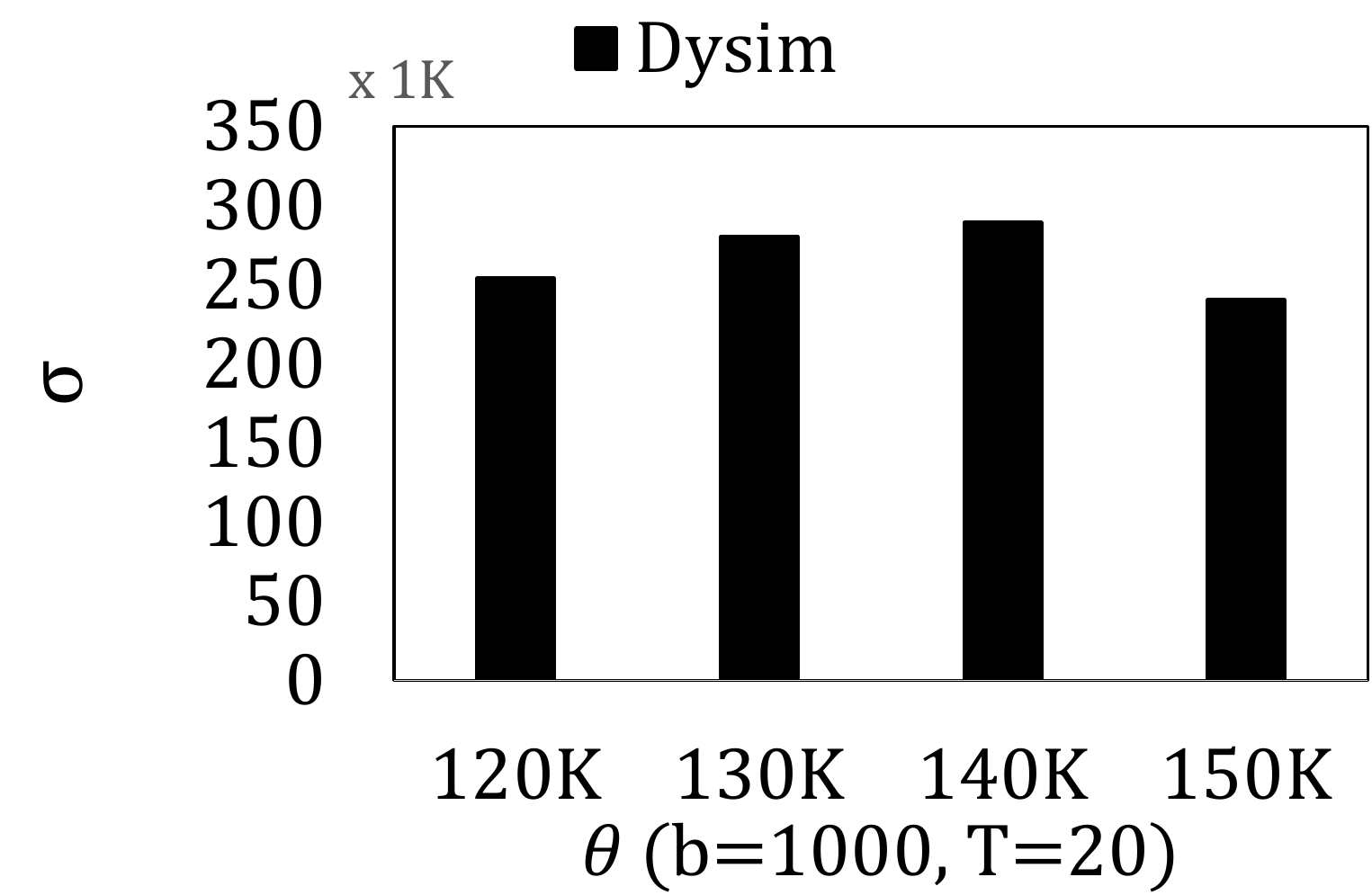}
        \label{Fig:amazon_theta}
    }%
    \subfigure[\textit{Douban}.]{
        \centering
        \includegraphics[width=0.24\textwidth]{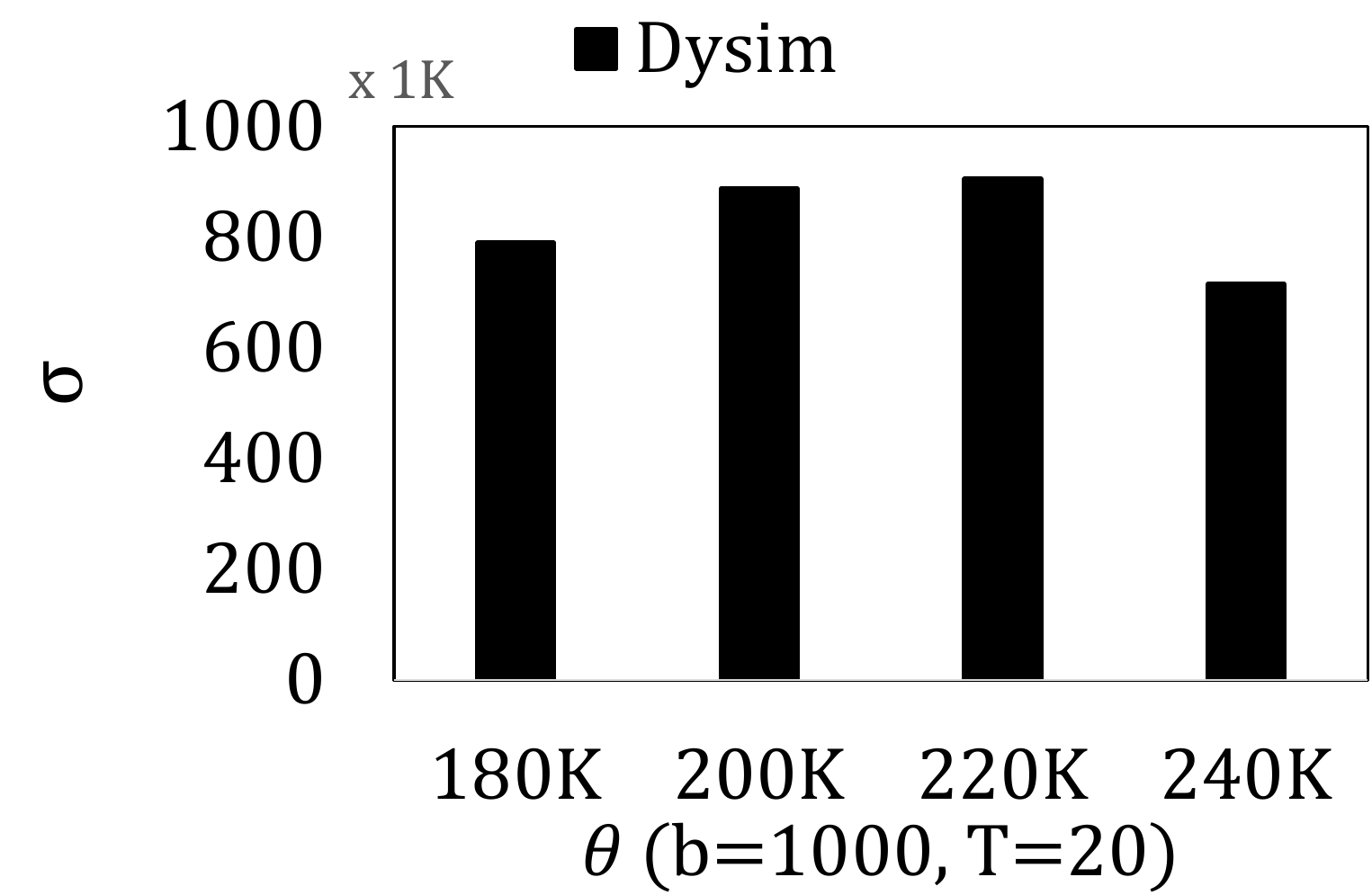}
        \label{Fig:douban_theta}
    }\hfill
    \caption{Sensitivity tests for $\theta$ in TMI.}
    \label{Fig:large_theta}
\end{figure*}

Fig.~\ref{Fig:large_meta} compares the importance-aware influence under different numbers of meta-graphs. With more meta-graphs, \algo\ achieves larger importance-aware influence by better capturing users' perceptions, demonstrating the importance of modeling item relationships via KG and meta-graphs for a sequence of promotions.
Fig.~\ref{Fig:large_theta} shows the sensitivity to the threshold $\theta$ for identifying target markets with common users in TMI. A large $\theta$ slightly reduces the importance-aware influence since target markets may promote substitutable items to their common users. Nevertheless, a small $\theta$ also slightly deteriorates the performance because the promotional duration of a target market may be insufficient (when there are too many target markets in the same $\mathcal{G}$) to foster the promotion of complementary items properly. 
}

\section{Conclusion}

To the best of our knowledge, this paper makes the first attempt to study the problem of influence maximization under a sequence of promotions for multiple relevant items. By exploring KG and meta-graphs to capture dynamic personal perceptions of item relationships, we formulate a new problem, named \problem, \opt{short}{and \revise{its fundamental problem, namely \sproblem,} }to choose items and select seed users for promotions at proper timings. We prove the hardness of \revise{\opt{short}{\sproblem\ and \problem}\opt{full}{\problem}} and design an approximation algorithm \algo\ to solve \problem. \algo\ first identifies nominees and target markets to promote complementary items to socially close users in consecutive promotions. For each target market, \algo\ prioritizes the items to be promoted by dynamic reachability of items. Then, \algo\ determines proper promotional timings with the highest substantial influence for each nominee. \opt{short}{\revise{We also design an approximation algorithm \salgo\ for \sproblem.}}
Experiments on real social networks and KGs demonstrate that \revise{\opt{short}{\salgo\ and }\algo\ can effectively achieve up to \opt{short}{10.95 times and }6.7 times\opt{short}{, respectively,} of the influence spread.} \revise{Furthermore, the empirical study validates that exploring the dynamic personal perceptions of item relationships, dynamic preference for items, dynamic social influence strength, and item associations is crucial for influence maximization under a sequence of promotions on relevant items.}

\section*{Acknowledgement}
We thank to National Center for High-performance Computing (NCHC) of National Applied Research Laboratories (NARLabs) in Taiwan for providing computational and storage resources.

\bibliographystyle{IEEEtran}
\bibliography{IEEEabrv,bibliographies/ref}


\end{document}